\documentclass{caltechcc}

\usepackage{setspace}
\usepackage{array, pifont}
\usepackage{makecell, multirow}
\usepackage{xcolor,colortbl}
\usepackage{afterpage}
\usepackage{pbox}
\allowdisplaybreaks

\newcommand{\cmark}{\ding{51}}%
\newcommand{\xmark}{\ding{55}}%
\newcommand{\halfcmark}{\cmark\textsuperscript{\kern-0.6em\raisebox{-0.48ex}{\xmark}}}

\newcommand{\expnumber}[2]{{#1}\mathrm{e}{#2}}
\newcommand{\first}[1]{$\mathbf{#1}$}
\newcommand{\second}[1]{\underline{#1}}

\renewcommand{\algorithmiccomment}[1]{\bgroup\hfill{$\triangleright$~#1}\egroup}
\newcommand{\comment}[1]{\text{\small(#1)}}

\newcolumntype{L}[1]{>{\raggedright\arraybackslash}p{#1}}
\newcolumntype{C}[1]{>{\centering\arraybackslash}p{#1}}
\newcolumntype{R}[1]{>{\raggedleft\arraybackslash}p{#1}}

\def\epsilon{\varepsilon}
\def\Id{Id}

\def\shortneg{\text{-}}

\def\dXt{{d\xbm_t}}
\def\dBt{{d\Bbm_t}}
\def\Xt{{\xbm_t}}

\def\Bt{{\Bbm_t}}

\def\nut{{\nu_t}}
\def\nuz{{\nu_0}}
\def\nuj{{\nu_{t,0}}}
\def\nutz{{\nu_{t|0}}}
\def\nuzt{{\nu_{0|t}}}

\def\Xk{{\xbm_k}}
\def\Xkga{{\xbm_{k\gamma}}}
\def\Zk{{\Zbm_k}}

\def\Bkga{{\Bbm_{k\gamma}}}
\def\Xkpo{{\xbm_{k+1}}}

\def\Xzero{{\xbm_{0}}}

\def\sigmak{{\sigma_k}}
\def\alphak{{\alpha_k}}
\def\psigma{p_\sigma}
\def\psigmak{p_\sigmak}
\def\Pcalk{\Pcal_k}
\def\Gcalk{\Gcal_k}

\def\LG{L_\mathcal{G}}
\def\LGk{L_{\mathcal{G}_k}}
\def\LP{L_\mathcal{P}}
\def\LPk{L_{\mathcal{P}_k}}

\def\Lf{L_{f}}
\def\Ltilde{\tilde{L}}
\def\Ltildek{\tilde{L}_k}
\def\Lmax{L_{\mathsf{\tiny max}}}
\def\Ltildemax{\tilde{L}_{\mathsf{\tiny max}}}

\def\sigmabar{\bar{\sigma}}

\def\epsilonbar{\bar{\epsilon}}

\def\anneal{\textsc{APMC}}

\definecolor{limegreen}{rgb}{0.2, 0.8, 0.2}

\title{Provable Probabilistic Imaging using Score-based Generative Priors}

\author
{Yu Sun \Envelope, Zihui Wu, Yifan Chen, Berthy Feng, and Katherine L. Bouman\\
\vspace{1em}
\normalfont{
\small Department of Computational and Mathematical Sciences \\
\small California Institute of Technology \\
\Envelope~Corresponding author: sunyu@caltech.edu
}\\
\vspace{2em}
}

\headertitle{Plug-and-play Monte Carlo}
\headerauthors{Sun et al.}

\begin{document}

\date{}
\maketitle
\thispagestyle{firstpagestyle}

\begin{abstract}
Estimating high-quality images while also quantifying their uncertainty are two desired features in an image reconstruction algorithm for solving ill-posed inverse problems.
In this paper, we propose \emph{plug-and-play Monte Carlo (PMC)} as a principled framework for characterizing the space of possible solutions to a general inverse problem.
PMC is able to incorporate expressive score-based generative priors for high-quality image reconstruction while also performing uncertainty quantification via posterior sampling.
In particular, we develop two PMC algorithms that can be viewed as the sampling analogues of the traditional plug-and-play priors (PnP) and regularization by denoising (RED) algorithms.
To improve the sampling efficiency, we introduce weighted annealing into these PMC algorithms, further developing two additional annealed PMC algorithms (APMC).
We establish a theoretical analysis for characterizing the convergence behavior of PMC algorithms.
Our analysis provides non-asymptotic stationarity guarantees in terms of the Fisher information, fully compatible with the joint presence of weighted annealing, potentially non-log-concave likelihoods, and imperfect score networks.
We demonstrate the performance of the PMC algorithms on multiple representative inverse problems with both linear and nonlinear forward models.
Experimental results show that PMC significantly improves reconstruction quality and enables high-fidelity uncertainty quantification.
\end{abstract}

\section{Introduction}
\label{Sec:InvPro}

The problem of accurately reconstructing high-quality images $\xbm\in\R^n$ from a set of sparse and noisy measurements $\ybm\in\R^m$ is fundamental in computational imaging.
These measurements often do not contain sufficient information to losslessly specify the target image, making the inverse problem \emph{ill-posed}. Image reconstruction methods should account for ill-posedness in two ways: by \emph{1)} imposing \emph{prior} knowledge to regularize the final solution image, and by \emph{2)} characterizing the \emph{posterior} probability distribution of all possible solutions.
While much of the computational imaging research has centered on exploiting prior knowledge, it remains essential to characterize the full distribution of solutions to understand uncertainty, regardless of any imposed constraints.
Traditionally, either the posterior distribution is ignored or it is derived under simplified image priors to make the problem tractable, leading to biased results. 
The objective of this work is to leverage recent advances in learning-based generative models in order to provably specify the posterior distribution under expressive image priors.

The Bayesian framework is commonly used to infer the posterior distribution $\pi(\xbm|\ybm)$ from the prior distribution $p(\xbm)$ of the desired image by incorporating the \emph{likelihood} $\ell(\ybm|\xbm)$
\begin{equation}
\label{Eq:Posterior}
\pi(\xbm|\ybm) \propto \ell(\ybm|\xbm) p(\xbm).
\end{equation}
Here, $\ell(\ybm|\xbm)$ probabilistically relates the image to the measurements and is defined by the known imaging system.
The quality of the inference of $\xbm|\ybm$ directly depends on the expressivity of the prior.
Classic choices include the sparsity-promoting priors such as total variation (TV)~\cite{Rudin.etal1992, Beck.Teboulle2009a}. However, these hand-crafted priors are not expressive enough to characterize complicated image structures. 
The focus has recently shifted to exploring \emph{learning-based} priors due to their strong expressive power; such priors are often parameterized by deep neural networks.
Plug-and-play priors (PnP)~\cite{Venkatakrishnan.etal2013, Sreehari.etal2016} and regularization by denoising (RED)~\cite{Romano.etal2017} are two popular methods in this direction.
They consider a pre-trained denoiser as an implicit image prior and use it to replace the functional prior module in the traditional iterative algorithms.
Despite the practical success under deep denoisers~\cite{Wu.etal2020,Ahmad.etal2020,zhang2020plug}, PnP/RED methods are essentially based on the \emph{maximum a posteriori (MAP)} estimation which only produces a single image and cannot account for the full posterior distribution.

Sampling from the posterior distribution offers a principled approach to gain insight into the uncertainty and credibility of a reconstructed image.
Such \emph{uncertainty quantification (UQ)} is especially crucial in nonlinear inverse problems where multiple distinct solutions could exist and lead to very different interpretations.
Traditional UQ methods, which rely on simplified image assumptions, can perform reasonably well for certain types of images~\cite{Bardsley2012,Repetti.etal2019}. 
However, they often yield limited performance when dealing with images containing complex structures.
Recently, \emph{score-based generative models (SGM)} have emerged as a powerful deep learning (DL) tool for sampling from complex high-dimensional distributions. In particular, SGMs learn the \emph{score} of an image distribution and use it in a \emph{Markov chain Monte Carlo (MCMC)} algorithm to perform iterative sampling.
It has been shown that SGMs achieve state-of-the-art performance in unconditional image generation~\cite{Yang.etal2023diffusion}.
Recently, several works have explored using different SGMs for posterior sampling~\cite{kawar2021snips, Jalal.etal2021, Song2022solving, Chung.etal2023diffusion}, demonstrating promising results.
Unfortunately, these methods encountered a trade-off between the applicability to nonlinear systems, theoretical guarantees on convergence, and accelerated sampling procedure; pursuing any combination of two of these goals would lead to limited performance in the remaining one.
This indicates that leveraging a SGM to conduct \emph{provable} and \emph{efficient} posterior sampling subject to \emph{sophisticated physical constraints} is underexplored.

In this paper, we aim to bridge the gap by proposing \emph{plug-and-play Monte Carlo (PMC)} as a principled posterior sampling framework for solving inverse problems.
PMC is built on the fusion of PnP/RED and SGM: it leverages powerful score-based generative priors in a plug-and-play fashion, similar to PnP/RED, while also enabling provable posterior sampling by incorporating the MCMC formulation used in the SGM.
Specifically, the key contributions of this work are as follows:
\begin{itemize}

\item We demonstrate how PnP and RED can be generalized to sampling by studying their continuous limits.
We show that PnP and RED converge to the same \emph{gradient-flow ordinary differential equation (ODE)} of the posterior, which corresponds to the noise-free version of the \emph{Langevin stochastic differential equation (SDE)}.
Note that this observation also leads to new insights about the mathematical equivalence between PnP and RED.

\item We develop two PMC algorithms by discretizing the Langevin SDE according to the formulations of PnP and RED.
We name our algorithms \emph{PMC-PnP} and \emph{PMC-RED}, respectively.
To the best of our knowledge, PMC-PnP represents a new addition to the existing sampling algorithms, while PMC-RED aligns with the plug-and-play unadjusted Langevin algorithm~\cite{Laumont.etal2022}.
Inspired by recent advances in SGM, we employ \emph{weighted annealing} for PMC algorithms to improve sampling efficiency.
Experimental results show that it can accelerate the sampling speed and facilitate the exploration of multiple modes more efficiently.

\item We present a comprehensive theoretical analysis of all the PMC algorithms by leveraging the optimization interpretation of the Langevin SDE. 
Our analysis establishes the stationary-distribution convergence in terms of Fisher information, with an explicit rate of $O(1/N)$.
This result is in direct alignment with the fixed-point analysis for PnP/RED in optimization.
Additionally, our analysis is fully compatible with the joint presence of potentially \emph{non-log-concave} likelihoods, \emph{imperfect} score networks, and \emph{weighted annealing}, providing a theoretical characterization of the algorithmic behavior in these scenarios.

\item We validate our algorithms on three real-world imaging inverse problems: compressed sensing, magnetic resonance imaging, and black-hole interferometric imaging.
Note that the black-hole imaging problem corresponds to a non-log-concave likelihood.
Our experimental results demonstrate improved reconstruction quality and reliable UQ enabled by the proposed PMC algorithms.

\end{itemize}

\section{Background}
\label{Sec:Background}

\subsection{Solving inverse problems with Bayesian inference}
\label{Sec:InverseProblem}
Consider the general inverse problem
\begin{equation}
\label{Eq:Inverse}
\ybm = \Abm(\xbm)+ \ebm,
\end{equation}
where the goal is to recover $\xbm \in \R^n$ given the measurements $\ybm \in \R^m$. Here, the measurement operator $\Abm: \R^n \rightarrow \R^m$ models the response of the imaging system, and $\ebm \in \R^m$ represents the measurement noise, which is often assumed to be additive white Gaussian noise (AWGN).
One popular Bayesian inference framework for imaging inverse problems is based on the MAP estimation:
\begin{equation}
\label{Eq:MAP}
\xbmhat = \argmax_{\xbm\in\R^n} \ell(\ybm|\xbm) p(\xbm) 
= \argmin_{\xbm\in\R^n} \big\{ g(\xbm) + h(\xbm) \big\}.
\end{equation}
Here, $g(\xbm) = -\log \ell (\ybm|\xbm)$ is often known as the data-fidelity term and $h(\xbm) = -\log p(\xbm)$ as the regularizer\footnotemark.
Commonly, the data-fidelity is set to the least-square loss 
$g(\xbm) = \frac{1}{2\beta^2}\|\ybm - \Abm(\xbm)\|_2^2$, which corresponds to the Gaussian likelihood $\ybm|\xbm \sim \Ncal(\Abm(\xbm), \beta^2I)$ with $\beta^2>0$ controlling the variance.

\footnotetext{While the words \emph{`regularizer'} and \emph{`prior'} are often interchangeable in the literature, we here use \emph{`regularizer'} to explicitly denote function $h$ in~\eqref{Eq:MAP}.}

Many popular regularizers, including those based on sparsity, are not differentiable. This precludes using simple algorithms such as gradient descent to solve~\eqref{Eq:MAP}.
Proximal methods~\cite{Boyd.etal2011,Beck.Teboulle2009} are commonly employed to accommodate nonsmooth regularizers by leveraging a mathematical concept known as the \emph{proximal operator}
\begin{equation}
\label{Eq:ProximalOperator}
\prox_{\mu h}(\zbm) \defn \argmin_{\xbm \in \R^n} \left\{\frac{1}{2}\|\xbm-\zbm\|_2^2 + \mu h(\xbm)\right\}.
\end{equation}
Here, the squared error enforces the output $\xbm$ to be close to the input $\zbm$, while $h$ penalizes the solutions falling outside the constraint set with $\mu>0$ adjusting the strength.
For many regularizers, the sub-optimization in~\eqref{Eq:ProximalOperator} can be efficiently solved without differentiating $h$~\cite{Beck.Teboulle2009,Beck.Teboulle2009a}.
One commonly used proximal method is the \emph{proximal gradient method}, which pairs the proximal operator with gradient descent
\begin{align}
\label{Eq:PGM}
\Xkpo &= \prox_{\gamma h}\big(\Xk - \gamma \nabla g(\Xk) \big),
\end{align}
where $\mu$ is usually set to $\gamma$ for ensuring convergence to the minimizer of \eqref{Eq:MAP}.
We highlight that the prior information is only imposed via $\prox_{\gamma h}$. 
Since the squared loss corresponds to Gaussian likelihood in denoising problems involving AWGN, a comparison between~\eqref{Eq:MAP} and~\eqref{Eq:ProximalOperator} indicates that the proximal operator can be interpreted as a MAP denoiser for AWGN by letting $h=-\log p(\xbm)$.

\begin{table*}[t]
\newcolumntype{D}{>{\centering\arraybackslash}p{30pt}}
\newcolumntype{E}{>{\centering\arraybackslash}p{50pt}}
\newcolumntype{F}{>{\centering\arraybackslash}p{50pt}}
\newcolumntype{G}{>{\centering\arraybackslash}p{50pt}}
\newcolumntype{H}{>{\centering\arraybackslash}p{80pt}}
\newcolumntype{I}{>{\centering\arraybackslash}p{50pt}}
    \centering
    \scriptsize
    \caption{A conceptual comparison of existing posterior sampling methods for probabilistic imaging. Note that the \emph{``Annealing"} column is used for explaining the differences within the MCMC-based methods.}
    \begin{tabularx}{0.93\linewidth}{DEFGHHI}
        \toprule
        \textbf{Category} & \textbf{Reference} & \textbf{Generative prior} & \textbf{Model agnostic} & \textbf{Type of $A(\cdot)$} & \textbf{Convergence guarantees} & \textbf{Annealing} \\
        \midrule
        \multirow{2}{*}{\pbox{2cm}{Variational \\ Bayesian}} & \cite{SunHe.2021} & \xmark & \xmark & General & \xmark & - \\
         & \cite{Feng.etal2023scorebased} & \cmark & \xmark & General & \xmark & - \\ 
        \noalign{\vskip 0.6ex}\hdashline\noalign{\vskip 0.6ex}
        \multirow{3}{*}{\pbox{2cm}{DM-based}} & \cite{Song2022solving} & \cmark & \cmark & Linear & \xmark & - \\
         & \cite{Chung.etal2023diffusion} & \cmark & \cmark & General & \xmark & - \\
         & \cite{Liu.etal2023} & \cmark & \xmark & Linear & \xmark & - \\
         \noalign{\vskip 0.6ex}\hdashline\noalign{\vskip 0.6ex}
         \multirow{5}{*}{\pbox{2cm}{MCMC-based}} & \cite{Jalal.etal2021} & \cmark & \cmark & Linear & \cmark$^1$ & \cmark  \\
         & \cite{kawar2021snips} & \cmark & \cmark & Linear & \xmark & \cmark \\
         & \cite{Laumont.etal2022} & \xmark & \cmark & General & \cmark & \xmark \\
         & \cite{Coeurdoux.etal2023} & \cmark & \cmark & Linear & \xmark & \xmark \\
         & \cite{Bouman.etal2023generative} & \xmark & \cmark & Linear & \cmark$^2$ & \xmark \\  
         \noalign{\vskip 0.6ex}\hline\noalign{\vskip 0.6ex}
         \pbox{2cm}{MCMC-based} & \textbf{Ours} & \cmark & \cmark & General & \cmark & \cmark \\
        \bottomrule
    \end{tabularx}
    \label{Tab:Comparison}\\
    \vspace{0.1cm}
    \begin{minipage}{0.9\linewidth}
        $^1$Requires $\Abm(\cdot)$ to be Gaussian random matrix. \quad $^2$Gaurantees on asymptotic convergence. \\
    \end{minipage}
    \vspace{-0.1cm}
    \vspace{-10pt}
\end{table*}

\subsection{Using denoisers to impose image priors}
Inspired by the connection between proximal operator and MAP denoiser, PnP~\cite{Venkatakrishnan.etal2013, Sreehari.etal2016} emerged as a pioneering approach for imposing expressive yet implicit image priors through denoising. 
In this method, the $\prox_{\gamma h}$ step is replaced with an general image denoiser $\Dcal_\sigma:\R^n\rightarrow\R^n$ with $\sigma>0$ controlling the denoising strength.
For example, the formulation of PnP proximal gradient method~\cite{Kamilov.etal2017} is given by
\begin{align}
\label{Eq:PnP}
\Xkpo &= \Dcal_\sigma\big(\Xk - \gamma \nabla g(\Xk) \big)
\end{align}
where $\gamma>0$ denotes the step-size. 
In order to be backward compatible with~\eqref{Eq:PGM}, the strength parameter is often scaled with the step-size as $\sigma = \sqrt{\gamma\lambda}$ for some $\lambda>0$.
However, the denoiser may not correspond to any explicit regularizer $h$, making PnP generally lose the interpretation as proximal optimization.
Other PnP algorithms have also been proposed based on different formulations~\cite{Metzler.etal2016, Buzzard.etal2017, Sun.etal2019a, Sun.etal2020}, and we refer to~\cite{Kamilov.etal2023} for a comprehensive review.

RED~\cite{Romano.etal2017} is an alternative method for leveraging image denoiser within MAP optimization algorithms.
Different from PnP, RED does not originate from the proximal-denoiser connection. Instead, it uses the noise residual to approximate the gradient of an implicit regularizer.
One widely-used example is the RED gradient descent~\cite{Romano.etal2017}
\begin{equation}
\label{Eq:RED}
\Xkpo = \Xk - \gamma \Big( \nabla g(\Xk) + \tau \big( \Xk-\Dcal_\sigma(\Xk) \big) \Big),
\end{equation}
where $\tau>0$ is the regularization parameter. 
In some special cases, RED is able to link the noise residual $\xbm-\Dcal_\sigma(\xbm)$ to an explicit regularizer $h(\xbm) = \frac{1}{2}\xbm^\Tsf(\xbm-\Dcal_\sigma(\xbm))$~\cite{Romano.etal2017, Reehorst.Schniter2019}; however, such RED regularizers do not exist for general denoisers, including those parameterized by deep neural networks.
On the other hand, RED can be connected with the Moreau envelope~\cite{Sun.etal2019c, Laumont.etal2023}, which allows RED to be interpreted as a generalization of proximal optimization.

The success of PnP and RED has motivated theoretical studies to understand their convergence under general denoisers.
Due to the nonexistence of explicit regularizers, PnP and RED are commonly interpreted as fixed-point iterations, and their convergence to a fixed point has been shown for various algorithms~\cite{Metzler.etal2016, Chan.etal2016, Buzzard.etal2017, Reehorst.Schniter2019, Sun.etal2019a, Ryu.etal2019, Sun.etal2020, Gavaskar.Chaudhury2020, Xu.etal2020, Nair.etal2021, Cohen.etal2021, Wu.etal2020, Sun.etal2021, Hu.etal2022, Kamilov.etal2023}.
As such a fixed-point analysis often relies on the denoiser to be (firmly) nonexpansive, prior works have studied how to impose such conditions in training DL-based denoisers~\cite{Ryu.etal2019, Sun.etal2019c, Pesquet.etal2021}.
Considerable interest has also been devoted to investigating the mathematical equivalence between PnP and RED~\cite{Reehorst.Schniter2019, Cohen.etal2021, Liu.2021}.
Nevertheless, a direct algorithmic equivalence between \eqref{Eq:PnP} and \eqref{Eq:RED} is still absent in the literature.

\subsection{Score-based generative models (SGM)}
\label{Sec:SGM}
SGMs have been recently proposed as a powerful generative method for drawing samples from a complex high-dimensional distribution. 
The main idea of SGMs is to use the score $\nabla \log p(\xbm)$ of the desired distribution $p$ in either MCMC algorithms~\cite{Song.etal2019} or diffusion models (DM)~\cite{Ho.etal2020,Song.etal2021score}.

As the true score is almost impossible to obtain for complex distributions $p$, a SGM resorts to finding a computable approximation of $\nabla \log p$ in the form of a neural network; this technique is often known as \emph{score matching}~\cite{Hyvarien2005, Vincent.etal2011, Alain.etal2014}.
One elegant approximation is inspired by Tweedie's formula~\cite{Efron2011}.
Let $\zbm \sim p$ be a random sample from the desired distribution, $\ebm\sim\Ncal(0,\sigma^2 I)$ the AWGN, and $\xbm = \zbm + \ebm$ the noisy image. 
Then, we have
\begin{equation}
\label{Eq:Tweedie}
\nabla \log p_\sigma(\xbm) = \frac{\E[\zbm | \xbm] - \xbm}{\sigma^2},
\end{equation}
where $\E[\zbm|\xbm]$ is the \emph{minimum mean squared error (MMSE)} estimator of $\zbm$ given $\xbm$, and $\psigma$ is the distribution of the noisy image, which is distinct from the distribution of $p$ itself.
As the sum of random variables results in a distribution equal to the convolution of their individual distributions, $\psigma$ is given by 
$$\psigma(\xbm) = \int_{\R^n} p(\zbm) \phi_\sigma(\zbm-\xbm)d\zbm,$$
where $\phi_\sigma$ denotes the probability density function of $\Ncal(0,\sigma^2 I)$.
The distribution $\psigma$ can be intuitively interpreted as a smoothed version of the true distribution.

An extension of~\eqref{Eq:Tweedie} leads to the denoising score matching (DSM) loss~\cite{Vincent.etal2011} that enables learning the score of the smoothed prior from data
\begin{equation}
\label{Eq:DSM}
\textsf{DSM}(\theta) = \E\left[ \left\|\frac{\zbm - \xbm}{\sigma^2} - \Scal_\theta(\xbm, \sigma) \right\|_2^2 \right],
\end{equation}
where $\Scal_\theta(\xbm, \sigma)$ is a score network parameterized by $\theta$ that is meant to approximate $\nabla \log p_\sigma(\xbm)$.
For a review of recent SGM algorithms, we refer to the survey presented in~\cite{Yang.etal2023diffusion}.
Lastly, we note that the use of Tweedie’s formula and score matching has also been studied in the context of PnP and RED~\cite{Bigdeli.etal2017, Reehorst.Schniter2019, Xu.etal2020, Kamilov.etal2023}.

\subsection{Related probabilistic imaging methods}
We here review the related methods that leverage DL models to perform posterior sampling defined in~\eqref{Eq:Posterior}.
Table~\ref{Tab:Comparison} provides a conceptual comparison of the methods discussed in this section.
Note that this table should not be interpreted as an indicator of the practical performance of these methods.

\subsubsection{Variational Bayesian methods}
Deep variational Bayesian methods often parameterize the posterior distribution as a deep neural network, which can be trained by minimizing the evidence lower bound (ELBO).
It has been shown that ELBO is related to the Kullback-Leibler divergence with respect to (w.r.t.) the posterior distribution~\cite{Kingma.etal2013}.
Recent works have leveraged this methodology for probabilistic imaging under either hand-crafted~\cite{SunHe.2021} or score-based priors~\cite{Feng.etal2023scorebased}.
While these methods are easy to use, the evaluation of the ELBO requires access to the log prior.
Due to this requirement, \cite{Feng.etal2023scorebased} needs to solve an expensive high-dimensional ODE to obtain the log prior associated with the score network, making itself unscalable to large-scale imaging problems. 
A follow-up work~\cite{Feng.etal2023} addresses the computation limit at the cost of using an inexact log prior.

\subsubsection{DM-based methods}
A DM is one type of SGM that learns to sample a distribution $p$ by reversing a diffusion process from $p$ to some simple distribution~\cite{Yang.etal2023diffusion}.
The two processes can be mathematically formulated as a diffusion SDE and its reverse-time SDE that relies on the time-varying score $\nabla\log p_t(\xbm_t)$.
To adapt DMs for posterior sampling, recent works focus on constraining the reverse process by integrating the likelihood. 
One line of work uses the proximal operator~\cite{Song2022solving, Liu.etal2023}, while others resort to leveraging the gradient~\cite{Chung.etal2023diffusion}.
These methods are empirically attractive due to their short inference time, but generally lack theoretical guarantees on their convergence to the posterior distribution.
Recent experimental evidence~\cite{Feng.etal2023scorebased,Cardoso.eal2023} shows that they lead to inaccurate sampling even in simple cases.

\subsubsection{MCMC-based methods}
MCMC methods recently received significant interest due to the emergence of SGMs.
The majority of MCMC-based methods for imaging inverse problems are based on the Langevin MCMC.
Recent works have studied the incorporation of linear forward models via singular value decomposition (SVD)~\cite{kawar2021snips} or the recovery analysis under random Gaussian matrix~\cite{Jalal.etal2021}.
Despite their empirical performance, these works are often restricted to linear inverse problems and lack theoretical guarantees on their algorithmic convergence to the posterior distribution; 
see also~\cite{Guo.etal2019, Kadkhodaie.etal2021}.
A recent work~\cite{Laumont.etal2022} has analyzed the convergence of Langevin MCMC for posterior sampling,  under potentially non-log-concave likelihoods and DL-based image denoisers. 
The analysis in \cite{Laumont.etal2022} obtains convergence guarantees in terms of the Wasserstein metric or total variations (TV), which is stronger than the Fisher information criterion used in the proposed analysis. 
On the other hand, ours obtains an explicit convergence rate and is fully compatible with annealing, covering aspects that were not addressed by~\cite{Laumont.etal2022}.
We further note that our analysis can be interpreted as a sampling analogue to the fixed-point analysis established for PnP/RED.
Other MCMC algorithms have also been investigated for posterior sampling.
A line of concurrent works combines the Gibbs splitting with either denoising~\cite{ Bouman.etal2023generative} or generative priors~\cite{Coeurdoux.etal2023}.
The key idea of this approach is to alternately sample from the likelihood and prior distributions conditioned on the previous iterate.
These methods differ from the Langevin-based methods, including the proposed PMC, by not relying on the score function.

Before concluding this section, we also highlight some pioneering works that investigated probabilistic imaging before the emergence of SGMs. 
One notable line of work combines Langevin MCMC with proximal techniques for sampling log-concave yet nonsmooth distributions~\cite{Pereyra2016, Brose.etal2017, Durmus.etal2018}.
By leveraging convex analysis, convergence guarantees can be established for those methods under different setups, including Moreau-Yosida envelope~\cite{Pereyra2016, Brose.etal2017, Durmus.etal2018} and proximal gradient~\cite{Pereyra2016, Eftekhari.etal2023}.
We note that these methods serve as the foundation for the development of \emph{plug-and-play unadjusted Langevin algorithm (PnP-ULA)}~\cite{Laumont.etal2022}.

\section{Method: Plug-and-Play Monte Carlo}
\label{Sec:Method}
We start the presentation of PMC by deriving the gradient-flow ODE as the continuous limit of both PnP and RED.
By extending the gradient-flow ODE to Langevin SDE, we then formulate the PMC algorithms. 
Lastly, we present the annealed PMC algorithms by introducing weighted annealing.

\begin{algorithm}[t]
\setstretch{1.05}
\caption{Plug-and-play Monte Carlo (PMC)}\label{Alg:PMC}
\begin{algorithmic}[1]
\STATE \textbf{input: } $\xbm_0 \in \R^n$, $\gamma > 0$, and $\sigma>0$.
\FOR{$k = 0, 1, \dots, N-1$}
\STATE $\Zk \leftarrow \Ncal(0, I)$
\STATE \textbf{switch} \emph{discretization}
\STATE \quad \textbf{case} PnP: \hfill (PMC-PnP)
\STATE \quad \quad $\Pcal(\Xk) \leftarrow \nabla g(\Xk) - \Scal_\theta \big( \Xk - \gamma \nabla g(\Xk), \sigma \big)$
\STATE \quad \quad $\Xkpo \leftarrow \Xk - \gamma \Pcal(\Xk) + \sqrt{2\gamma} \Zk$
\vspace{0.5em}
\STATE \quad \textbf{case} RED: \hfill (PMC-RED)
\STATE \quad \quad  $\Gcal(\Xk) \leftarrow \nabla g(\Xk) - \Scal_\theta \big( \Xk, \sigma \big)$
\STATE \quad \quad $\Xkpo \leftarrow \Xk - \gamma \Gcal(\Xk) + \sqrt{2\gamma} \Zk$
\ENDFOR\label{euclidendwhile}
\end{algorithmic}
\end{algorithm}%

\subsection{Continuous-time interpretation of PnP/RED}
Let $\Rcal_\sigma\defn\Id-\Dcal_\sigma$ be the residual predictor. Consider the following reformulation of the PnP algorithm in~\eqref{Eq:PnP}
\begin{align}
\label{Eq:PnP2}
\Xkpo &= \Dcal_\sigma\big(\Xk - \gamma \nabla g(\Xk) \big) \nonumber\\
&= \Xk - \gamma \Big( \nabla g(\Xk) + \frac{1}{\gamma}\Rcal_\sigma \big( \Xk - \gamma \nabla g(\Xk) \big) \Big),
\end{align}
and the RED algorithm in~\eqref{Eq:RED}
\begin{align}
\label{Eq:RED2}
\Xkpo &= \Xk - \gamma \Big( \nabla g(\Xk) + \tau \big( \xbm-\Dcal_\sigma(\Xk) \big) \Big) \nonumber\\
\ 
&= \Xk - \gamma \Big( \nabla g(\Xk) + \tau\Rcal_\sigma \big( \Xk \big) \Big).
\end{align}
If $\Dcal_\sigma(\xbm)$ is an MMSE image denoiser, by invoking Tweedie's formula in~\eqref{Eq:Tweedie} we can derive the residual predictor as
$\Rcal(\xbm)=-\sigma^2\nabla\log \psigma(\xbm)$.
Plugging the residual predictor into~\eqref{Eq:PnP2} and~\eqref{Eq:RED2} and rearranging the terms yields the following discrete gradient-flow formulation for PnP
\begin{align}
\label{Eq:PnPscore}
&\frac{\Xkpo - \Xk}{\gamma}= - \Pbf(\Xk)\quad\text{where}\nonumber\\
&\Pbf(\xbm)\defni \nabla g(\xbm) - (\sigma^2/\gamma) \nabla \log \psigma \big( \xbm - \gamma\nabla g(\xbm) \big)
\end{align}
and for RED
\begin{align}
\label{Eq:REDscore}
&\frac{\Xkpo - \Xk}{\gamma} = - \Gbf(\Xk)\quad\text{where}\nonumber\\
&\Gbf(\xbm)\defni \nabla g(\xbm) - \tau\sigma^2 \nabla \log \psigma(\xbm).\quad\quad\quad\quad\quad\;
\end{align}
Note that the dynamics of PnP and RED are fully characterized by $\Pbf(\xbm)$ and $\Gbf(\xbm)$, respectively.
Let $\gamma = \sigma^2$ and $\tau = 1/\sigma^2$. We have
\begin{align}
\lim_{\sigma\rightarrow 0}\Pbf(\xbm)
&= \nabla g(\xbm) - \lim_{\sigma\rightarrow 0} \big\{\nabla \log \psigma \big( \xbm - \sigma^2\nabla g(\Xk) \big) \big\} \nonumber\\
&= \nabla g(\xbm) - \lim_{\sigma \rightarrow 0} \big\{ \nabla \log \psigma \vphantom{\big(\big)} ( \xbm ) \big\} = \lim_{\sigma\rightarrow 0}\Gbf(\xbm) \nonumber\\
&= -\nabla \log \ell(\ybm | \xbm) - \nabla\log p(\xbm) \nonumber\\
&= -\nabla \log \pi(\xbm | \ybm).
\end{align}
Here, we need $\nabla\log \psigma(\xbm)$ to converge uniformly to $\nabla \log p(\xbm)$, which holds under Assumption~\ref{As:Prior} and \ref{As:Mismatch} made in Section~\ref{Sec:ConvPMCRED}.
Note that $\sigma^2 = \gamma$ resembles the requirement in proximal gradient method that the strength of the proximal operator $\mu$ is set to $\gamma$.
By plugging the limits of $\Pbf(\xbm)$ and $\Gbf(\xbm)$ into~\eqref{Eq:PnPscore} and~\eqref{Eq:REDscore} and taking $\gamma$ to zero, it follows that PnP and RED both converge to the gradient-flow ODE given by
\begin{equation}
\label{Eq:GFODE}
\dXt = \nabla \log \pi(\Xt | \ybm)dt.
\end{equation}
Recall that $\pi(\xbm|\ybm) \propto \ell(\ybm|\xbm) p(\xbm)$ is the posterior distribution.
With this continuous picture, we can draw a few insights about PnP and RED that are not clear through the discrete formulations.
First, both PnP and RED are governed by the exact same gradient-flow ODE under the MMSE denoiser.
Although the two algorithms apply the residual predictor to different iterates, they can be viewed as simply using different discretization strategies for numerically computing the ODE.
Second, the gradient-flow ODE links PnP/RED to the Langevin diffusion described by the following SDE 
\begin{equation}
\label{Eq:Langevin}
\dXt = \nabla\log\pi(\Xt | \ybm)dt + \sqrt{2}\,\dBt,
\end{equation}
where $\{\Bt\}_{t\geq0}$ denotes the $n$-dimensional Brownian motion.
As shown in the seminal work~\cite{Jordan.etal1998}, Langevin diffusion can be interpreted as the gradient flow in the space of probability distributions.
The symmetry between~\eqref{Eq:GFODE} and~\eqref{Eq:Langevin} inspires us to develop the parallel MCMC algorithms of PnP/RED for posterior sampling.

\subsection{Formulation of PMC-PnP/RED}

The derivation of PMC-PnP and PMC-RED is based on discretizing the Langevin diffusion in~\eqref{Eq:Langevin} according to the discretization strategies used in PnP and RED, respectively.
To approximate $\nabla\log \psigma(\xbm)$, one can use an MMSE image denoiser $\Dcal_\sigma(\xbm)=\E[\zbm|\xbm]$ as in PnP/RED; however, it is difficult to obtain an exact MMSE denoiser in practice. 
As the score is the key, we instead bypass this difficulty by leveraging a score network to obtain a direct approximation $\Scal_\theta(\xbm,\sigma)\approx\nabla\log\psigma(\xbm)$, which is trained by minimizing the DSM loss in~\eqref{Eq:DSM}.
The formulation of PMC-PnP is given by following~\eqref{Eq:PnPscore} and~\eqref{Eq:Langevin}
\begin{align}
\label{Eq:PMCPnP}
&\Xkpo = \Xk - \gamma \Pcal(\Xk) + \sqrt{2\gamma}\Zk \nonumber\\
&\text{where}\quad \Pcal(\xbm) = \nabla g(\xbm) - \Scal_\theta \big( \xbm - \gamma\nabla g(\xbm), \sigma \big),
\end{align}
where $\Zk=\int_{k}^{k+1}\dBt$ follows the $n$-dimensional i.i.d standard normal distribution.
Similarly, we can derive PMC-RED
\begin{align}
\label{Eq:PMCRED}
&\Xkpo = \Xk - \gamma \Gcal(\Xk) + \sqrt{2\gamma}\Zk \nonumber\\
&\text{where}\quad \Gcal(\xbm) = \nabla g(\xbm) - \Scal_\theta(\xbm, \sigma). \quad\quad\quad\quad\;
\end{align}
Algorithms~\ref{Alg:PMC} summarizes the details of PMC-PnP and PMC-RED.
When $\sigma$ is sufficiently small, PMC-PnP/RED approximately samples from the true posterior distribution.
In section~\ref{Sec:Analysis}, we present a detailed analysis of both PMC algorithms.
We note that PMC-RED has been proposed as the PnP-ULA in~\cite{Laumont.etal2022}, while PMC-PnP has not been explicitly studied in the existing literature.

\begin{figure}[t!]
\centering
\includegraphics[width=0.6\linewidth]{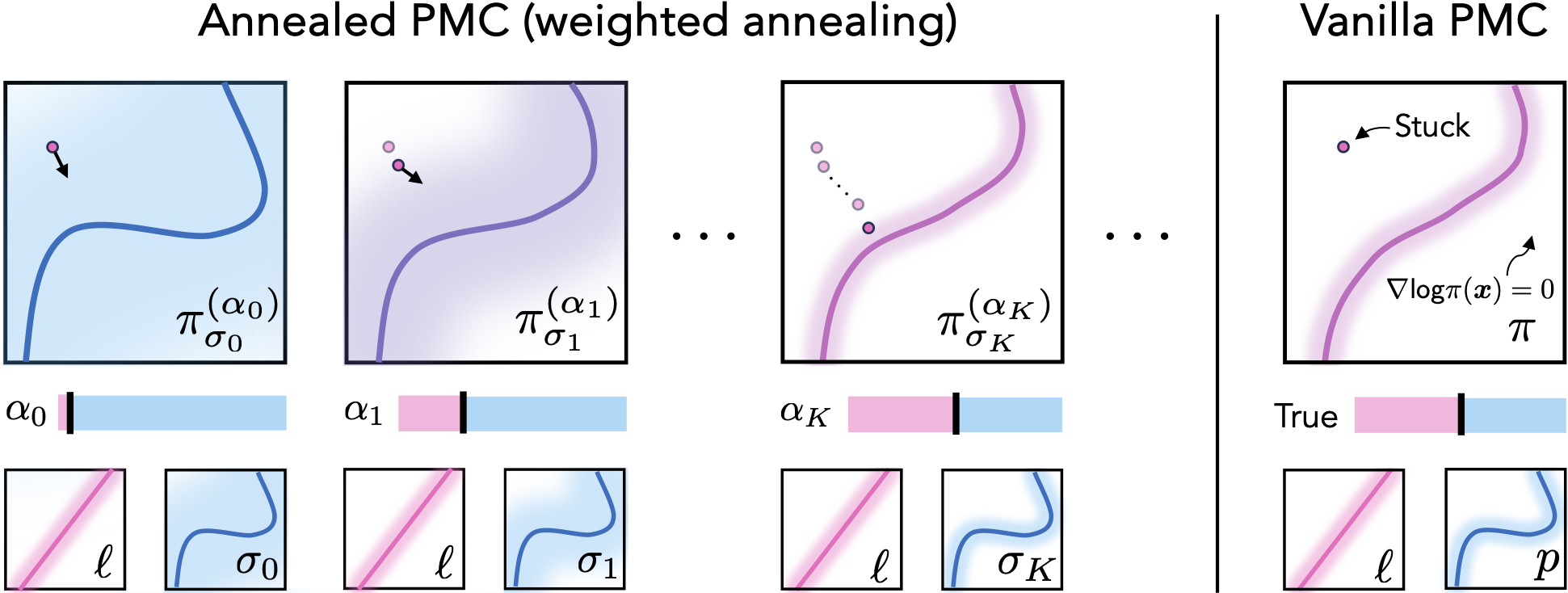}
\caption{
Conceptual illustration of how weighted annealing improves the convergence of APMC algorithms by introducing the weighted posteriors $\{\pi^{(\alphak)}_\sigmak\}$.
The solid curves and shade respectively denote the mean and probability density of the distribution; the white area means $\nabla \log p(\xbm)=0$.
In order to facilitate the \emph{vanilla PMC algorithm} to escape from plateaus in $\nabla \log p(\xbm)$, weighted annealing progressively decreases \emph{1)} the smoothing strength ($\sigmak$) of the prior and \emph{2)} its relative weights ($\alphak$) to the likelihood.
}
\vspace{-10pt}
\label{Fig:Annealing}
\end{figure}

\subsection{Weighted annealing for enhanced performance}

Langevin algorithms are known to suffer from slow convergence and mode collapse when sampling high-dimensional multimodal distributions.
Inspired by the annealed importance sampling~\cite{Neal.2001}, we propose \emph{weighted annealing} for alleviating these problems.
Our strategy considers a sequence of weighted posterior distributions 
\begin{align}\nonumber
&\pi^{(\alphak)}_\sigmak(\xbm | \ybm) \propto \ell(\ybm | \xbm)\psigmak^\alphak(\xbm), \\
\text{where} \quad &\alpha_0 > \alpha_1 > \dots > \alpha_K = \dots =\alpha_{N-1}=1, \nonumber\\
&\sigma_0 > \sigma_1 > \dots > \sigma_K = \dots = \sigma_{N-1}\approx0, \nonumber
\end{align}
where $\{\alphak\}_{k=0}^{N-1}$ and $\{\sigmak\}_{k=0}^{N-1}$ decays from large initial values to one and almost zero, respectively.
Note that $\alpha$ adjusts the relative weights of the likelihood and prior via powering.
At the beginning, the weighted posterior is dominated by the smoothed prior $\psigmak$, whose density is less concentrated and can facilitate the algorithm to escape from the plateaus in $\nabla \log p(\xbm)$. 
As the iteration number increases, the likelihood starts to contribute more strongly, and the smoothed prior $p_{\sigma_k}$ becomes similar to the true prior $p$, forcing the iterates to converge to the desired posterior.
Fig.~\ref{Fig:Annealing} conceptually illustrates this evolution. 
Intuitively, this annealing procedure speeds up the initial burn-in process by first flatting the distribution landscape and then gradually adding back complex structures. 
In practice, we observe that the weighted annealing works the best when $\{\alpha\}_{k=0}^{N-1}$ and $\{\sigma\}_{k=0}^{N-1}$ share the same schedule.

Algorithm~\ref{Alg:APMC} summarizes the details of the \emph{annealed PMC-PnP (APMC-PnP) and PMC-RED (APMC-RED)}.
A powering-free strategy has been proposed in~\cite{Jalal.etal2021},
which applies $\sigma$-smoothing to both likelihood and prior under heterogenous schedules.
Ours differs from~\cite{Jalal.etal2021} by using a shared schedule and combining the $\alpha$-powering and $\sigma$-smoothing mechanisms.
In the next section, we also analyze the convergence of the APMC algorithms.

\begin{algorithm}[t!]
\setstretch{1.05}
\caption{Annealed Plug-and-play Monte Carlo (APMC)}\label{Alg:APMC}
\begin{algorithmic}[1]
\STATE \textbf{input: } $\xbm_0 \in \R^n$, $\gamma > 0$, $\alpha_0 > 0$, and $\sigma_0 > 0$.
\FOR{$k = 0, 1, \dots, N-1$}
\STATE $\Zk \leftarrow \Ncal(0, I)$
\STATE $\sigma_k, \alpha_k \leftarrow \mathsf{WeightedAnnealing}(\sigma_0, \alpha_0, k)$.
\STATE \textbf{switch} \emph{discretization}
\STATE \quad \textbf{case} PnP:  \hfill (APMC-PnP)
\STATE \quad\quad $\Pcalk(\Xk) \leftarrow \nabla g(\Xk) - \alpha_k\Scal_\theta \big( \Xk - \gamma \nabla g(\Xk), \sigma_k \big)$
\STATE \quad \quad $\Xkpo \leftarrow \Xk - \gamma \Pcalk(\Xk) + \sqrt{2\gamma} \Zk$
\vspace{0.5em}
\STATE \quad \textbf{case} RED:  \hfill (APMC-RED)
\STATE \quad \quad $\Gcalk(\Xk) \leftarrow \nabla g(\Xk) - \alphak \Scal_\theta \big( \Xk, \sigmak \big)$
\STATE \quad \quad $\Xkpo \leftarrow \Xk - \gamma \Gcalk(\Xk) + \sqrt{2\gamma} \Zk$
\ENDFOR\label{euclidendwhile}
\end{algorithmic}
\end{algorithm}%

\section{Stationary-Distribution Analysis}
\label{Sec:Analysis}
Inspired by the fixed-point analysis of PnP/RED, we seek to establish the stationary-distribution convergence for the proposed PMC algorithms.
We start our analysis by first introducing the optimization interpretation of Langevin diffusion~\cite{Jordan.etal1998}.
Then, we present the assumptions and main results for the vanila and annealed PMC algorithms, respectively.

\subsection{Langevin diffusion as optimization}
Consider the following optimization of \emph{Kullback-Leibler (KL)} divergence in the space of probability distributions equipped with the Wasserstein metric
\begin{align}
\label{Eq:KLOptimization}
&\argmin_{\nu} \; \KL(\nu \,\|\, \pi) \nonumber\\
&\text{where}\quad\KL(\nu \,\|\, \pi) = \int_{\R^n} \nu(\xbm) \log\frac{\nu(\xbm)}{\pi(\xbm)} d\xbm.
\end{align}
where $\pi$ denotes the desired posterior distribution and $\nu$ the iterate.
Similar to the gradient concept in the Euclidean space, the gradient under the Wasserstein metric can be defined~\cite{villani2009optimal}.
In particular, the Wasserstein gradient of $\KL(\nu \,|\, \pi)$ is given by $\nabla_\nu \KL(\nu \,|\, \pi) = \nabla \log\frac{\nu(\xbm)}{\pi(\xbm)}$~\cite{Ambrosio2008}, and its expected norm is known as the \emph{relative Fisher information (FI)}
\begin{equation}
\label{Eq:FI}
\FI(\nu \,\|\, \pi)
= \int_{\R^n} \left\| \nabla \log\frac{\nu(\xbm)}{\pi(\xbm)} \right\|^2_2 \nu(\xbm) d\xbm.
\end{equation}
If $\nut$ denotes the distribution obtained by Langevin diffusion at time $t$, then the time derivative of $\KL(\nut \,\|\, \pi)$ is the negative FI, i.e. $\frac{d}{dt}\KL(\nut \,\|\, \pi) = -\FI(\nut \,\|\, \pi)$~\cite{Ambrosio2008,villani2009optimal},
which shows that the Langevin diffusion is a gradient flow in the probability space.
From an optimization point of view, $\FI(\nu_t \,\|\, \pi)$ is analogous to the $\ell$-2 norm of the gradient in $\R^n$\cite{Balasubramanian.etal2022}.
To leverage this, we derive the convergence of $\FI(\nu_t \,\|\, \pi)$ under a ``linear interpolation" of the distributions obtained by PMC algorithms (see Supplement I for more details), which implies the stationarity of the discrete algorithms.
Different from optimization, if $\nu$ and $\pi$ have positive and smooth densities, $\FI(\nu \,\|\, \pi)=0$ indicates that $\nu$ and $\pi$ are equal.
This implies that the value of FI can serve as a criterion to measure the convergence of distributions.
We refer to~\cite{Balasubramanian.etal2022} for more discussions on this topic.

\subsection{Convergence of stationary PMC algorithms}
\label{Sec:ConvPMCRED}

We begin our analysis by first considering the stationary PMC algorithms without weighted annealing.
\begin{assumption}
\label{As:Likelihood}
We assume that the log-likelihood $\log \ell(\ybm|\xbm)$ is differentiable and has a Lipschitz continuous gradient with constant $L_g>0$ for any $\xbm_1, \xbm_2\in\R^n$
\begin{equation}\tag{A1}
\big\| \nabla \log \ell(\ybm|\xbm_1) - \nabla \log \ell(\ybm|\xbm_2) \big\|_2 \leq L_g \big\|\xbm_1 - \xbm_2 \big\|_2.
\end{equation}
This is equivalent to assuming $\nabla g(\xbm) = -\nabla\log\ell(\ybm|\xbm)$ to be Lipschitz continuous with $L_g$.
\end{assumption}
Note that Assumption~\ref{As:Likelihood} does not assume the log-concavity of the likelihood (i.e. convexity of the data-fidelity term), meaning that our analysis is compatible with nonlinear inverse problems.

\begin{assumption}
\label{As:Prior}
We assume that the log-prior $\log p(\xbm)$ is differentiable and has a Lipschitz continuous gradient with a finite constant $L_p>0$ for any $\xbm_1, \xbm_2\in\R^n$
\begin{equation}\tag{A2}
\big\| \nabla \log p(\xbm_1) - \nabla \log p(\xbm_2) \big\|_2 \leq L_p \big\|\xbm_1 - \xbm_2 \big\|_2.
\end{equation}
\end{assumption}
This assumption is general and only poses the basic regularity condition for the underlying true prior. 

\begin{assumption}
\label{As:Score}
We assume the score network $\Scal_\theta(\xbm, \sigma)$ satisfies the following conditions
\begin{enumerate}[ref=\theassumption.\theenumi, label=(\alph*)]
\item\label{As:Score.a} For any $\sigma>0$, $\Scal_\theta(\xbm, \sigma)$ is Lipschitz continuous with $L_\sigma>0$ for any $\xbm_1, \xbm_2\in\R^n$
\begin{equation}\tag{A3.a}
\label{ScoreA}
\|\Scal_\theta(\xbm_1, \sigma) - \Scal_\theta(\xbm_2, \sigma)\|_2 \leq L_\sigma \|\xbm_1 - \xbm_2\|_2
\end{equation}
\item\label{As:Score.b} For any $\sigma>0$, 
$\Scal_\theta(\xbm, \sigma)$ has a bounded error $\epsilon_\sigma<+\infty$ for any $\xbm\in\R^n$
\begin{equation}\tag{A3.b}
\label{ScoreB}
\|\Scal_\theta(\xbm, \sigma) - \nabla\log\psigma(\xbm)\|_2 \leq \epsilon_\sigma
\end{equation}
\end{enumerate}
\end{assumption}
We highlight that Assumption~\ref{As:Score}(b) accounts for the network approximation error that is inevitable in practice.

\begin{assumption}
\label{As:Mismatch}
Let $\psigma(\xbm) = \int_{\R^n} p(\zbm) \phi_\sigma(\xbm-\zbm)d\zbm$ denote the smoothed prior, where $\phi_\sigma$ is the probability density function of $\Ncal(0,\sigma^2 I)$.
We assume $\nabla \log p_\sigma(\xbm)$ has a bounded error from $\nabla \log p(\xbm)$, that is, for any $\sigma>0$ and $\xbm\in\R^n$
\begin{equation}\tag{A4}
\| \nabla \log p_\sigma(\xbm) - \nabla \log p(\xbm) \|_2 \leq \sigma C.
\end{equation}
Note that $\log\psigma(\xbm)$ is continuously differentiable as $\psigma(\xbm)$ is a convolution of $p(\xbm)$ and $\phi_\sigma(\xbm)$.
\end{assumption}
Under Assumption~\ref{As:Prior} and \ref{As:Mismatch}, it holds that $\nabla\log\psigma(\xbm)$ converges to $\nabla\log p(\xbm)$ as $\sigma$ decreases to zero.
In special cases, such as $p$ being a Gaussian, the bound of the score mismatch can be derived analytically. However, it is generally difficult to derive a closed-form expression of the mismatch for an arbitrary distribution.

Under these assumptions, we now derive the convergence of PMC-RED.
\begin{theorem}[PMC-RED]
\label{Th:PMCRED} 
Let $\{\nu_t\}_{t\geq0}$ denote the law for the continuous interpolation of $\{\Xk\}_{k=0}^{N}$ generated by PMC-RED and $N>0$ the total number of iterations.
Assume that Assumptions~\ref{As:Likelihood}-\ref{As:Mismatch} hold. 
Then, for any $\gamma$ such that $\gamma L \leq 1/\sqrt{32}$, we have 
\begin{align}
&\frac{1}{N\gamma}\int_{0}^{N\gamma} \FI(\nu_t \,\|\, \pi) \,dt \tag{T1}\\
&\leq \frac{4\KL(\nu_{0} \,\|\, \pi)}{N\gamma}
+ \underbrace{\vphantom{\Big|} \Asf_1\gamma}_{\substack{\text{Discretization} \\ \text{Error}}}
+ \underbrace{\vphantom{\Big|} \Asf_2 \sigma^2}_{\substack{\text{Score Mismatch} \\ \text{Error}}}
+ \underbrace{\vphantom{\Big|} \Asf_3\epsilon_\sigma^2}_{\substack{\text{Approximation} \\ \text{Error}}} \nonumber
\end{align}
where $L=L_g + \max\{L_\sigma,L_p\}$ and the constants are given by 
\begin{gather}
\Asf_1=24nL^2, \quad \Asf_2=24 C^2, \quad \Asf_3 = 24. \nonumber
\end{gather}
\end{theorem}
\begin{proof}
See Supplement~I.B for a detailed proof.
\end{proof}

In order to derive the result for PMC-PnP we need one additional assumption on the likelihood:
\begin{assumption}
\label{As:LikelihoodExtra}
We assume that the $\ell$-2 norm of the gradient of log-likelihood is bounded, namely $\|\nabla g(\xbm)\|_2 \leq R_g$ with $R_g>0$.
\end{assumption}
The boundness of the $\nabla g(\xbm_k)$ is practical and can be achieved by imposing gradient clipping at every iteration.
In practice, we observe that the algorithm converges under a Gaussian likelihood without clipping.

\begin{theorem}[PMC-PnP]
\label{Th:PMCPnP} 
Let $\{\nu_t\}_{t\geq0}$ denote the law for the continuous interpolation of $\{\Xk\}_{k=0}^{N}$ generated by PMC-PnP and $N>0$ the total number of iterations.
Assume that Assumptions~\ref{As:Likelihood}-\ref{As:LikelihoodExtra} hold. 
Then, for any $\gamma$ such that $\gamma L \leq 1/\sqrt{32}$, we have
\begin{align}
&\frac{1}{N\gamma}\int_{0}^{N\gamma} \FI(\nu_t \,\|\, \pi) \,dt \tag{T2}\\
&\leq \frac{4\KL(\nu_{0} \,\|\, \pi)}{N\gamma}
+ \underbrace{\vphantom{\Big|} \Bsf_1\gamma}_{\substack{\text{Discretization} \\ \text{Error}}}
+ \underbrace{\vphantom{\Big|} \Bsf_2\sigma^2}_{\substack{\text{Score Mismatch} \\ \text{Error}}}
+ \underbrace{\vphantom{\Big|} \Bsf_3\epsilon_\sigma^2}_{\substack{\text{Approximation} \\ \text{Error}}} \nonumber
\end{align}
where $L=L_g + \max\{L_\sigma + \gamma L_g L_\sigma, L_p\}$ and the constants are given by 
\begin{gather}
\Bsf_1=24nL^2 + 7 L_\sigma R_g^2, \quad \Bsf_2=36C^2, \quad \Bsf_3 = 36. \nonumber
\end{gather}
\end{theorem}
\begin{proof}
See Supplement~I.C for a detailed proof and the range of step-size.
\end{proof}
We note that the two theorems resemble each other, indicating that PMC-RED and PMC-PnP possess similar convergence behaviors. This resemblance is consistent with the fact that both algorithms originate from~\eqref{Eq:Langevin}, albeit employing different discretization rules.
Furthermore, the theorems demonstrate that the averaged $\FI(\nu_t \,\|\, \pi)$ is bounded by a diminishing term as $N$ goes to infinity, up to some errors respectively proportional to step-size $\gamma$, squared smoothing strength $\sigma^2$, and squared approximation error $\epsilon_\sigma^2$.
Additionally, due to the convexity of FI, it follows that 
$$\FI(\bar{\nu}_{N\gamma} \,\|\, \pi) \leq \frac{1}{N\gamma}\int_{0}^{N\gamma} \FI(\nu_t \,\|\, \pi) \,dt,$$
where $\bar{\nu}_{N\gamma} \defn (N\gamma)^{-1}\int_0^{N\gamma} \nu_t dt$ denotes the averaged distribution of $\{\nut\}_{t\geq0}$.
This averaged distribution can be computed given the iterates of the algorithm; for more discussions we refer to \cite{Balasubramanian.etal2022}.
In this regard, Theorem 1 and 2 establish the stationary-distribution convergence of $\bar{\nu}_{N\gamma}$ generated by PMC-PnP and PMC-RED.

\subsection{Convergence of annealed PMC algorithms}
\label{Sec:ConvPMCPnP}

Characterizing the convergence of the annealed PMC algorithms is more challenging as the intermediate posterior varies over iterations.
In this section, we tackle this problem by deriving explicit bounds of FI for APMC algorithms.

\begin{assumption}
\label{As:Posterior}
We assume that the output of the score network $\Scal_\theta(\xbm, \sigma)$ is bounded in $\ell$-2 norm, namely $\|\Scal_\theta(\xbm,\sigma)\|_2\leq R_s$.
\end{assumption}
This assumption assumes that the output of the deep score network is bounded. Similar to Assumption~\ref{As:LikelihoodExtra}, we can implement this condition in practice by clipping the output of the score network in the inference time. 
In the experiments, we observe that APMC algorithms also converge without imposing this condition.

Note that Assumption~\ref{As:Prior} ensures the Lipschitz constant of the true score is not exploding.
This is necessary for the existence of $\sup\{L_\sigmak\}_{k=0}^{N-1}$ as $N$ goes to infinity. Here, $L_\sigmak$ denotes the Lipschitz constant of the score network at iteration $k$.

\begin{theorem}[APMC-RED]
\label{Th:PMCREDanneal}
Let $\{\alphak\}_{k=0}^{N-1}$ and $\{\sigmak\}_{k=0}^{N-1}$ be decreasing sequences where $\alpha_{K,\dots,N-1}=1$.
Let $\{\nu_t\}_{t\geq0}$ denote the law for the continuous interpolation of $\{\Xk\}_{k=0}^{N}$ generated by APMC-RED and $N>0$ the total number of iterations.
Assume that Assumptions~\ref{As:Likelihood}-\ref{As:Mismatch} and~\ref{As:Posterior} hold. 
Then, for any $\gamma$ such that $\gamma \Lmax \leq 1/\sqrt{32}$, we have
\begin{align}
&\frac{1}{N\gamma}\int_{0}^{N\gamma} \FI(\nu_t \,\|\, \pi) \,dt \tag{T3}\\
&\leq \frac{4\KL(\nu_{0} \,\|\, \pi) + \gamma\zeta}{N\gamma}
+ \underbrace{\vphantom{\Big|} \Csf_1\gamma}_{\substack{\text{Discretization} \\ \text{Error}}}
+ \underbrace{\vphantom{\Big|} \Csf_2 \sigmabar^2}_{\substack{\text{Score Mismatch} \\ \text{Error}}}
+ \underbrace{\vphantom{\Big|} \Csf_3\epsilonbar^2}_{\substack{\text{Approximation} \\ \text{Error}}} \nonumber
\end{align}
where $L_k = L_g + \max\{\alphak L_\sigmak, L_p\}$, $\Lmax = \sup\,\{L_k\}_{k=0}^{N-1}$, $\epsilonbar^2 = \frac{1}{N}\sum_{k=0}^{N-1}\epsilon_\sigmak^2$, and $\bar{\sigma}^2 = \frac{1}{N}\sum_{k=0}^{N-1}\sigma^2_k$.
Here, the constants are given by 
\begin{gather}
\Csf_1=24n\Lmax^2, \quad\Csf_2=36C^2, \quad\Csf_3=36, \nonumber\\
\zeta=36\sum_{k=0}^{K-1}(\alphak-1)^2R_s^2, \nonumber
\end{gather}
where we recall $0\leq K < N$ is the iteration index such that $\alpha_{K,\dots,N-1}=1$.
\end{theorem}
\begin{proof}
See Supplement~I.D for a detailed proof.
\end{proof}

\begin{theorem}[APMC-PnP]
\label{Th:PMCPnPanneal}
Let $\{\alphak\}_{k=0}^{N-1}$ and $\{\sigmak\}_{k=0}^{N-1}$ be decreasing sequences where $\alpha_K =\alpha_{K+1}=\dots=1$.
Let $\{\nu_t\}_{t\geq0}$ denote the law for the continuous interpolation of $\{\Xk\}_{k=0}^{N}$ generated by APMC-PnP and $N>0$ the total number of iterations.
Assume that Assumptions~\ref{As:Likelihood}-\ref{As:Posterior} hold. 
Then, for any $\gamma$ such that $\gamma\Lmax\leq1/\sqrt{32}$, we have
\begin{align}
&\frac{1}{N\gamma}\int_{0}^{N\gamma} \FI(\nu_t \,\|\, \pi) \,dt \tag{T4}\\
&\leq \frac{4\KL(\nu_{0} \,\|\, \pi)+\gamma\zeta}{N\gamma}
+ \underbrace{\vphantom{\Big|} \Dsf_1\gamma}_{\substack{\text{Discretization} \\ \text{Error}}}
+ \underbrace{\vphantom{\Big|} \Dsf_2 \sigmabar^2}_{\substack{\text{Score Mismatch} \\ \text{Error}}}
+ \underbrace{\vphantom{\Big|} \Dsf_2\epsilonbar^2}_{\substack{\text{Approximation} \\ \text{Error}}} \nonumber
\end{align}
where $L_k = L_g + \max\{\alpha_k L_\sigmak + \alpha_k\gamma L_g L_\sigmak, L_p\}$, 
$\Lmax = \sup\,\{L_k\}_{k=0}^{N-1}$, 
$\epsilonbar^2 = \frac{1}{N}\sum_{k=0}^{N-1}\epsilon_\sigmak^2$, and $\bar{\sigma}^2 = \frac{1}{N}\sum_{k=0}^{N-1}\sigma_k^2$. 
Here, the constants are given by
\begin{gather}
\Dsf_1=24n\Lmax^2 + 9L_{\sigma_\mathsf{max}} R_g^2, \quad\Dsf_2=48C^2, \quad\Dsf_3=48, \nonumber\\
\zeta=48\sum_{k=0}^{K-1}(\alphak-1)^2R_s^2, \nonumber
\end{gather}
where $L_{\sigma_\mathsf{max}} = \sup\,\{L_\sigmak\}_{k=0}^{N-1}$ and we recall $0\leq K < N$ is the iteration index such that $\alpha_{K,\dots,N-1}=1$.
\end{theorem}
\begin{proof}
See Supplement~I.E for a detailed proof and the range of step-size.
\end{proof}
Overall, the bounds of APMC algorithms are consistent with their stationary counterparts; similarly, the theorems also establish the stationary-distribution convergence of $\bar{\nu}_{N\gamma}$ generated by APMC-PnP and APMC-RED.
The key difference in the results of Theorems~\ref{Th:PMCREDanneal}-\ref{Th:PMCPnPanneal} is that the last two errors are proportional to the averaged squared values $\sigmabar^2$ and $\epsilonbar^2$, respectively.
As $\{\sigmak\}_{k=0}^{N-1}$ is pre-defined by the user, we can make the sequence square-summable such that $\sigmabar^2$ decreases to zero as $N$ goes to infinity. Thus, the score mismatch error can be asymptotically removed.
We note that the approximation error can also be eliminated if $\epsilonbar^2$ is squared-summable, which is a weaker condition than assuming the score network to be error-free.
Nevertheless, it is quite challenging to precisely control $\epsilon_\sigmak$ of score networks.

\section{Numerical Validations of Theory}
The objective of this section is to verify the capability of APMC algorithms to correctly sample from posterior distributions.
We focus on the annealed algorithms as they subsume the stationary variants.
We construct two experiments: a numerical validation of the proposed convergence analysis and a statistical study on sampling for images.

\subsection{Numerical validation of convergence}

One key insight of our analysis is that the averaged FI values obtained by annealed PMC algorithms w.r.t. posterior $\pi$ are proportional to step-size $\gamma$, averaged smoothing strength $\sigmabar^2$, and averaged approximation error $\epsilonbar^2$.
In order to efficiently verify the dependency, we consider a two-dimensional (2D) posterior distribution that is characterized by the Gaussian likelihood described in Section~\ref{Sec:InverseProblem} and a bimodal Gaussian mixture prior.
The measurement $\ybm$ is obtained by evaluating~\eqref{Eq:Inverse} at $x=(0,0)$, with AWGN corruption of variance $\beta^2=1$.
We construct twenty test posterior distributions by generating random realizations of $\Abm$ while keeping the rest unchanged.
Rather than use a learned score model, we simulate an imperfect score by adding AWGN to the analytical score of the prior. 
This allows us to control $\epsilonbar^2$ by adjusting the maximal norm $\epsilon_\mathsf{max}$ of the noise.
Similarly, we adjust the minimal smoothing strength $\sigma_\mathsf{min}$ to control $\sigmabar^2$.
Each PMC algorithm is run to infer a batch of $1000$ samples that are initialized randomly. This is equivalent of running $1000$ independent chains and taking the outputs of their last iteration as samples.
We refer to Supplement~II.A for additional technical details.

Table~\ref{Tab:FIKL} empirically evaluates the evolution of the minimal $\FI(\nu_k \,\|\, \pi)$ and $\KL(\nu_k \,\|\, \pi)$ obtained by \anneal-PnP and \anneal-RED with different $\gamma$, $\sigma_\mathsf{min}$, and $\epsilon_\mathsf{max}$ values.
The values of $\FI(\nu_k \,\|\, \pi)$ and $\KL(\nu_k \,\|\, \pi)$ are averaged over all testing distributions.
We observed convergence in both metrics for the APMC algorithms in all tests.
The table clearly illustrate the improvement in $\FI(\nu_k \,\|\, \pi)$ by reducing the value of these parameters as illustrated in Theorem~\ref{Th:PMCREDanneal} and~\ref{Th:PMCPnPanneal}.
Although FI can not be interpreted as a direct proxy of KL~\cite{Balasubramanian.etal2022}, remarkably a similar trend is also observed for $\KL(\nu_k \,\|\, \pi)$.
The convergence plots of APMC-PnP and APMC-RED are shown in Supplement~III.

\begin{table*}[t]
\centering
\scriptsize
\caption{Minimal values of the averaged relative Fisher information (FI) and Kullback-Leibler (KL) divergence over the testing posterior distributions.
The convergence in both metrics are observed for the annealed PMC algorithms in each test.}
\begin{tabular*}{455pt}{L{40pt} C{25pt} C{30pt}C{25pt}C{25pt} C{0pt} C{25pt}C{25pt}C{25pt} C{0pt} C{25pt}C{25pt}C{25pt}} \toprule
\multirow{2}{*}{\shortstack[l]{\textbf{Annealed} \\ \textbf{algorithms}}} & \multirow{2}{*}{\textbf{Metrics}} & \multicolumn{3}{c}{$\gamma$} & & \multicolumn{3}{c}{$\sigma_\mathsf{min}$} & & \multicolumn{3}{c}{$\epsilon_\mathsf{max}$} \\
\cmidrule{3-5} \cmidrule{7-9} \cmidrule{11-13}
& & $1.6$ & $0.8$ & $0.4$ & & $0.4$ & $0.2$ & $0.1$ & & $5.0$ & $2.5$ & $1.25$ \\
\cmidrule{1-13}
\multirow{2}{*}{APMC-PnP}   & FI & $0.9455$ & $0.0985$ & $0.0260$ & & $0.2658$ & $0.1148$ & $0.0565$ & & $0.3018$ & $0.0260$ & $0.0076$ \\
						& KL & $3.1782$ & $0.6336$ & $0.5690$ & & $1.0616$ & $0.7429$ & $0.6259$ & & $1.1187$ & $0.5690$ & $0.5651$ \\
\noalign{\vskip 0.5ex}
\hdashline\noalign{\vskip 0.5ex}
\multirow{2}{*}{APMC-RED}  & FI & $0.9247$ & $0.0818$ & $0.0218$ & & $0.2723$ & $0.1191$ & $0.0577$ & & $0.2923$ & $0.0218$ & $0.0051$ \\
                                                   & KL & $3.2423$ & $0.6186$ & $0.5632$ & & $1.0792$ & $0.7541$ & $0.6315$ & & $1.1097$ & $0.5632$ & $0.5614$ \\
\bottomrule
\end{tabular*}
\label{Tab:FIKL}
\vspace{-10pt}
\end{table*}

\subsection{Statistical validation of image posterior sampling}

We now validate the \anneal~algorithms for sampling from a posterior distribution of images.
To demonstrate this, we compare the sample statistics generated by APMC and the grounth-truth posterior.
All images used in this validation are taken from the CelebA dataset~\cite{Liu.etal2015-CelebA} with normalization to $[-1,1]$ and rescaling to $32\times32$ pixels for efficient computation.
Since the two APMC algorithms are symmetric, we focus the discussion on APMC-PnP, deferring the results of APMC-RED to Supplement~III for brevity.

We adopt the same setup of the posterior as the previous validation. 
We set $\Abm$ to a $307\times1024$ random Gaussian matrix, and $\ybm$ is simulated by evaluating~\eqref{Eq:Inverse} at a test image under AWGN with variance $\beta^2=0.01$.
The two modes of the bimodal Gaussian prior are set to correspond to the female and male images, respectively.
We manually shift the means of the two modes by $-2$ and $+2$; otherwise, they are indistinguishable under the Gaussian assumption.
Rather than using the true score corrupted by noise, we instead pre-train a deep score network (see Section~\ref{Sec:ScoreNet} for details) to approximate the score of the prior. 
The training images are obtained by drawing random samples from the Gaussian prior.
We observe that the Gaussian images are less structured, and thus the performance of the score network on this synthetic dataset should not be interpreted as an indication of its performance on natural images.
We also run the algorithms with the analytical score to show the performance of APMC-PnP in the ideal case.
Each PMC algorithm is run to infer a batch of $1000$ samples that are initialized randomly.
We refer to Supplement~II.B for the additional technical details.

\begin{figure}[t!]
\centering
\includegraphics[width=0.5\linewidth]{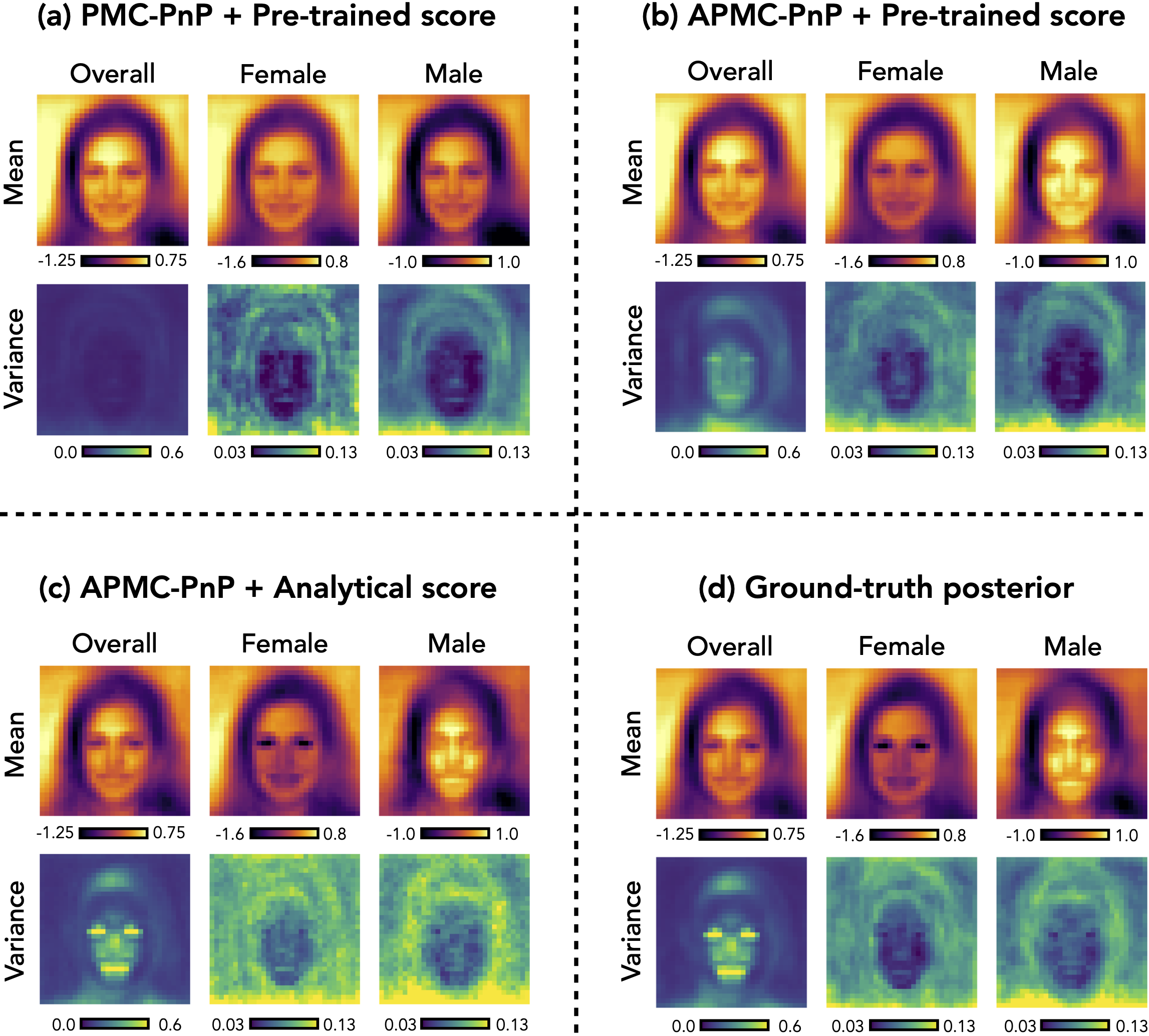}
\caption{
Comparison of the sample statistics obtained by PMC-PnP and APMC-PnP versus the ground-truth posterior distribution.
Each test algorithm is run to infer a batch of $1000$ samples, which are then classified into two modes according to their distance and angle with respect to the ground-truth modes.
PMC-PnP, which does not use annealing, fails to identify the bimodal distribution, as can be seen by the two groups of classified images being nearly indistinguishable.
Under the pre-trained score, APMC-PnP significantly improves in performance over PMC-PnP by distinguishing the female and male modes.
These two modes look similar to the ground truth posterior modes, although some differences remain due to the use of an approximate learned score.
In the ideal case of the analytical score, APMC-PnP recovers a distribution that closely resembles the ground-truth posterior.}
\vspace{-10pt}
\label{Fig:Face_PMCPnP}
\end{figure}

Fig.~\ref{Fig:Face_PMCPnP} compares the sampling performance of \anneal-PnP and PMC-PnP against ground-truth posterior.
All algorithms are run until convergence to sample $1000$ images, and we use a large non-divergent step-size to accelerate the procedure.
Fig.~\ref{Fig:Face_PMCPnP}(a) and~\ref{Fig:Face_PMCPnP}(b) compares the PMC-PnP and APMC-PnP under the pre-trained score.
It is clear that the PMC-PnP struggles in sampling the two modes, leading to a unitary cluster of samples that are inseparable.
On the contrary, APMC-PnP can successfully capture the two modes by leveraging weighted annealing. 
Note the improvement in the overall and per-mode statistics led by APMC-PnP.
When the analytical score is used, Fig.~\ref{Fig:Face_PMCPnP}(c) and~\ref{Fig:Face_PMCPnP}(d) display almost identical overall statistics, demonstrating the high-quality recovery of the posterior distribution by \anneal-PnP.
Note the closeness between the means of APMC-PnP and groundtruth.
However, we notice that \anneal-PnP tends to get larger sample variances for each mode.
This might be attributed to the large step-size ($\gamma = 10^{-3}$) and could be mitigated by using smaller values.

\section{Experiments on Imaging Inverse Problems}

In this section, we demonstrate the performance of APMC algorithms on real-world imaging inverse problems.
Our experiments consist of three high-dimensional image recovery tasks: compressed sensing (CS), magnetic resonance imaging (MRI), and black-hole interferometric imaging (BHI).
Together, these tasks provide coverage of both log-concave and non-log-concave likelihoods.

\subsection{Score-based generative priors}
\label{Sec:ScoreNet}
To fully leverage the latest advances in SGM, we implement our score network by customizing the state-of-the-art U-Net architecture in guided diffusion~\cite{Dhariwal.etal2021}.
We introduce a key modification to make the network take the smoothing strength $\sigma$ as input by leveraging techniques in~\cite{Song.etal2020b}.
We train individual networks for approximating the scores of three image prior distributions, that is, human face images, brain MRI images, and black-hole images, which are used in the experiments of CS, MRI, and BHI, respectively.
Note that the training of the score network is \emph{agnostic} to the forward models of these imaging tasks.
We train the score network over the smoothing ranges $\sigma\in[0.01, 348]$ and $\sigma\in[0.01, 192]$ for $256\times256$ and $64\times64$ images, respectively.
Note that these ranges are suggested by~\cite{Song.etal2020b} for achieving the optimized performance.
We refer to Supplement II for additional technical details on the architecture and training.

\subsection{Linear Inverse Problems: CS and MRI}
The inverse problems of CS and MRI that we tackle can be described by a linear version of the system in~\eqref{Eq:Inverse}, where $\Abm$ corresponds to either a i.i.d random Gaussian matrix or the Fourier transform with subsampling, and $\ebm$ corresponds to 40 dB input signal-to-noise ratio (SNR).
Both tasks are performed on a dataset of $40$ images with the spatial resolution of $256\times256$ pixels, where $10$ held-out images are used for finetuning the hyperparameters, and the remaining $30$ images are used for collecting the test results.
All images are normalized to $[-1,1]$.
We include four baseline methods for comparison: PnP, RED, PnP-ULA~\cite{Laumont.etal2022}, and DPS~\cite{Chung.etal2023diffusion}.
The first two algorithms correspond to MAP-based methods, while the rest are designed for posterior sampling.
In the test, PnP and RED are equipped with the pre-trained DnCNN denoisers from~\cite{Zhang.etal2017}, while PnP-ULA and DPS use the score networks.
We note that PnP-ULA and PMC-RED are essentially the same algorithm, with the sole difference being that PnP-ULA additionally projects every iterate to a convex and compact set. However, since the set is user-defined~\cite{Laumont.etal2022}, PMC-RED can be viewed as PnP-ULA equipped with a non-activated projection to a large set such as $[\shortneg10^4 , 10^4]^n$; hence, we refer to PMC-RED as \emph{PnP-ULA-NP} in this section. 
We additionally employed \emph{PnP-ULA-Ori}, which adopts a smaller set $[\shortneg1,2]^n$ and $\sigma=5/255$ originally used in~\cite{Laumont.etal2022}.
We refer to Supplement II for more discussion on PnP-ULA and implementation details of the baseline methods.
In each experiment, we run algorithms for a maximum number of $10,000$ iterations to collect the final results.
Each sampling algorithm is run to infer a batch of $50$ image samples that are initialized randomly.
The final reconstructed image is obtain by averaging these samples.

\afterpage{
\begin{table*}[t]
\centering
\tiny
\caption{
Averaged PSNR and MSE values obtained by the proposed PMC and baseline algorithms for the CS ($m/n=0.1$ \& $m/n=0.3$) and MRI ($\text{Accel.}=8\times$ \& $\text{Accel.}=4\times$) tasks. The best values of ours and baselines are highlighted in \textbf{bold} and \underline{underline}, respectively.}
\begin{tabular*}{470pt}{L{50pt} C{40pt}C{35pt} C{40pt}C{35pt} C{0pt} C{40pt}C{35pt} C{40pt}C{35pt}} \toprule
\multirow{2}{*}{\textbf{Method}} & \multicolumn{2}{c}{$m/n=0.1$} & \multicolumn{2}{c}{$m/n=0.3$} & & \multicolumn{2}{c}{$\text{Accel.} = 8\times$} & \multicolumn{2}{c}{$\text{Accel.} = 4\times$}\\
\cmidrule{2-5} \cmidrule{7-10}
& PSNR (dB) $\uparrow$ & MSE $\downarrow$ & PSNR (dB) $\uparrow$ & MSE $\downarrow$ & & PSNR (dB) $\uparrow$ & MSE $\downarrow$ & PSNR (dB) $\uparrow$ & MSE $\downarrow$\\
\cmidrule{1-10}
PnP & \second{$24.90$} & \second{$\expnumber{3.42}{\shortneg3}$} & $32.36$ & $\expnumber{6.11}{\shortneg4}$ & & $28.98$ & $\expnumber{1.40}{\shortneg3}$ & $33.81$ & $\expnumber{5.06}{\shortneg4}$ \\
RED & $24.74$ &  $\expnumber{3.55}{\shortneg3}$ & $32.18$ & $\expnumber{6.36}{\shortneg4}$ & & $28.96$ &  $\expnumber{1.41}{\shortneg3}$ & $33.39$ & $\expnumber{6.11}{\shortneg4}$ \\
\noalign{\vskip 0.5ex}
\hdashline\noalign{\vskip 0.5ex}
DPS & $24.37$ & $\expnumber{4.18}{\shortneg3}$ & $31.45$ & $\expnumber{8.04}{\shortneg4}$ & & $31.71$ & $\expnumber{7.43}{\shortneg4}$ & $34.88$ & $\expnumber{3.55}{\shortneg4}$ \\ 
PnP-ULA-Ori & $8.56$ & $\expnumber{1.51}{\shortneg1}$ & \second{$33.58$} & \second{$\expnumber{4.99}{\shortneg4}$} & & \second{$32.03$} & \second{$\expnumber{6.87}{\shortneg3}$} & $34.43$ & $\expnumber{4.33}{\shortneg4}$ \\ 
PnP-ULA-NP (PMC-RED) & $8.75$ & $\expnumber{1.43}{\shortneg1}$ & $18.18$ & $\expnumber{1.67}{\shortneg2}$ & & $27.51$ & $\expnumber{2.35}{\shortneg3}$ & \second{$35.98$} & \second{$\expnumber{2.84}{\shortneg4}$} \\ 
\noalign{\vskip 0.5ex}
\hdashline\noalign{\vskip 0.5ex}
PMC-PnP (ours) & $8.75$ & $\expnumber{1.43}{\shortneg1}$ & $18.18$ & $\expnumber{1.67}{\shortneg2}$ & & $27.52$ & $\expnumber{2.35}{\shortneg3}$ & $35.99$ & $\expnumber{2.83}{\shortneg4}$\\ 
APMC-RED (ours) & $28.37$ & $\expnumber{1.67}{\shortneg3}$ & $34.52$ & $\expnumber{4.06}{\shortneg4}$ & & $32.67$ & $\expnumber{6.43}{\shortneg4}$ & $36.02$ & $\expnumber{2.82}{\shortneg4}$\\ 
APMC-PnP (ours) & \first{28.46} & \first{\expnumber{1.63}{\shortneg3}} & \first{34.55} & \first{\expnumber{3.77}{\shortneg4}} & & \first{32.71} & \first{\expnumber{6.25}{\shortneg4}} & \first{36.03} & \first{\expnumber{2.81}{\shortneg4}} \\
\bottomrule
\end{tabular*}
\label{Tab:Recon}
\end{table*}
\begin{figure*}[t!]
\centering
\includegraphics[width=0.95\textwidth]{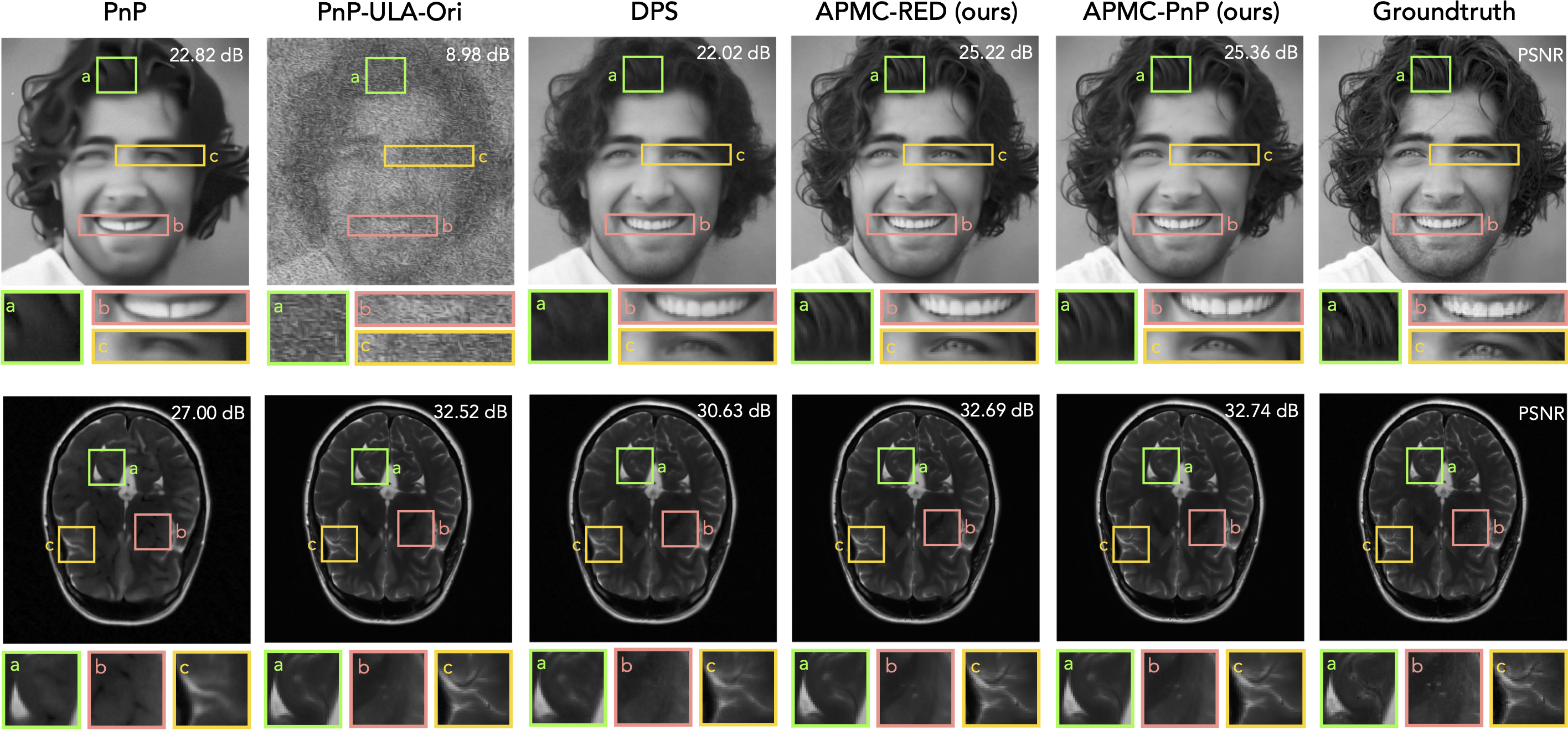}
\caption{
Visual comparison of the reconstructions obtained by APMC algorithms and baseline algorithms for $10\%$ CS (\emph{1st row}) and $8\times$ MRI (\emph{2nd row}) tasks.
The final images of the sampling algorithms are obtained by averaging $50$ image samples. 
The visual difference is highlighted in the zoom-in images.
Note how APMC algorithms restore the fine details.
}
\vspace{-10pt}
\label{Fig:Recon}
\end{figure*}
}

\afterpage{
\begin{table*}[t]
\centering
\tiny
\caption{Averaged NLL values obtained by the proposed PMC and baseline sampling algorithms for the CS ($m/n=0.1$ \& $m/n=0.3$) and MRI ($\text{Accel.}=8\times$ \& $\text{Accel.}=4\times$) tasks. The values of absolute error ($|\xbmbar-\xbm|$) and standard deviation (SD) that jointly determine NLL are also included. The best NLL values of ours and baselines are highlighted in \textbf{bold} and \underline{underline}, respectively.}
\begin{tabular*}{500pt}{L{50pt} C{20pt}C{30pt}C{20pt} C{20pt}C{30pt}C{20pt} C{0pt} C{20pt}C{30pt}C{20pt} C{20pt}C{30pt}C{20pt}} \toprule
\multirow{2}{*}{\textbf{Method}} & \multicolumn{3}{c}{$m/n=0.1$} & \multicolumn{3}{c}{$m/n=0.3$} & & \multicolumn{3}{c}{$\text{Accel.} = 8\times$} & \multicolumn{3}{c}{$\text{Accel.} = 4\times$}\\
\cmidrule{2-7} \cmidrule{9-14}
& NLL $\downarrow$ & $|\xbmbar-\xbm|$ $\downarrow$ & SD $\downarrow$ & NLL $\downarrow$ & $|\xbmbar-\xbm|$ $\downarrow$ & SD $\downarrow$ & & NLL $\downarrow$ & $|\xbmbar-\xbm|$ $\downarrow$ & SD $\downarrow$ & NLL $\downarrow$ & $|\xbmbar-\xbm|$ $\downarrow$ & SD $\downarrow$\\
\cmidrule{1-14}
DPS & \second{$\shortneg1.613$} & \second{$0.0433$} & \second{$0.0590$} & $\shortneg2.366$ & $0.0198$ & $0.0273$ & & \second{$\shortneg2.536$} & \second{$0.0172$} & \second{$0.0224$} & $\shortneg2.730$ & $0.0131$ & $0.0169$\\ 
PnP-ULA-Ori & $1.414$ & $0.3205$ & $0.1886$ & \second{$\shortneg2.474$} & \second{$0.0161$} & \second{$0.0249$} & & $\shortneg2.345$ & $0.0166$ & $0.0405$ & $\shortneg2.529$ & $0.0135$ & $0.0352$\\ 
PnP-ULA-NP (PMC-RED) & $1.366$ & $0.3122$ & $0.1873$ & $\shortneg0.622$ & $0.1006$ & $0.1083$ & & $\shortneg2.056$ & $0.0290$ & $0.0407$ & \second{$\shortneg2.883$} & \second{$0.0116$} & \second{$0.0166$}\\
\noalign{\vskip 0.5ex}
\hdashline\noalign{\vskip 0.5ex}
PMC-PnP (ours) & $1.366$ & $0.3121$ & $0.1874$ & $\shortneg0.620$ & $0.1007$ & $0.1085$ & & $\shortneg2.057$ & $0.0289$ & $0.0406$ & $\shortneg2.885$ & $0.0116$ & $0.0166$\\
APMC-RED (ours) & $\shortneg2.005$ & $0.0278$ & $0.0348$ & $\shortneg2.667$ & $0.0140$ & $0.0179$ & & $\shortneg2.689$ & $0.0152$ & $0.0207$ & $\shortneg2.884$ & $0.0115$ & $0.0165$\\
APMC-PnP (ours) & \first{\shortneg2.023} & \first{0.0275} & \first{0.0346} & \first{\shortneg2.674} & \first{0.0140} & \first{0.0177} & & \first{\shortneg2.690} & \first{0.0152} & \first{0.0206} & \first{\shortneg2.885} & \first{0.0115} & \first{0.0165}\\
\bottomrule
\end{tabular*}
\label{Tab:UQ}
\end{table*}
\begin{figure*}[t!]
\centering
\includegraphics[width=0.97\linewidth]{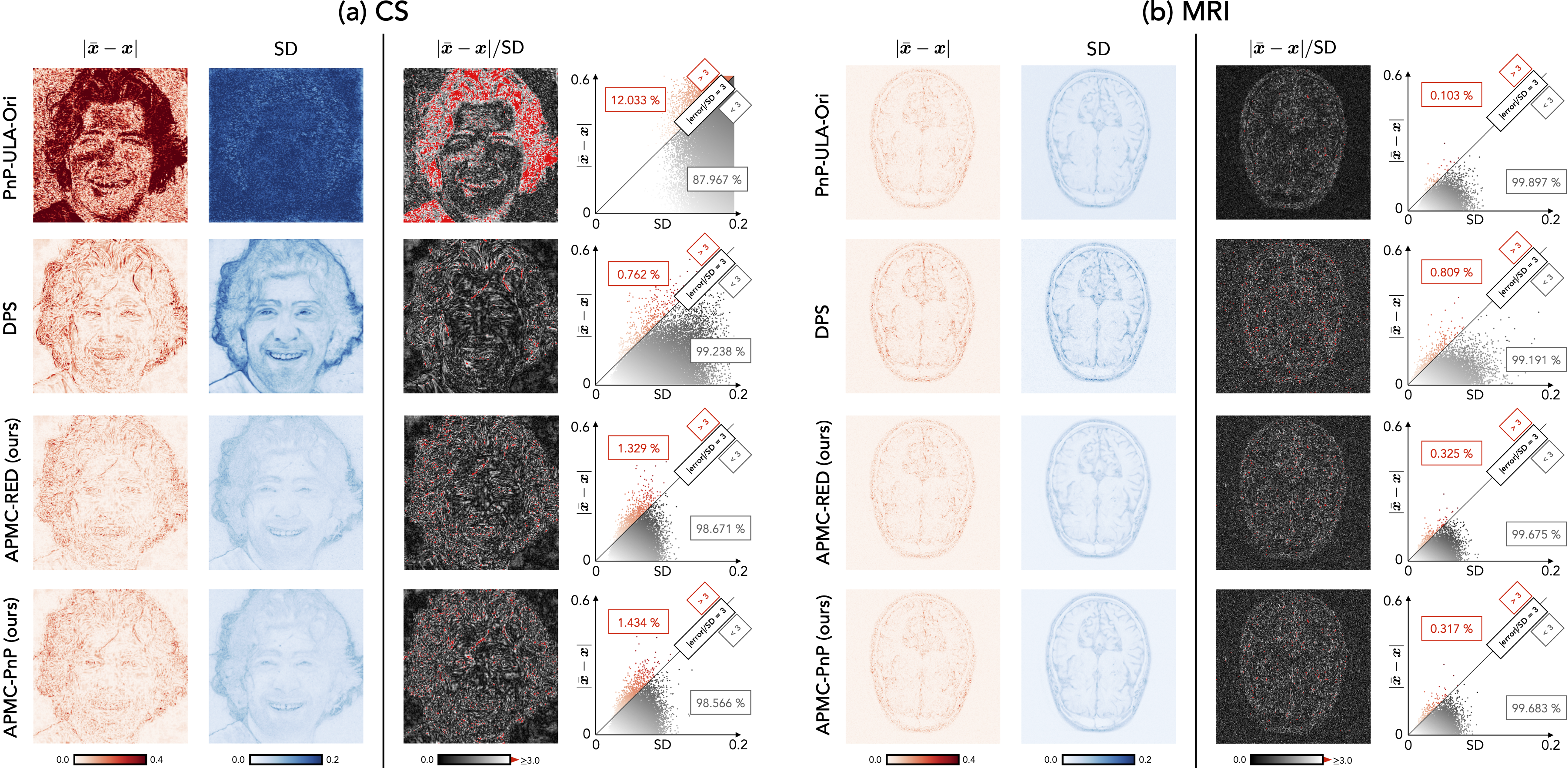}
\caption{
Visualization of the pixel-wise statistics associated with the CS and MRI reconstructions shown in Fig.~\ref{Fig:Recon}.
Figure (a) corresponds to CS, and figure (b) to MRI.
In each figure, the left columns plot the absolute error ($|\xbmbar-\xbm|$) and standard deviation (SD), and the right columns plot the 3-SD credible interval with the outlying pixels highlighted in red.
Note that APMC algorithms lead to a better UQ performance than the baselines by recovering an accurate mean and thus avoiding the need for an arbitrarily large SD.
}
\label{Fig:UQ}
\vspace{-10pt}
\end{figure*}
}

We evaluate the reconstruction quality by using the \emph{peak signal-to-noise ratio (PSNR)}, which is inversely related to the \emph{mean squared error (MSE)}.
We use PSNR as the criterion to finetune the algorithmic hyperparameters.
To quantitatively measure the quality of UQ, we compute the normalized \emph{negative log-likelihood (NLL)}~\cite{Lakshminarayanan.etal2017} of the groundtruth $\xbm$ 
under independent pixel-wise Gaussian distributions characterized by the sample mean $\xbmbar$ and standard deviation $\mathsf{SD}$
\begin{equation}
\mathsf{NLL}(\xbmbar, \xbm) = \frac{1}{n}\sum_{i=1}^{n}\frac{1}{2\mathsf{SD}^2_i}(\xbmbar_i-\xbm_i)^2 + \frac{1}{2}\log(2\pi \mathsf{SD}^2_i),
\end{equation}
where $i$ is the pixel index and $\mathsf{SD}_i$ the $i$th estimated standard deviation.
It follows that better UQ algorithms minimize NLL by producing a high-fidelity $\xbmbar$ and avoiding the need for an arbitrarily large $\mathsf{SD}$.
We additionally calculate the pixel-wise \emph{3-SD credible interval} to measure the coverage of the groundtruth.
If the distribution were truly Gaussian, $99\%$ of the ground-truth pixels should lie within the 3-SD interval; however, as our posterior distributions are not Gaussian with a score-based prior, we do not expect to reach $99\%$ coverage in our experiments.

\subsubsection{Compressed Sensing}
Two different CS ratios $(m/n)$ of $\{10\%, 30\%\}$ are considered in the experiments.
We train the score network using the FFHQ dataset~\cite{Karras.etal2019}, while randomly selecting the test images from the separate CelebA dataset~\cite{Liu.etal2015-CelebA}.

Table~\ref{Tab:Recon} ($m/n=0.1$ \& $m/n=0.3$) summarize the averaged PSNR and MSE values obtained by all algorithms.
First, the table clearly shows the substantial improvement led by APMC algorithms over the baselines in all settings.
In particular, when the compression is severe ($m/n=0.1$), APMC algorithms achieve around $28.40$ dB in averaged PSNR, outperforming the best baseline by $3.5$ dB.
The enhancement is further demonstrated in the visual comparison shown in Fig.~\ref{Fig:Recon} (\emph{1st row}).
We observe that APMC algorithms can faithfully recover fine details preserved in the groundtruth.  
In particular, note the hair, teeth, and eye shape highlighted in the zoom-in regions.
Second, PnP-ULA-Ori, PnP-ULA-NP, and PMC-PnP achieve poor results in the scenario of $m/n=0.1$ due to their slow convergence. 
We run all three algorithms on two test image for more iterations and fail to observe convergence within $150,000$ iterations.
This suggests that stationary Langevin algorithms may converge slowly for certain inverse problems. 
Equipped with weighted annealing, APMC algorithms generally converge within $4,000$ iterations, showing a significant acceleration in speed in this CS task.
On the other hand, PnP-ULA-Ori successfully converges within $10,000$ when the compression is moderate ($m/n=0.3$), while PnP-ULA-NP and PMC-PnP still do not achieve the same level of convergence.
This acceleration is also reflected in the enhanced PSNR and NLL values obtained by PnP-ULA-Ori.
We refer to Supplement~III for a detailed discussion on the speed of convergence.

Table~\ref{Tab:UQ} ($m/n=0.1$ \& $m/n=0.3$) summarize the averaged NLL values obtained by the APMC algorithms and sampling baselines.
We additionally summarize the averaged pixel-wise absolute error ($|\xbmbar-\xbm|$) and standard deviation (SD) as they are the two factors that jointly determine the final value of NLL.
Beyond high-quality reconstruction, the results show that APMC algorithms also achieve better UQ performance than the baselines, by simultaneously producing a high-fidelity mean and avoiding the need for a large SD.
Fig.~\ref{Fig:UQ}(a) visualizes the pixel-wise statistics associated with the reconstruction in Fig~\ref{Fig:Recon} (\emph{1st row}).
In the left columns, we plot the 3-SD credible interval where the outside pixels are highlighted in red.
Note that around $97\%$ of the pixels in the ground-truth image lie in the 3-SD interval from the reconstructions produced by APMC algorithms.
Despite DPS also achieving a high coverage ratio, it yields inaccurate $\xbmbar$ and thus large SD which leads to poor NLL performance.
In the scenario of $m/n=0.3$, we observe that PnP-ULA-Ori outperforms DPS and achieves the best NLL values among the baselines. 
However, we note that PnP-ULA-Ori generates samples with larger SD, thus achieving worse NLL values than AMPC algorithms.

\subsubsection{Magnetic resonance imaging}
In this task, we consider the radial subsampling mask corresponding to $\{4\times, 8\times\}$ acceleration.
We use the FastMRI dataset~\cite{Zbontar.etal2018} to train the score network and to form the testing dataset.

Table~\ref{Tab:Recon} and~\ref{Tab:UQ} ($\text{Accel.}=8\times$ \& $\text{Accel.}=4\times$) summarize the numerical performance of all algorithms in image reconstruction and UQ, respectively.
Overall, the APMC algorithms still yield the best PSNR and NLL values, showing consistent performance across different inverse problems.
On the other hand, we observe that PnP-ULA-Ori converges in both acceleration settings,
which is corroborated in its close numerical performance to APMC algorithms.
Fig.~\ref{Fig:Recon} (\emph{2nd row}) provides a visual comparison of the reconstructed images under $8\times$ acceleration, and Fig.~\ref{Fig:UQ}(b) visualizes the associated pixel-wise statistics.
Note how APMC algorithms restore the fine features in different brain areas and cover about $99\%$ of the ground-truth pixels with a narrower 3-SD credible interval.
Also note the close performance of PnP-ULA-Ori for this example image.
Furthermore, PnP-ULA-NP (PMC-RED) and PMC-PnP also converge under $4\times$ acceleration, achieving numerical performance on par with their APMC variants. This observation shows that the stationary PMC algorithms are able to sample from the posterior at the expense of more iterations.

We additionally compare APMC algorithms with the state-of-the-art end-to-end VarNet~\cite{Sriram.etal2020} in terms of the reconstruction quality.
By just leveraging the model-agnostic priors, both APMC algorithms yield results on par with those of VarNet that requires model-specific training. 
We refer to Supplement III for a detailed discussion.

\begin{figure}[t!]
\centering
\includegraphics[width=0.5\linewidth]{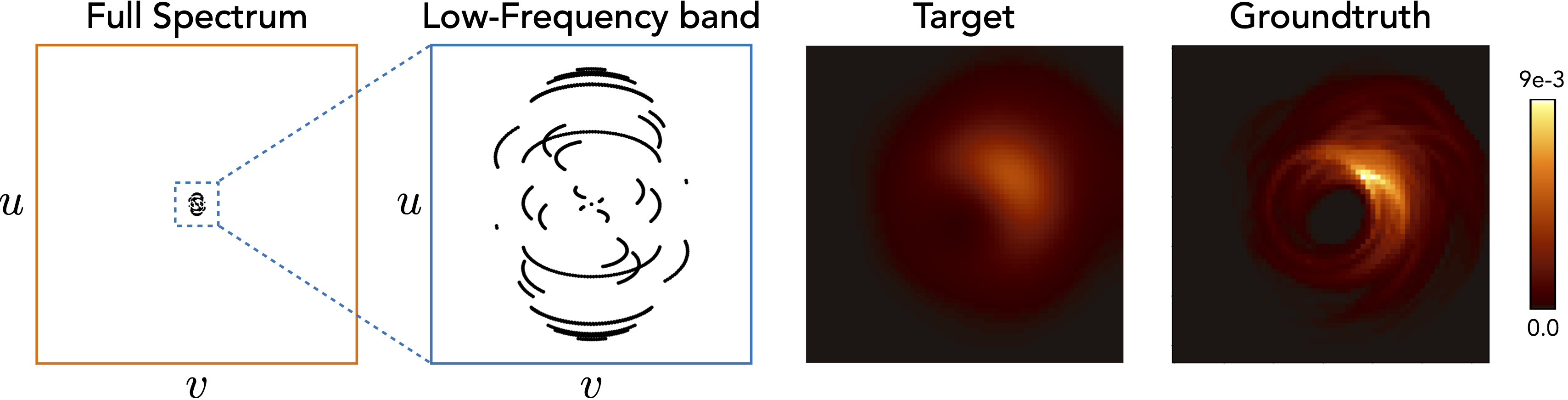}
\caption{
Visual illustration of BHI.
The left two images together demonstrate the subsampling pattern in the Fourier spectrum. The \emph{Groundtruth} image shows the ground-truth black hole simulation image used in this experiment. The \emph{Target} image corresponds to the scenario where the low-frequency band is fully sampled, resembling a single-dish telescope the size of the Earth. This target image represents the intrinsic resolution of our telescope; 
an effort to recover sharper features would be classified as attempting superresolution.
}
\label{Fig:BHI_Setup}
\vspace{-10pt}
\end{figure}

\begin{figure}[t!]
\centering
\includegraphics[width=0.5\linewidth]{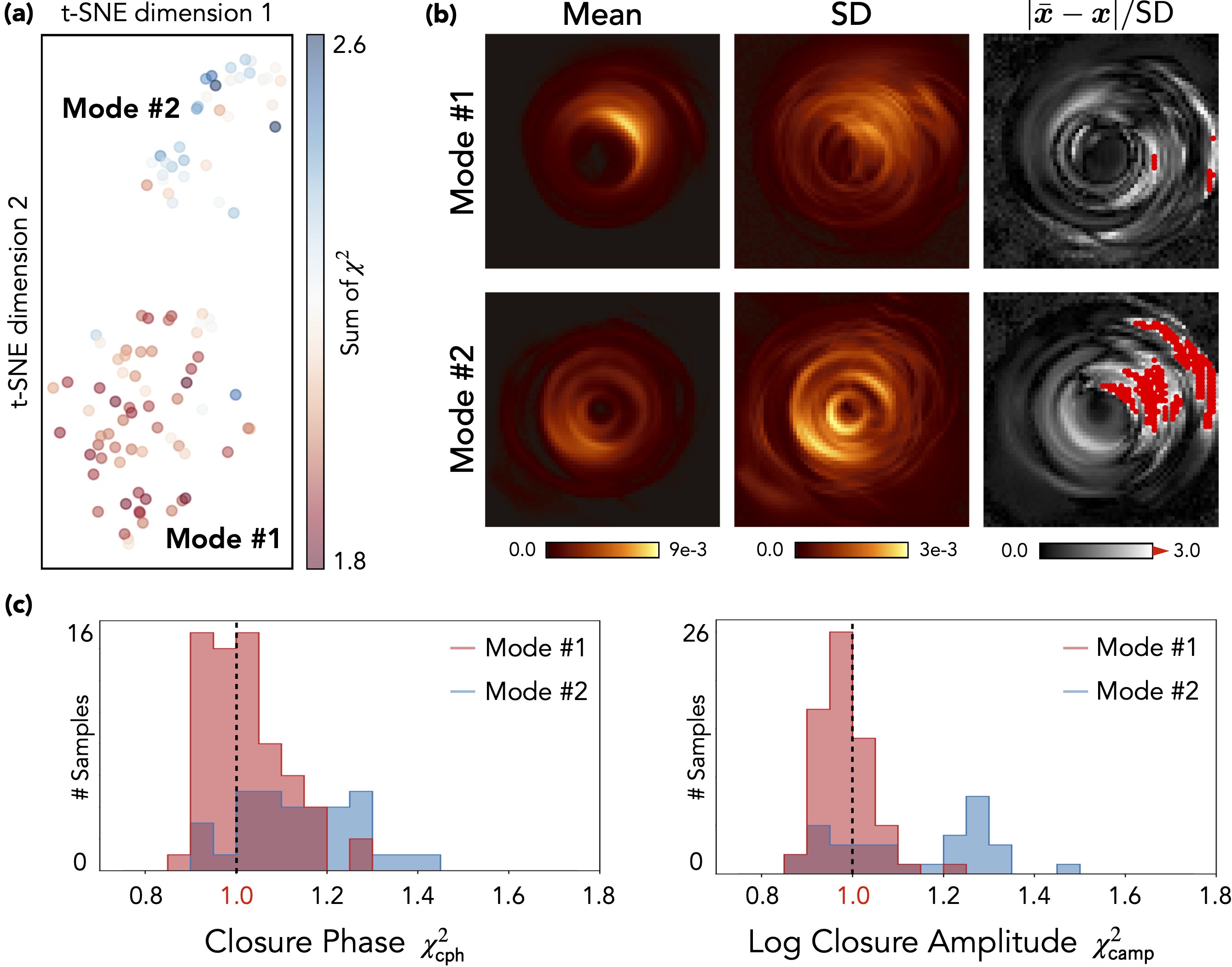}
\caption{
Visualization of the sampling results obtained by APMC-PnP. 
In total 100 samples were drawn. (a) The t-SNE plot (perplexity=20) shows the distribution of the samples. Note that this t-SNE plot shows there are two distinct image modes.
(b) Pixel-wise statistics of each mode. 
(c) The distribution of the closure phase $\chi^2_\mathsf{cph}$ and log closure amplitude $\chi^2_\mathsf{camp}$ statistics for each mode. 
Note that APMC-PnP successfully recovers the two modes of the posterior distribution, with both modes resulting in $\chi^2$ statistics close to 1. 
}
\label{Fig:BHI_PMCPnP}
\vspace{-10pt}
\end{figure}

\subsection{Nonlinear black-hole interferometric imaging}

We now validate APMC algorithms on the nonlinear BHI task.
The imaging system of BHI can be mathematically characterized by the van Cittert-Zernike theorem, which links each Fourier component, or so-called \emph{visibility}, of the black-hole image to the coherence measured by a pair of telescopes.
The measurement equation for each visibility is given by
\begin{equation}
V_{a,b}^t = g_a^t g_b^t \cdot e^{-i(\phi_a^t-\phi_b^t)} \cdot \tilde{\Ibm}_{a,b}^t(\xbm) + \eta_{a,b}
\end{equation}
where $a$ and $b$ index the telescopes, $t$ represents time, and $\tilde{\Ibm}_{a,b}^t(\xbm)$ is the ideal Fourier component of image $\xbmast$ corresponding to the baseline between telescopes $a$ and $b$ at time $t$.
In practice, the visibility is corrupted by three types of noise: telescope-based gain error $g_a$ and $g_b$, telescope-based phase error $\phi_a^t$ and $\phi_a^t$, and baseline-based AWGN $\eta_{a,b}$.
The first two types of noise are usually caused by atmospheric turbulence and instrument miscalibration, while the last one corresponds to thermal noise.
To mitigate the gain and phase error, multiple noisy visibilities are combined into noise-canceling data products termed \emph{closure phases} and \emph{log closure amplitudes}
\begin{align}
&\ybm^\mathsf{cph}_{t,(a,b,c)} = \angle(V_{a,b}V_{b,c}V_{a,c}) \defn \Abm^\mathsf{cph}_{t,(a,b,c)}(\xbm) \\
&\ybm^\mathsf{camp}_{t,(a,b,c,d)} = \log\left( \frac{|V_{a,b}^t| |V_{c,d}^t|}{|V_{a,c}| |V_{b,d}^t|} \right) \defn \Abm^\mathsf{camp}_{t,(a,b,c,d)}(\xbm)
\end{align}
where $\angle$ computes the angle of a complex number.
Given total $M>0$ telescopes, the number of combined measurements $\ybm^{\mathsf{cph}}_t$ and $\ybm^{\mathsf{camp}}_t$ at time $t$ are given by $\frac{(M-1)(M-2)}{2}$ and $\frac{M(M-3)}{2}$, respectively, after excluding repetitive measurements.
Here, we adopt a 9-telescope array ($M=9$) consisting of telescopes that currently participate in the \emph{Event Horizon Telescope (EHT)}.
In summary, the forward problem of BHI is formulated as
\begin{align}
\begin{cases}
\ybm^\mathsf{cph}_{t} = \Abm^\mathsf{cph}_{t}(\xbm) \\
\ybm^\mathsf{camp}_{t} = \Abm^\mathsf{camp}_{t}(\xbm)
\end{cases}
\hspace{-10pt}
\text{where}\;
\begin{cases}
\Abm^\mathsf{cph}_{t} = ( \Abm^\mathsf{cph}_{t, \mathsf{c}_1}, \cdots, \Abm^\mathsf{cph}_{t, \mathsf{c}_{k_1}} )^T \\
\Abm^\mathsf{camp}_{t} = ( \Abm^\mathsf{camp}_{t, \mathsf{c}_1}, \cdots, \Abm^\mathsf{camp}_{t, \mathsf{c}_{k_2}} )^T
\end{cases}\nonumber
\end{align}
Note that the inverse problem of BHI is severely ill-posed: even if we have an Earth-size telescope, the high-frequency visibilities are still immeasurable; in practice, the low-frequency band is further subsampled.
A visual illustration of BHI is provided in Fig.~\ref{Fig:BHI_Setup}.
Additionally, the gain and phase errors result in significant information loss; for instance, the absolute phase of the image can never be recovered and the total flux (i.e. summation of the pixel values) of the image is not constrained by either of the closure quantities. Since this is the case we include an additional constraint in the likelihood to constrain the total flux:
\begin{equation}
\label{Eq:bhi_likelihood}
\ell(\ybm|\xbm) = 
\underbrace{\sum_{t,\mathsf{c}} \frac{\|\Abm^\mathsf{cph}_{t, \mathsf{c}}(\xbm) - \ybm^\mathsf{cph}_{t,\mathsf{c}}\|^2_2}{2\beta_\mathsf{cph}^2}}_{\chi^2_\mathsf{cph}} + 
\underbrace{\sum_{t,\mathsf{c}}\frac{\|\Abm^\mathsf{camp}_{t, \mathsf{c}}(\xbm) - \ybm^\mathsf{camp}_{t,\mathsf{c}}\|_2^2}{2\beta_\mathsf{camp}^2}}_{\chi^2_\mathsf{camp}} +\rho\frac{\left\|\sum_{i}\xbm_i-\ybm^\mathsf{flux}\right\|_2^2}{2},
\end{equation}
where $\beta_\mathsf{cph}$ and $\beta_\mathsf{camp}$ are assumed to be known, $\ybm^\mathsf{flux}$ can be accurately measured, and $\rho>0$ is an optimization parameter.
The first two terms in~\eqref{Eq:bhi_likelihood} are referred to as the $\chi^2$ errors.

\paragraph{Experimental setup}
We test our APMC algorithms on a synthetic BHI problem that has previously been shown to lead to a bimodal posterior distribution~\cite{SunHe.2021}.
The ground-truth test image is shown in Fig.~\ref{Fig:BHI_Setup}.
One of the primary objectives of this experiment is to determine if the algorithms can successfully reconstruct this bimodal distribution.
We additionally assess the fidelity of the reconstruction using the $\chi^2$ errors.
Note that $\chi^2_\mathsf{cph}=\chi^2_\mathsf{camp}= 1$ indicates that the measurements and prior are ideally balanced.

We use the GRMHD dataset~\cite{Wong.etal2022} to train the score network.
All images used in the experiment are resized to $64\times64$ pixels.
We observe that APMC-PnP and APMC-RED yield identical results for BHI as the step-size allowed for convergence is small ($10^{-6}$); hence, we restrict our discussion to the results of APMC-PnP for brevity.
We include DPI with a score-based prior (scoreDPI)~\cite{Feng.etal2023scorebased,Feng.etal2023} as the baseline method.
We run each algorithm to draw $100$ samples for computing the final results.
Note that scoreDPI is a variational method and needs to be retrained for different test cases.
Additionally, due to computational and memory constraints, an approximation of the prior log-density must be used in scoreDPI in order to  handle $64\times64$ images.

\begin{figure}[t!]
\centering
\includegraphics[width=0.5\linewidth]{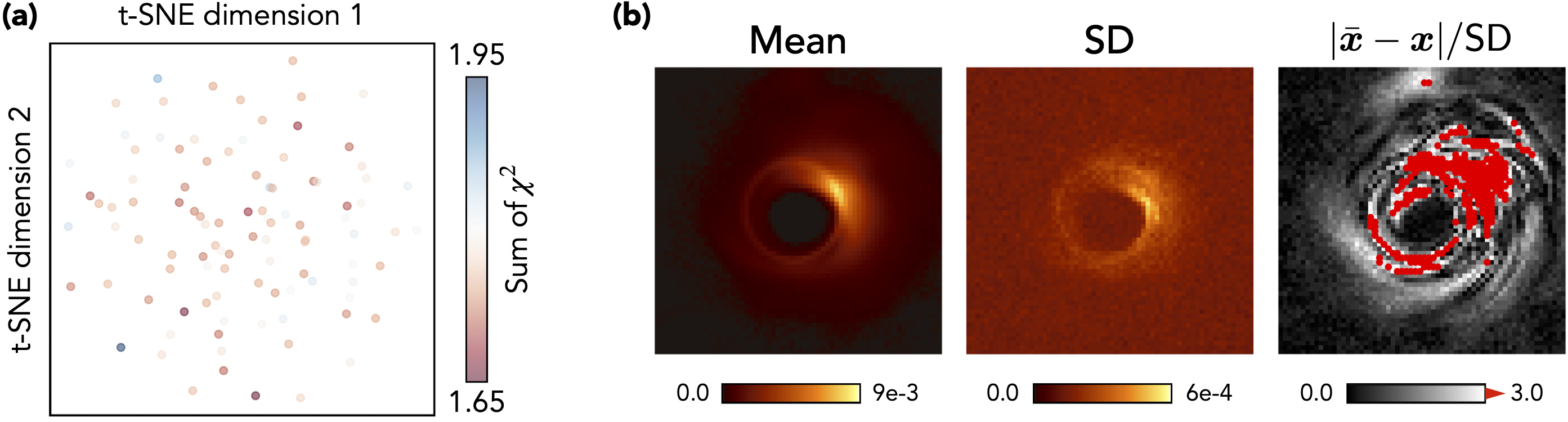}
\caption{
Visualization of the sampling results obtained by scoreDPI.
(a) shows the t-SNE plot (perpelxity=20), and (b) plots the sample statistics.
Note that scoreDPI recovers a single-mode distribution rather than a bimodal distribution.
}
\vspace{-10pt}
\label{Fig:BHI_ScoreDPI}
\end{figure}

\paragraph{Results}
Fig.~\ref{Fig:BHI_PMCPnP} visualizes the statistics of the samples generated by APMC-PnP.
We use the t-distributed stochastic neighbor embedding (t-SNE) to plot the distribution of the samples.
As shown in Fig.~\ref{Fig:BHI_PMCPnP}(a), the samples can be classified into two modes: \emph{Mode \#1} is composed of the black holes with their flux concentrating in the top-right corner (as is also the case in the groundtruth), while the black holes from \emph{Mode \#2} have the flux concentration at the bottom-left corner.
Our result is aligned with the previous experiment that recovered two similar modes (see Fig. 6 in~\cite{SunHe.2021}).
The uncertainty of each mode is shown in Fig.~\ref{Fig:BHI_PMCPnP}(b).
Note the almost full coverage of the groundtruth by the 3-SD credible interval of \emph{Mode \#1}.
Fig.~\ref{Fig:BHI_PMCPnP}(c) illustrates the data-fidelity for each sample by plotting the histogram of the $\chi^2_\mathsf{cph}$ and $\chi^2_\mathsf{camp}$ errors.
We highlight that the values obtained by both modes are distributed close to 1.
Fig.~\ref{Fig:BHI_ScoreDPI} visualizes the statistics for scoreDPI.
The t-SNE plot shows that scoreDPI recovers a single-mode distribution rather than a bimodal distribution.
The inability of scoreDPI to recover the bimodal distribution could be caused by a few factors that potentially lead to inaccurate posterior estimation: \emph{1)} the posterior is approximated by a normalizing flow network that likely struggles to model the rich posterior of $64\times 64$ pixel images, and \emph{2)} to handle this image size an approximation to the prior log-density was used. 
In contrast, our method does not suffer from these factors and can handle large image sizes without making additional approximations.
We also note that there is less coverage of the ground-truth pixels within 3-SD interval of scoreDPI when compared with \emph{Mode \#1}'s result obtained by APMC-PnP.

\section{Conclusion}
In this paper, we develop PMC as a principled posterior sampling framework for solving general imaging inverse problems.
PMC jointly leverages the expressive score-based generative priors and physical constraints while also enabling the UQ of the reconstructed image via posterior sampling.
In particular, we introduce two PMC algorithms which can be backward-related to the traditional PnP and RED algorithms.
A comprehensive convergence analysis for both stationary and annealed variants of the PMC algorithms is presented.
Our results show that all algorithms converge at the rate of $O(1/N)$.
Related experiments are also presented to empirically confirm the proposed theorems and to elucidate the capability of PMC in various representative inverse problems, including a non-convex problem that leads to a bimodal posterior distribution.

\bibliographystyle{ieeetr}


\newpage
\beginsupplement

{\LARGE\noindent \textbf{Supplementary Material}}

\section{Technical Proofs}
In this section, we present the technical proofs.
The general idea of our proof is inspired by the interpolation technique used in~\cite{Vempala.etal2019, Balasubramanian.etal2022}.
In order to analyze the convergence of a PMC algorithm (which is discrete in time), we follow~\cite{Balasubramanian.etal2022} to ``linearly" interpolate the iterations of the algorithm, forming a corresponding continuous-time diffusion process.
By doing so, we are able to use tools in continuous mathematics such as differential inequalities to study the change of the Kullback–Leibler (KL) divergence along the interpolated diffusion process. 
Indeed, the interpolated diffusion process can be understood as a ``piece-wise constant'' gradient flow of the KL divergence in the space of probability distributions. 
Therefore, we can follow the practice in optimization to derive the stationary-distribution convergence. 
Here, the stationary condition will be given in terms of the relative Fisher information (FI), which measures the norm of the ``gradient'' of KL divergence under the Wasserstein metric.

This section is organized as follows. In Section \ref{Sec:Lemmas}, we present two lemmas that we adapted from \cite{Balasubramanian.etal2022}; we also introduce some simpler proof of these lemmas. With the two lemmas, we prove the main theorem of our paper in Sections \ref{Sec:Proof1}, \ref{Sec:Proof2}, \ref{Sec:Proof3}, and \ref{Sec:Proof4}.

\medskip\noindent
\textbf{Notations\quad}
Throughout the proof, we consider the probability space $(\Omega,\Fscr,\mathbb{P})$, where $\Omega$ denotes the sample space, $\Fscr$ the $\sigma$-algebra, and $\mathbb{P}$ the probability measure. 
For the random variable $\zeta: \Omega \rightarrow \R^n$, we denote its expectation by $\E[\zeta] = \int_\Omega\zeta(\omega)\mathbb{P}(d\omega)$.

We recall that the posterior distribution that we are interested in takes the form $\pi(\xbm|\ybm) \propto \ell(\ybm|\xbm) p(\xbm)$. We define $g(\xbm) = -\log \ell (\ybm|\xbm)$ .
For ease of notation, we will omit the presence of $\ybm$ and simply write $\pi(\xbm|\ybm)$ by $\pi(\xbm)$ in our proof. As a reminder, we have the following relations:
\begin{equation}
\pi\propto e^{-f}, \quad \nabla f = -\nabla \log \pi = \nabla g - \nabla \log p \, .\nonumber
\end{equation}
Here, we assume $\nabla f \in C^1$ is differentiable and has a continuous derivative. 
Note $\nabla f$ is also Lipschitz continuous according to Assumption 1 and 2.
We will use the above notations and relations constantly in the proof. For reader's convenience, we recall the definition of the KL divergence for two probability densities $\nu$ and $\pi$ is
$$\KL\left( \nu \,\|\, \pi \right) 
= \int_{\R^n} \nu(\xbm) \log \frac{\nu(\xbm)}{\pi(\xbm)} \,d\xbm\, ,$$
and the FI is
\begin{align}
\FI\left( \nu \,\|\, \pi \right) 
&= \int_{\R^n} \left\|\nabla \log\frac{\nu(\xbm)}{\pi(\xbm)}\right\|^2_2 \nu(\xbm)\,d\xbm \nonumber\\
&= \int_{\R^n} \big\|\nabla \log \nu(\xbm) - \nabla \log \pi(\xbm) \big\|^2_2 \, \nu(\xbm)\, d\xbm . \nonumber
\end{align}

\subsection{Lemmas}
\label{Sec:Lemmas}
The first lemma concerns the density evolution of an interpolated diffusion process.
\begin{lemma}[\cite{Balasubramanian.etal2022}]
\label{Le:KL}
Consider the stochastic process defined by
\begin{equation}\tag{L1.1}
\label{Eq:L1.1}
\Xt := \xbm_0 - t w_0+ \sqrt{2}\Bt, \quad\text{with}\quad w_0 = w(\Xzero), \;\Xzero\sim\nu_0 
\end{equation}
where $w_0$ is integrable and $\{\Bt\}_{t\geq0}$ is a standard Brownian motion in $\R^n$ which is independent of $(\xbm_0, w_0)$. Then, writing $\nu_t$ for the probability density of $\Xt$, we have
\begin{align}\tag{L1.2}
\label{Eq:L1.2}
\frac{d}{d t}\KL(\nu_t \,\|\, \pi) \leq -\frac{3}{4}\FI(\nu_t \,\|\, \pi) + \E\big[ \left\| \nabla f(\xbm_t) - w_0 \right\|^2_2 \big],
\end{align}
where we recall that $\pi\propto e^{-f}$, and the expectation in the last term is taken over $x_0 \sim \nu_0$ and $x_t \sim \nu_t$.
\end{lemma}

\begin{proof}
This lemma is based on Lemma 12 in~\cite{Balasubramanian.etal2022}, which is adapted from~\cite{Vempala.etal2019}. For readers' convenience, we also present 
a detailed proof here.
The main idea is to derive the evolution of the density of $\Xt$, and then plug it into the the derivative formula for the KL divergence.

\vspace{-0.5em}
\paragraph{Step 1: Deriving the density evolution equation}
For each $t>0$, let $\nuj$ denote the density of the joint distribution of $(\Xt, \Xzero)$. Let $\nutz$ be the conditional distribution of $\Xt$ conditioned on $\Xzero$, and $\nuzt$ be the conditional distribution of $\Xzero$ conditioned on $\Xt$. We have the relation
\begin{equation}
\nuj({\bm x}, \Xzero) = \nutz({\bm x} | \Xzero)\nu_0(\Xzero) = \nuzt(\Xzero | {\bm x})\nu_t({\bm x})\, .
\end{equation}
Conditioning on $\Xzero$, we have that $w_0$ is a constant vector. Then, the conditional distribution $\nutz$ evolves according to the following Fokker-Planck equation, which is also known as Kolmogorov Forward equation:
\begin{align}
\label{Eq:CondFP}
\frac{\partial}{\partial t} \nutz({\bm x} | \Xzero) &= \div\left( \nutz({\bm x} | \Xzero)w_0 \right) + \Delta \nutz({\bm x} | \Xzero) \nonumber\\
& = \div \left( \nutz({\bm x} | \Xzero)w_0 + \nabla\nutz({\bm x} | \Xzero) \right),
\end{align}
where $\Delta=\div(\nabla\cdot)$ denotes the Laplace operator. Now, our goal is to derive the evolution equation for the marginal distribution $\nut(\xbm)$. To achieve so, we need to take the expectation over $\Xzero\sim\nuz$. Multiplying both sides of \eqref{Eq:CondFP} by $\nuz(\Xzero)$ and integrating over $\Xzero$, we have
\begin{align}
\label{Eq:FP}
\frac{\partial}{\partial t} \nu_{t}({\bm x}) &= \int_{\R^n} \bigg(\frac{\partial}{\partial t}\nutz({\bm x} | \Xzero)\bigg) \nu_0(\Xzero) \,d\xbm_0 \nonumber\\
& = \int_{\R^n} \div\left( \nutz({\bm x} | \Xzero)w_0 + \nabla\nutz({\bm x} | \Xzero) \right) \nu_0(\Xzero) \,d\xbm_0  \nonumber\\
& = \int_{\R^n} \div\left( \nuj({\bm x}, \Xzero)w_0 + \nabla\nuj({\bm x}, \Xzero) \right) \,d\xbm_0  \nonumber\\
& = \int_{\R^n} \div\left( \nuj({\bm x}, \Xzero)w_0 \right) \,d\xbm_0 + \int_{\R^n} \Delta \nuj({\bm x}, \Xzero) \,d\xbm_0  \nonumber\\
& = \div\left( \nu_t({\bm x}) \int_{\R^n}  \nuzt(\Xzero | {\bm x})w_0 \,d\xbm_0 \right) + \Delta \nut({\bm x})  \nonumber\\
& = \div\big( \nu_t({\bm x}) \E_\nuzt\left[ w_0 | \Xt=\xbm \right] \big) + \Delta \nut(\Xt)\, ,
\end{align}
where in the last two eqaulities, we use the definition of conditional distributions. As a consequence, we obtain the evolution equation for $\nu_t$.

\vspace{-0.5em}
\paragraph{Step 2: Calculating the derivatives of the KL divergence} The time derivative of the KL divergence with respect to $\pi$ is given by
\begin{align}
\frac{d}{d t}\KL\left( \nut \,\|\, \pi \right) 
&= \frac{d}{dt} \int_{\R^n} \nut(\xbm) \log \frac{\nut(\xbm)}{\pi(\xbm)} \,d\xbm \nonumber \\
&= \int_{\R^n} \frac{\partial \nut(\xbm)}{\partial t} \log \frac{\nut(\xbm)}{\pi(\xbm)} \,d\xbm
	\;+\; \int_{\R^n} \nut(\xbm) \frac{\partial}{\partial t}\log \frac{\nut(\xbm)}{\pi(\xbm)} \,d\xbm \nonumber \\
&= \int_{\R^n} \frac{\partial \nut(\xbm)}{\partial t} \log \frac{\nut(\xbm)}{\pi(\xbm)} \,d\xbm
	\;+\; \int_{\R^n} \nut(\xbm) \, \frac{\pi(\xbm)}{\nut(\xbm)} \, \frac{1}{\pi(\xbm)} \, \frac{\partial\nut(\xbm)}{\partial t} \,d\xbm \nonumber\\
&= \int_{\R^n} \frac{\partial \nut(\xbm)}{\partial t} \log \frac{\nut(\xbm)}{\pi(\xbm)} \,d\xbm
	\;+\; \int_{\R^n} \frac{\partial\nut(\xbm)}{\partial t} \,d\xbm \nonumber\\
&= \int_{\R^n} \frac{\partial \nut(\xbm)}{\partial t} \log \frac{\nut(\xbm)}{\pi(\xbm)} \,d\xbm
	\;+\; \frac{d}{dt} \int_{\R^n} \nut(\xbm) \,d\xbm \nonumber\\
&= \int_{\R^n} \frac{\partial \nut(\xbm)}{\partial t} \log \frac{\nut(\xbm)}{\pi(\xbm)} \,d\xbm\, . \nonumber
\end{align}
By using~\eqref{Eq:FP}, we can derive
\begin{align}
\label{Eq:KL}
\frac{d}{dt}\KL\left( \nut \,\|\, \pi \right)
&= \int_{\R^n} \frac{\partial \nut(\xbm)}{\partial t} \log \frac{\nut(\xbm)}{\pi(\xbm)} \,d\xbm, \nonumber\\
&= \int_{\R^n} \Big( \div\big( \nu_t(\xbm) \E_\nuzt\left[ w_0 | \Xt=\xbm \right] \big) + \Delta \nut(\xbm) \Big) \log \frac{\nut(\xbm)}{\pi(\xbm)} d\xbm \nonumber\\
&= \int_{\R^n} \Big( \div\big( \nu_t(\xbm) \E_\nuzt\left[ w_0 | \Xt=\xbm \right] + \nabla \nut(\xbm) \big) \Big) \log \frac{\nut(\xbm)}{\pi(\xbm)} d\xbm \nonumber\\
&= -\int_{\R^n} \Big\langle \nu_t(\xbm) \E_\nuzt\left[ w_0 | \Xt=\xbm \right] + \nabla \nut(\xbm), \nabla\log \frac{\nut(\xbm)}{\pi(\xbm)} \Big\rangle d\xbm \nonumber\\
& \comment{Apply chain rule and $\nabla\log\pi(\xbm) = -\nabla f(\xbm)$} \nonumber\\
&= -\int_{\R^n} \Big\langle \nu_t(\xbm) \left( \E_\nuzt\left[ w_0 | \Xt=\xbm \right] + \nabla \log \frac{\nut(\xbm)}{\pi(\xbm)} - \nabla f(\xbm)\right), \nabla\log \frac{\nut(\xbm)}{\pi(\xbm)} \Big\rangle d\xbm \nonumber\\
&= -\FI(\nut \,\|\, \pi) + \int_{\R^n} \Big\langle \nabla f(\xbm) - \E_\nuzt\left[ w_0 | \Xt=\xbm \right], \nabla\log \frac{\nut(\xbm)}{\pi(\xbm)} \Big\rangle \nu_t(\xbm) d\xbm\, .
\end{align}
To handle the second term, we apply Peter-Paul inequality ($ab \leq a^2 / (2 \epsilon) + \epsilon b^2 / 2$ for some $\epsilon>0$)
\begin{align}
\label{Eq:Young}
&\int_{\R^n} \Big\langle \nabla f(\xbm) - \E_\nuzt\left[ w_0 | \Xt=\xbm \right], \nabla\log \frac{\nut(\xbm)}{\pi(\xbm)} \Big\rangle \nu_t(\xbm) d\xbm \nonumber\\
&\leq \frac{1}{4}\FI(\nut \,\|\, \pi) \;+\; \E \big[ \left\|\nabla f(\Xt) - \E_\nuzt\left[ w_0 | \Xt \right]\right\|^2_2 \big] \nonumber\\
&\leq \frac{1}{4}\FI(\nut \,\|\, \pi) \;+\; \E \big[ \left\|\nabla f(\Xt) - w_0 \right\|^2_2\big]\, ,
\end{align}
where in the last inequality, we use the property of conditional expectations.

Plugging \eqref{Eq:Young} into \eqref{Eq:KL} and rearranging terms yields the result \eqref{Eq:L1.2}. The proof is complete.
\end{proof}

The second lemma concerns the bound on the Fisher information. One can prove this lemma based on Lemma 16 from~\cite{Chewi22a}. Here, we also present a simpler, direct proof without introducing the infinitesimal generator of the Langevin diffusion.
\begin{lemma}
\label{Le:Bound}
Let $\Fcal: \R^n \rightarrow \R^n$ denote a measurable function. Assume $\nabla f$ is Lipschitz continuous with a Lipschitz constant $\Lf>0$. For any probability density $\nu$, it holds that
\begin{equation}\tag{L2}
\E_\nu\left[ \| \Fcal(\xbm) \|^2_2\right] \leq 2\FI(\nu \,\|\, \pi) + 4n\Lf + 2\E_\nu \big[ \|\Fcal(\xbm) - \nabla f(\xbm)\|^2_2 \big]\, ,
\end{equation}
where we recall that $\pi\propto e^{-f}$.
\end{lemma}

\begin{proof}
First, note that $\Fcal(\xbm) = \Fcal(\xbm) - \nabla f(\xbm) + \nabla f(\xbm)$. By invoking the convexity of the $\ell$-2 norm, we have
\begin{equation}
\label{Eq:Prep}
\E_\nu \big[ \|\Fcal(\xbm)\|^2_2 \big] \leq 2\E_\nu \big[ \|\nabla f(\xbm)\|^2_2 \big] + 2\E_\nu \big[ \|\Fcal(\xbm) - \nabla f(\xbm)\|^2_2 \big].
\end{equation}
Recall the fact that $\nabla f = -\nabla \log \pi$. We can derive
\begin{equation}
    \begin{aligned}
        \E_\nu \big[ \|\nabla f(\xbm)\|^2_2 \big] &= \E_\nu \big[ \|\nabla \log \pi(\xbm) - \nabla \log \nu(\xbm) + \nabla \log \nu(\xbm)\|^2_2 \big]\\
        & = \E_\nu \big[ \|\nabla \log \pi(\xbm) - \nabla \log \nu(\xbm)\|^2_2 +2\langle \nabla \log\pi(\xbm) - \nabla \log \nu(\xbm), \nabla \log \nu(\xbm) \rangle + \|\nabla \log \nu(\xbm)\|^2_2 \big]\\
        & = \E_\nu \big[ \|\nabla \log \pi(\xbm) - \nabla \log \nu(\xbm)\|^2_2 + \langle2\nabla \log\pi(\xbm) - \nabla \log \nu(\xbm), \nabla \log \nu(\xbm)\rangle\big] \\
        & \leq \FI(\nu \,\|\, \pi) + 2 \E_\nu \big[\langle\nabla \log\pi(\xbm), \nabla \log \nu(\xbm)\rangle\big]\, .
    \end{aligned}
\end{equation}
Here, the first three identities are obtained by algebraic manipulations. 
In the last inequality, we used the definition of the Fisher information and the fact that $\E_{\nu} \big[\|\nabla \log \nu(\xbm)\|^2_2\big] \geq 0$.
To proceed, we note that 
\begin{equation}
    \begin{aligned}
        \E_\nu \big[\langle\nabla \log\pi(\xbm), \nabla \log \nu(\xbm)\rangle\big]
        = & \int_{\mathbb{R}^n} -\langle\nabla f(\xbm), \nabla \nu(\xbm) \rangle d\xbm \\
        = & \int_{\mathbb{R}^n} \big(-\div (\nabla f(\xbm) \nu(\xbm)) + \Delta f(\xbm) \nu(\xbm)\big) d\xbm\\
        = &\int_{\mathbb{R}^n} \Delta f(\xbm) \nu(\xbm) d\xbm \leq n\Lf\, .
    \end{aligned}
\end{equation}
In the first identity, we use the relation $\nabla \log\pi(\xbm) = -f(\xbm)$ and $\nu(\xbm) \nabla \log \nu(\xbm) = \nabla \nu(\xbm)$. In the second identity, we use the equality $\div (\nabla f(\xbm) \nu(\xbm)) = \Delta f(\xbm) \nu(\xbm) + \langle\nabla f(\xbm), \nabla \nu(\xbm) \rangle$ and the fact that
$$\int_{\mathbb{R}^n}\div (\nabla f(\xbm) \nu(\xbm)) d\xbm = 0,$$
which is due to the Divergence theorem.
In the last inequality, we use the fact that $|\Delta f|\leq n\Lf$. This is because that $\nabla f \in C^1$ is Lipschitz continuous with a constant $\Lf>0$, so the eigenvalues of the Hessian of $f$ is bounded by $\Lf$. Then by definition, $\Delta f$ is the summation of the diagonals of the Hessian, so its absolute value is bounded by $n\Lf$.

Combining all the inequalities above, we get
\begin{equation}
\E_\nu\left[ \| \Fcal(\xbm)\|^2_2 \right] \leq 2\FI(\nu \,\|\, \pi) \,+\, 4nL \,+\, 2\E_\nu \big[ \|\Fcal(\xbm) - \nabla f(\xbm)\|^2_2 \big].
\nonumber
\end{equation}
The proof is complete.
\end{proof}

\subsection{Proof of PMC-RED (Theorem~1)}
\label{Sec:Proof1}

We construct the following interpolation for PMC-RED
\begin{equation}
\label{Eq:AvgPMCRED}
\Xt= \Xkga - (t-k\gamma) \Gcal(\Xkga) + \sqrt{2}\big( \Bt - \Bkga \big) \quad\text{for}\; t \in [k\gamma, (k+1)\gamma]\, .
\end{equation}
We can observe that \eqref{Eq:AvgPMCRED} is simply a `linear' interpolation of PMC-RED. Let $\nu_t$ be the law of $\Xt$. Recall that $\Gcal(\xbm) = \nabla g(\xbm) - \Scal_\theta \big( \xbm , \sigma \big)$.
According to Assumptions~1 and 3, $\Gcal$ is Lipschitz continuous with
\begin{equation}
\LG = L_g + L_\sigma,
\end{equation}
and its $\ell$-2 distance from $\nabla f$ is given by
\begin{equation}
\|\Gcal(\xbm) - \nabla f(\xbm) \|_2 \leq \|\Scal_\theta(\xbm,\sigma) - \nabla \log p(\xbm)\|_2 \leq \sigma C + \epsilon_\sigma\, .
\end{equation}
\medskip
\noindent
Now we are ready to prove Theorem~1. 
From Lemma~\ref{Le:KL}, we know that for $t \in [k\gamma, (k+1)\gamma]$
\begin{equation}
\label{Eq:KLstart}
\frac{d}{d t}\KL(\nu_t \,\|\, \pi) \leq -\frac{3}{4}\FI(\nu_t \,\|\, \pi) + \E\big[ \| \nabla f(\xbm_t) - \Gcal(\xbm_{k\gamma}) \|^2_2 \big].
\end{equation}
The second term can be bounded by invoking the convexity of $\ell$-2 norm
\begin{align}
\label{Eq:SecTerm}
\E\big[ \| \nabla f(\xbm_t) - \Gcal(\xbm_{k\gamma}) \|^2_2 \big]
&\leq 2\E\big[ \| \Gcal(\Xt) -   \Gcal(\xbm_{k\gamma}) \|^2_2 \big] \,+\, 2\E\big[ \| \nabla f(\Xt) - \Gcal(\Xt) \|^2_2 \big]\nonumber\\
&\leq 2 \LG^2 \E \left[ \| \Xt - \xbm_{k\gamma} \|^2_2 \right] \,+\, 2\left( \sigma C + \epsilon_\sigma \right)^2,
\end{align}
We can bound the first term via
\begin{align}
\E \left[ \| \Xt - \xbm_{k\gamma} \|^2_2 \right]
&\leq (t-k\gamma)^2 \E \left[ \|\Gcal(\xbm_{k\gamma})\|^2_2 \right] \,+\, 2\E \left[ \|\Bt - \Bbm_{k\gamma}\|^2_2 \right] \nonumber \\
&= (t-k\gamma)^2 \E \left[ \|\Gcal(\xbm_{k\gamma})  - \Gcal(\Xt) + \Gcal(\Xt)\|^2_2 \right] \,+\, 2n(t-k\gamma) \nonumber\\
&\leq 2(t-k\gamma)^2\Big( \E \left[ \|\Gcal(\Xt)\|^2_2 \right] \,+\, \LG^2 \E\left[ \| \Xt - \xbm_{k\gamma} \|^2_2 \right] \Big) \,+\, 2n(t-k\gamma). \nonumber
\end{align}
Rearranging the terms yields
\begin{equation}
\left( 1-2(t-k\gamma)^2\LG^2 \right) \E \left[ \| \Xt - \xbm_{k\gamma} \|^2_2 \right] 
\leq 2(t-k\gamma)^2 \E \left[ \|\Gcal(\Xt)\|^2_2 \right] \,+\, 2n(t-k\gamma) \nonumber
\end{equation}
which can be simplified by letting $\gamma\leq\frac{1}{2\LG} \Rightarrow 1-2(t-k\gamma)^2\LG^2 \geq 1-2\gamma^2\LG^2 \geq \frac{1}{2}$. Therefore, when $\gamma\leq\frac{1}{2\LG}$, it holds that
\begin{equation}
\label{Eq:Norm}
\E \left[ \| \Xt - \xbm_{k\gamma} \|^2_2 \right]
\leq 4(t-k\gamma)^2 \E \left[ \|\Gcal(\Xt)\|^2_2 \right] \,+\, 4n(t-k\gamma).
\end{equation}
By plugging~\eqref{Eq:Norm} and~\eqref{Eq:SecTerm} into~\eqref{Eq:KLstart} and invoking Lemma~\ref{Le:Bound}, we can obtain
\begin{align}
\label{Eq:ApplyLemma2}
&\frac{d}{d t}\KL(\nu_t \,\|\, \pi) \nonumber\\
&\comment{Plug-in Eq.~\eqref{Eq:SecTerm}} \nonumber\\
&\leq -\frac{3}{4}\FI(\nu_t \,\|\, \pi) \,+\, 2 \LG^2 \E \left[ \| \Xt - \xbm_{k\gamma} \|^2_2 \right] \,+\, 2\left( \sigma C + \epsilon_\sigma \right)^2 \nonumber\\
&\comment{Plug-in Eq.~\eqref{Eq:Norm}} \nonumber\\
&\leq -\frac{3}{4}\FI(\nu_t \,\|\, \pi) \,+\, 8(t-k\gamma)^2\LG^2 \E \left[ \| \Gcal(\Xt)\|^2_2 \right] \,+\, 8n(t-k\gamma)\LG^2 \,+\, 2\left( \sigma C + \epsilon_\sigma \right)^2 \nonumber\\
&\comment{Plug-in Lemma~\ref{Le:Bound} with $\| \nabla f(\xbm) - \Gcal(\xbm) \|_2 \leq \sigma C + \epsilon_\sigma$}\nonumber\\
&\leq -\frac{3}{4}\FI(\nu_t \,\|\, \pi) \,+\, 16(t-k\gamma)^2\LG^2 \left( \FI(\nut \,\|\, \pi) \,+\, 2n\Lf \,+\, \left( \sigma C + \epsilon_\sigma \right)^2 \right) \,+\, 8n(t-k\gamma)\LG^2 \,+\, 2\left( \sigma C + \epsilon_\sigma \right)^2,
\end{align}
where we recall that $\Lf$ is the Lipschitz constant of $\nabla f = \nabla g - \nabla \log p$ and is given by 
$$\Lf = L_g + L_p.$$
Let $L = \max\{\LG, \Lf\}$. We can simplify~\eqref{Eq:ApplyLemma2} by letting $\gamma\leq\frac{1}{\sqrt{32}L} \Rightarrow 16(t-k\gamma)^2L^2 \leq 16\gamma^2L^2 \leq \frac{1}{2}$. Therefore, once $\gamma\leq\frac{1}{\sqrt{32}L}$, we get
\begin{align}
\label{Eq:KL2}
&\frac{d}{d t}\KL(\nu_t \,\|\, \pi) \nonumber\\
&\leq -\frac{1}{4}\FI(\nu_t \,\|\, \pi) \,+\, 16(t-k\gamma)^2L^2 \left( 2nL \,+\, \left( \sigma C + \epsilon_\sigma \right)^2 \right) \,+\, 8(t-k\gamma)nL^2 \,+\, 2\left( \sigma C + \epsilon_\sigma \right)^2.
\end{align}
Integrating~\eqref{Eq:KL2} between $[k\gamma, (k+1)\gamma]$ yields
\begin{align}
&\KL(\nu_{(k+1)\gamma} \,\|\, \pi) - \KL(\nu_{k\gamma} \,\|\, \pi) \nonumber\\
&\leq -\frac{1}{4} \int_{k\gamma}^{(k+1)\gamma} \FI(\nu_t \,\|\, \pi) \,dt \,+\, \frac{16}{3}L^2\gamma^3 \left( 2nL \,+\, \left( \sigma C + \epsilon_\sigma \right)^2 \right) \,+\, 4nL^2\gamma^2 \,+\, 2\gamma\left( \sigma C + \epsilon_\sigma \right)^2 \nonumber\\
&= -\frac{1}{4} \int_{k\gamma}^{(k+1)\gamma} \FI(\nu_t \,\|\, \pi) \,dt \,+\, \left( \frac{32}{3}L\gamma + 4 \right)nL^2\gamma^2 \,+\, \left( \frac{16}{3}L^2\gamma^2 \,+\, 2 \right)\gamma\left( \sigma C + \epsilon_\sigma \right)^2 \nonumber\\
&\leq -\frac{1}{4} \int_{k\gamma}^{(k+1)\gamma} \FI(\nu_t \,\|\, \pi) \,dt \,+\, 6nL^2\gamma^2 \,+\, 3\gamma\left( \sigma C + \epsilon_\sigma \right)^2, \nonumber
\end{align}
where in the last inequality, we invoke $\gamma\leq\frac{1}{\sqrt{32}L}$.
Now by averaging over $N>0$ iterations, dropping the negative term, and applying $\left( \sum_{i=1}^{N} a_i \right)^2 \leq N \sum_{i=1}^{N} a_i^2$, we can derive Theorem~1
\begin{equation}
\frac{1}{N\gamma}\int_{0}^{N\gamma} \FI(\nu_t \,\|\, \pi) \,dt
\leq \frac{4\KL(\nu_{0} \,\|\, \pi)}{N\gamma}
\,+\, \underbrace{\vphantom{\big(} 24nL^2}_{\Asf_1}\gamma
\,+\, \underbrace{\vphantom{\big(} 24 C^2}_{\Asf_2}\sigma^2
\,+\, \underbrace{\vphantom{\big(} 24}_{\Asf_3}\epsilon_\sigma^2.
\nonumber
\end{equation}

\subsection{Proof of PMC-PnP (Theorem~2)}
\label{Sec:Proof2}
The proof of PnP-PnP follows the same logic of the proof in Section~\ref{Sec:Proof1}. 
Here, we sketch the proof by showing the key steps.

\medskip
\noindent
We construct the continuous interpolation of PMC-PnP
\begin{equation}
\label{Eq:AvgPMCPnP}
\Xt= \Xkga - (t-k\gamma) \Pcal(\Xkga) + \sqrt{2}\big( \Bt - \Bkga \big) \quad\text{for}\; t \in [k\gamma, (k+1)\gamma]\, .
\end{equation}
Let $\nu_t$ be the law of $\Xt$. Recall that $\Pcal(\xbm) = \nabla g(\xbm) - \Scal_\theta \big( \xbm - \gamma \nabla g(\xbm), \sigma \big)$. 
According to Assumptions~1 and 3, we know that $\Pcal$ is Lipschitz continuous with
\begin{equation}
\LP = L_g + L_\sigma + \gamma L_g L_\sigma
\end{equation}
and its $\ell$-2 distance from $\nabla f$ is given by
\begin{equation}
\|\Pcal(\xbm) - \nabla f(\xbm) \|_2 \leq \|\Scal_\theta(\xbm-\gamma \nabla g(\xbm),\sigma) - \nabla \log p(\xbm)\|_2 \leq \gamma L_\sigma R_g + \sigma C + \epsilon_\sigma\, .
\end{equation}
Invoking Lemma~\ref{Le:Bound} yields
\begin{align}
\E_\nut\left[ \| \Pcal(\xbm) \|^2_2\right] 
&\leq 2\FI(\nut \,\|\, \pi) + 4n\Lf + 2\E_\nut \big[ \|\Pcal(\xbm) - \nabla f(\xbm)\|^2_2 \big] \nonumber\\
&\leq 2\FI(\nut \,\|\, \pi) + 4n\Lf + 2\left( \gamma L_\sigma R_g + \sigma C + \epsilon_\sigma \right)^2,
\end{align}
where $L_f = L_g + L_p$.
By invoking Lemma~\ref{Le:KL} and following the steps until \eqref{Eq:Norm}, we have 
\begin{equation}
\left( 1-2(t-k\gamma)^2\LP^2 \right) \E \left[ \| \Xt - \xbm_{k\gamma} \|^2_2 \right] 
\leq 2(t-k\gamma)^2 \E \left[ \|\Pcal(\Xt)\|^2_2 \right] \,+\, 2n(t-k\gamma) \nonumber
\end{equation}
To proceed, we want $1-2(t-k\gamma)^2\LP^2 \geq 1-2\gamma^2\LP^2 \geq \frac{1}{2}$. 
Since $\LP$ is dependent on $\gamma$, we need to solve the quadratic inequality $\gamma^2\LP^2 \leq \frac{1}{4}$ to determine the range of $\gamma$
\begin{equation}
\label{Eq:StepsizePnP1}
\gamma \leq \frac{\sqrt{2L_g L_\sigma + (L_g + L_\sigma)^2} - (L_g + L_\sigma)}{2L_\sigma L_g} = \frac{1}{2} \sqrt{ \frac{2}{L_g L_\sigma} + \left( \frac{L_g + L_\sigma}{L_g L_\sigma} \right)^2 }  - \frac{L_g + L_\sigma}{2 L_g L_\sigma}.
\end{equation}
Note that the right hand side is larger than zero.
By choosing $\gamma$ in the above range, we can derive \eqref{Eq:NormPnP} similar as~\eqref{Eq:Norm}
\begin{equation}
\label{Eq:NormPnP}
\E \left[ \| \Xt - \xbm_{k\gamma} \|^2_2 \right]
\leq 4(t-k\gamma)^2 \E \left[ \|\Pcal(\Xt)\|^2_2 \right] \,+\, 4n(t-k\gamma).
\end{equation}
Following the derivation of~\eqref{Eq:ApplyLemma2}, we have
\begin{align}
\label{Eq:KLPnP2}
&\frac{d}{d t}\KL(\nu_t \,\|\, \pi) \nonumber\\
&\leq -\frac{3}{4}\FI(\nu_t \,\|\, \pi) \,+\, 16(t-k\gamma)^2\LP^2 \left( \FI(\nut \,\|\, \pi) \,+\, 2n\Lf \,+\, \left( \gamma L_\sigma R_g + \sigma C + \epsilon_\sigma \right)^2 \right) \nonumber\\
&\hspace{70pt} \,+\, 8n(t-k\gamma)\LP^2 \,+\, 2\left( \gamma L_\sigma R_g + \sigma C + \epsilon_\sigma \right)^2,
\end{align}
Let $L=\max\{\LP, \Lf\}$. To simplify the above inequality, we want to have $16(t-k\gamma)^2\LP^2\leq16\gamma^2L^2 \leq \frac{1}{2}$. 
Since $\LP$ involves $\gamma$ while $L_f$ does not, we derive the range of step-size by cases: If $L = L_f$, then 
$\gamma \leq \frac{1}{\sqrt{32}L_f}$
If $L = \LP$, then we need to solve the quadratic inequality $16\gamma^2\LP^2 \leq \frac{1}{2}$, which yields
\vspace{5pt}
\begin{equation}
\label{Eq:StepsizePnP2}
\gamma \leq \frac{\sqrt{\sqrt{2}/2 L_g L_\sigma + 2(L_g + L_\sigma)^2} - (L_g + L_\sigma)}{2L_gL_\sigma} = \frac{1}{2} \sqrt{ \frac{\sqrt{2}}{L_g L_\sigma} + \left(\frac{L_g + L_\sigma}{L_g L_\sigma}\right)^2 } - \frac{L_g + L_\sigma}{2L_g L_\sigma}.
\end{equation}
Note that~\eqref{Eq:StepsizePnP2} is tighter than~\eqref{Eq:StepsizePnP1}. 
To simplify, we define 
$$\Ltilde = \max \left\{\frac{2L_\sigma L_g}{\sqrt{\sqrt{2}/2 L_g L_\sigma + 2(L_g + L_\sigma)^2} - (L_g + L_\sigma)}, \quad \sqrt{32} L_f \right\}.$$
By selecting $\gamma\leq\frac{1}{\Ltilde}$, we can derive the following inequality
\begin{align}
\label{Eq:KLPnP3}
&\frac{d}{d t}\KL(\nu_t \,\|\, \pi) \nonumber\\
&\leq -\frac{1}{4}\FI(\nu_t \,\|\, \pi) \,+\, 16(t-k\gamma)^2L^2 \left( 2nL \,+\, \left(\gamma L_\sigma R_g + \sigma C + \epsilon_\sigma\right)^2\right) \nonumber\\ 
&\quad\quad\quad\quad\quad\quad\quad\,+\, 8(t-k\gamma)nL^2 \,+\, 2\left(\gamma L_\sigma R_g + \sigma C + \epsilon_\sigma\right)^2.
\end{align}
Integrating~\eqref{Eq:KLPnP3} over $t$ between $[k\gamma, (k+1)\gamma]$ yields
\begin{align}
&\KL(\nu_{(k+1)\gamma} \,\|\, \pi) - \KL(\nu_{k\gamma} \,\|\, \pi) \nonumber\\
&\leq -\frac{1}{4} \int_{k\gamma}^{(k+1)\gamma} \FI(\nu_t \,\|\, \pi) \,dt \,+\, \frac{16}{3}L^2\gamma^3 \left( 2nL \,+\, \left( \gamma L_\sigma R_g + \sigma C + \epsilon_\sigma \right)^2 \right) \,+\, 4nL^2\gamma^2 \,+\, 2\gamma\left( \gamma L_\sigma R_g + \sigma C + \epsilon_\sigma \right)^2 \nonumber\\
&= -\frac{1}{4} \int_{k\gamma}^{(k+1)\gamma} \FI(\nu_t \,\|\, \pi) \,dt \,+\, \left( \frac{32}{3}L\gamma + 4 \right)nL^2\gamma^2 \,+\, \left( \frac{16}{3}L^2\gamma^2 \,+\, 2 \right)\gamma\left( \gamma L_\sigma R_g + \sigma C + \epsilon_\sigma \right)^2 \nonumber\\
&\leq -\frac{1}{4} \int_{k\gamma}^{(k+1)\gamma} \FI(\nu_t \,\|\, \pi) \,dt \,+\, 6nL^2\gamma^2 \,+\, 3\gamma\left( \gamma L_\sigma R_g + \sigma C + \epsilon_\sigma \right)^2, \nonumber
\end{align}
where in the last inequality, we use the fact $\gamma L \leq \frac{1}{\sqrt{32}}$.
Now by averaging over $N>0$ iterations, dropping the negative term, and using $\left( \sum_{i=1}^{N} a_i \right)^2 \leq N \sum_{i=1}^{N} a_i^2$ and $L_\sigma \leq L$, we can derive Theorem~2
\begin{align}
\frac{1}{N\gamma}\int_{0}^{N\gamma} \FI(\nu_t \,\|\, \pi) \,dt
&\leq \frac{4\KL(\nu_{0} \,\|\, \pi)}{N\gamma} \,+\, 24\gamma nL^2 \,+\, 12 \left(\gamma L_\sigma R_g + \sigma C + \epsilon_\sigma\right)^2 \nonumber\\
&\leq \frac{4\KL(\nu_{0} \,\|\, \pi)}{N\gamma} \,+\, 
(24nL^2 + 36\gamma L L_\sigma R_g^2)\gamma \,+\, 
36C^2 \sigma^2 \,+\,
36 \epsilon_\sigma^2 \nonumber\\
&\leq \frac{4\KL(\nu_{0} \,\|\, \pi)}{N\gamma} \,+\, 
\underbrace{(24nL^2 + 7 L_\sigma R_g^2) \vphantom{\big(}}_{\Bsf_1}\gamma \,+\, 
\underbrace{36C^2 \vphantom{\big(}}_{\Bsf_2}\sigma^2 \,+\, 
\underbrace{36 \vphantom{\big(}}_{\Bsf_3}\epsilon_\sigma^2. \nonumber
\end{align}

\subsection{Proof of APMC-RED (Theorem~3)}
\label{Sec:Proof3}
The proofs of annealed PMC algorithms follow the proof framework established for their counterparts shown in Section~\ref{Sec:Proof1} and~\ref{Sec:Proof2}.
Hence, we restrict the proof to the key steps for clarity.

\medskip
\noindent
We construct the continuous interpolation of \anneal-RED
\begin{equation}
\label{Eq:AdaPPS}
\Xt= \Xkga - (t-k\gamma) \Gcal_k(\Xkga) + \sqrt{2}\big( \Bt - \Bkga \big) \quad\text{for}\; t \in [k\gamma, (k+1)\gamma]
\end{equation}
Let $\nu_t$ be the law of $\Xt$. Reall that $\Gcalk(\xbm) = \nabla g(\xbm) + \alpha_k\Scal_\theta \big( \xbm, \sigma_k \big)$.
According to Assumptions~1 and 3, $\Gcal_k$ is Lipschitz continuous with
\begin{equation}
\LGk = L_g + \alphak L_\sigmak
\end{equation}
and the $\ell$-2 distance from $\nabla f$ is given by
\begin{equation}
\|\Gcal_k(\xbm) - \nabla f(\xbm) \|_2 \leq \|\alpha_k\Scal_\theta(\xbm,\sigma_k) - \nabla \log p(\xbm)\|_2 \leq \sigmak C + \epsilon_\sigmak + (\alphak-1)R_s 
\end{equation}
We invoking Lemma~\ref{Le:Bound} to yield
\begin{equation}
\E_\nut\left[ \| \Gcal_k(\xbm) \|^2_2\right] 
\leq 2\FI(\nut \,\|\, \pi) + 4n\Lf + 2\Big( \sigmak C + \epsilon_\sigmak + (\alphak-1)R_s \Big)^2 \nonumber
\end{equation}
where $L_f = L_g + L_p$. Let $L_k = \max\{\LGk, \Lf\}$. Then, by invoking Lemma~\ref{Le:KL} and following the steps before~\eqref{Eq:KL2}, we can derive the following inequality under the condition $\gamma\leq\frac{1}{\sqrt{32}L_k}$
\begin{align}
\label{Eq:KL4}
&\frac{d}{d t}\KL(\nu_t \,\|\, \pi) \nonumber\\
&\leq -\frac{1}{4}\FI(\nu_t \,\|\, \pi) \,+\, 16(t-k\gamma)^2L_k^2 \left( 2nL_k \,+\, \Big(\alphak(\sigmak C + \epsilon_\sigmak) + (\alphak-1)R_s\Big)^2\right) \nonumber\\
&\quad\quad\quad\quad\quad\quad\quad\,+\, 8(t-k\gamma)nL_k^2 \,+\, 2\Big(\sigmak C + \epsilon_\sigmak + (\alphak-1)R_s\Big)^2.
\end{align}
Integrating~\eqref{Eq:KL4} over $t$ between $[k\gamma, (k+1)\gamma]$ yields
\begin{align}
&\KL(\nu_{(k+1)\gamma} \,\|\, \pi) - \KL(\nu_{k\gamma} \,\|\, \pi) \nonumber\\
&\leq -\frac{1}{4} \int_{k\gamma}^{(k+1)\gamma} \FI(\nu_t \,\|\, \pi) \,dt \,+\, 6nL_k^2\gamma^2 \,+\, 3\gamma\Big(\sigmak C + \epsilon_\sigmak + (\alphak-1)R_s\Big)^2, \nonumber
\end{align}
Let $\gamma\leq\frac{1}{\sqrt{32}\Lmax}$, where $\Lmax = \sup\,\{L_k\}_{k=0}^{N-1}$. 
By averaging over $N>0$ iterations, dropping the negative term, and applying $\left( \sum_{i=1}^{N} a_i \right)^2 \leq N \sum_{i=1}^{N} a_i^2$, we can derive
\begin{align}
\label{Eq:Averaged1}
\frac{1}{N\gamma}\int_{0}^{N\gamma} \FI(\nu_t \,\|\, \pi) \,dt
&\leq \frac{4\KL(\nu_{0} \,\|\, \pi)}{N\gamma} \,+\, \frac{24n\gamma}{N} \sum_{k=0}^{N-1} L_k^2 \,+\, \frac{12}{N} \sum_{k=0}^N\Big(\sigmak C + \epsilon_\sigmak + (\alphak-1)R_s\Big)^2 \nonumber\\
&\leq \frac{4\KL(\nu_{0} \,\|\, \pi)}{N\gamma} \,+\, 24n\gamma \Lmax^2 \,+\, \frac{36}{N}\sum_{k=0}^{N-1} \Big( \sigmak^2C^2 + \epsilon_\sigmak^2 + (\alphak-1)^2R_s^2 \Big)
\end{align}
As $\{\alphak\}_{k=0}^{N-1}$ decreases to one at some iteration $K>0$, we have
\begin{align}
\frac{1}{N}\sum_{k=0}^{N-1} (\alphak-1)^2R_s^2 &= \frac{1}{N}\sum_{k=0}^{K-1} (\alphak-1)^2R_s^2 + \frac{1}{N}\sum_{k=K}^{N-1} (\alphak-1)^2R_s^2 \nonumber\\
&= \frac{1}{N}\sum_{k=0}^{K-1} (\alphak-1)^2R_s^2.
\end{align}
which asymptotically goes to zero. This means that weighted annealing will not introduce extra error influencing the convergence accuracy.
We can derive Theorem~3 by simplifying~\eqref{Eq:Averaged1}
\begin{align}
\frac{1}{N\gamma}\int_{0}^{N\gamma} \FI(\nu_t \,\|\, \pi) \,dt
&\leq \frac{4\KL(\nu_{0} \,\|\, \pi)+\gamma\zeta}{N\gamma} \,+\, 
\underbrace{24n\Lmax^2 \vphantom{\big(}}_{\Csf_1}\gamma \,+\, 
\underbrace{36C^2 \vphantom{\big(}}_{\Csf_2}\sigmabar^2 \,+\, 
\underbrace{36 \vphantom{\big(}}_{\Csf_3}\epsilonbar^2
\end{align}
where $\sigmabar^2 = \frac{1}{N}\sum_{k=0}^{N-1}\sigmak^2$, $\epsilonbar = \frac{1}{N}\sum_{k=0}^{N-1}\epsilon_\sigmak$, and $\zeta=36\sum_{k=0}^{K-1}(\alphak-1)^2R_s^2$. 

\subsection{Proof of APMC-PnP (Theorem~4)}
\label{Sec:Proof4}

We construct the continuous interpolation of \anneal-PnP
\begin{equation}
\label{Eq:AdaPPS}
\Xt= \Xkga - (t-k\gamma) \Pcal_k(\Xkga) + \sqrt{2}\big( \Bt - \Bkga \big) \quad\text{for}\; t \in [k\gamma, (k+1)\gamma]\, .
\end{equation}
Let $\nu_t$ be the law of $\Xt$. Recall $\Pcalk(\xbm) = \nabla g(\xbm) + \alpha_k\Scal_\theta \big( \xbm - \gamma \nabla g(\xbm), \sigma_k \big)$.
According to Assumptions~1 and 3, $\Pcal_k$ is Lipschitz continuous with
\begin{equation}
\LPk = L_g + \alpha_k L_\sigmak + \alpha_k\gamma L_g L_\sigmak,
\end{equation}
and the $\ell$-2 distance from $\nabla f$ is given by
\begin{equation}
\|\Pcal_k(\xbm) - \nabla f(\xbm) \|_2 \leq \|\Scal_\theta(\xbm-\nabla g(\xbm),\sigma) - \nabla \log p(\xbm)\|_2 \leq \gamma L_\sigmak R_g + \sigmak C + \epsilon_\sigmak + (\alphak-1)R_s
\end{equation}
Then, we invoke Lemma~\ref{Le:Bound} to yield
\begin{equation}
\E_\nut\left[ \| \Pcal_k(\xbm) \|^2_2\right] \leq 2\FI(\nut \,\|\, \pi) + 4n\Lf + 2\Big( \gamma L_\sigmak R_g + \sigmak C + \epsilon_\sigmak + (\alphak-1)R_s \Big)^2, \nonumber
\end{equation}
where $L_f = L_g + L_p$. 
By invoking Lemma~\ref{Le:KL} and following the steps until \eqref{Eq:Norm}, we have 
\begin{equation}
\left( 1-2(t-k\gamma)^2\LPk^2 \right) \E \left[ \| \Xt - \xbm_{k\gamma} \|^2_2 \right] 
\leq 2(t-k\gamma)^2 \E \left[ \|\Pcalk(\Xt)\|^2_2 \right] \,+\, 2n(t-k\gamma) \nonumber
\end{equation}
Solving the inequality $\gamma^2\LPk^2 \leq \frac{1}{4}$ yields 
\begin{equation}
\label{Eq:StepsizeAPnP1}
\gamma \leq 
\frac{1}{2} \sqrt{ \frac{2}{\alphak L_g L_\sigmak} + \left( \frac{L_g + \alphak L_\sigmak}{\alphak L_g L_\sigmak} \right)^2 }  - \frac{L_g + \alphak L_\sigmak}{2 L_g L_\sigmak}.
\end{equation}
By choosing $\gamma$ in the above range, we can derive \eqref{Eq:NormAPnP} similar as~\eqref{Eq:Norm}
\begin{equation}
\label{Eq:NormAPnP}
\E \left[ \| \Xt - \xbm_{k\gamma} \|^2_2 \right]
\leq 4(t-k\gamma)^2 \E \left[ \|\Pcal(\Xt)\|^2_2 \right] \,+\, 4n(t-k\gamma).
\end{equation}
Following the derivation of~\eqref{Eq:ApplyLemma2}, we have
\begin{align}
&\frac{d}{d t}\KL(\nu_t \,\|\, \pi) \nonumber\\
&\leq -\frac{3}{4}\FI(\nu_t \,\|\, \pi) \,+\, 16(t-k\gamma)^2\LPk^2 \left( \FI(\nut \,\|\, \pi) \,+\, 2n\Lf \,+\, \Big( \gamma L_\sigmak R_g + \sigmak C + \epsilon_\sigmak + (\alphak-1)R_s \Big)^2 \right) \nonumber\\
&\hspace{70pt} \,+\, 8n(t-k\gamma)\LPk^2 \,+\, 2\Big( \gamma L_\sigmak R_g + \sigmak C + \epsilon_\sigmak + (\alphak-1)R_s \Big)^2, \nonumber
\end{align}
Let $L_k = \max\{\LPk, \Lf\}$. As shown in Section~\ref{Sec:Proof2}, we can simplify the above inequality by letting $16\gamma^2L_k^2 \leq \frac{1}{2}$. 
If $L_k = L_f$, then $\gamma \leq \frac{1}{\sqrt{32}L_f}$;
If $L_k = \LPk$, then we need to solve the quadratic inequality $16\gamma^2\LPk^2 \leq \frac{1}{2}$, which yields
\begin{equation}
\label{Eq:StepsizeAPnP2}
\gamma \leq \frac{1}{2} \sqrt{ \frac{\sqrt{2}}{\alphak L_g L_\sigma} + \left(\frac{L_g + \alphak L_\sigma}{\alphak L_g L_\sigma}\right)^2 } - \frac{L_g + \alphak L_\sigma}{2\alphak L_g L_\sigma}.
\end{equation}
Note that~\eqref{Eq:StepsizeAPnP2} is tighter than~\eqref{Eq:StepsizeAPnP1}.
To simplify, we define 
$$\Ltildek = \max \left\{\frac{2\alphak L_\sigma L_g}{\sqrt{\sqrt{2}/2 \alphak L_g L_\sigma + (L_g + \alphak L_\sigma)^2} - (L_g + \alphak L_\sigma)}, \quad \sqrt{32} L_f \right\}$$
By taking $\gamma\leq\frac{1}{\Ltildek}$, we can derive the following inequality
\begin{align}
\label{Eq:KL5}
&\frac{d}{d t}\KL(\nu_t \,\|\, \pi) \nonumber\\
&\leq -\frac{1}{4}\FI(\nu_t \,\|\, \pi) \,+\, 16(t-k\gamma)^2L_k^2 \left( 2nL_k \,+\, \Big( \gamma L_\sigmak R_g + \sigmak C + \epsilon_\sigmak + (\alphak-1)R_s \Big)^2\right) \nonumber\\
&\quad\quad\quad\quad\quad\quad\quad\,+\, 8(t-k\gamma)nL_k^2 \,+\, 2\Big( \gamma L_\sigmak R_g + \sigmak C + \epsilon_\sigmak + (\alphak-1)R_s \Big)^2.
\end{align}
Integrating~\eqref{Eq:KL5} over $t$ between $[k\gamma, (k+1)\gamma]$ yields
\begin{align}
&\KL(\nu_{(k+1)\gamma} \,\|\, \pi) - \KL(\nu_{k\gamma} \,\|\, \pi) \nonumber\\
&\leq -\frac{1}{4} \int_{k\gamma}^{(k+1)\gamma} \FI(\nu_t \,\|\, \pi) \,dt \,+\, 6nL_k^2\gamma^2 \,+\, 3\gamma\Big( \gamma L_\sigmak R_g + \sigmak C + \epsilon_\sigmak + (\alphak-1)R_s \Big)^2, \nonumber
\end{align}
where we use the fact $\gamma L_k \leq \frac{1}{\sqrt{32}}$. Let $\Lmax = \sup\,\{L_k\}_{k=0}^{N-1}$ and $\Ltildemax = \sup\,\{\Ltildek\}_{k=0}^{N-1}$, and choose $\gamma \leq \frac{1}{\Ltildemax}$. 
By averaging over $N>0$ iterations, dropping the negative term, and applying $\left( \sum_{i=1}^{N} a_i \right)^2 \leq N \sum_{i=1}^{N} a_i^2$, we can derive
\begin{align}
\label{Eq:Averaged}
&\frac{1}{N\gamma}\int_{0}^{N\gamma} \FI(\nu_t \,\|\, \pi) \,dt \nonumber\\
&\leq \frac{4\KL(\nu_{0} \,\|\, \pi)}{N\gamma} \,+\, \frac{24n\gamma}{N} \sum_{k=0}^{N-1} L_k^2 \,+\, \frac{12}{N} \sum_{k=0}^N\Big( \gamma L_\sigmak R_g + \sigmak C + \epsilon_\sigmak + (\alphak-1)R_s \Big)^2 \nonumber\\
&\leq \frac{4\KL(\nu_{0} \,\|\, \pi)}{N\gamma} \,+\, 24n\gamma \Lmax^2 \,+\, \frac{48}{N}\sum_{k=0}^{N-1} \left(\gamma^2 L_\sigmak^2 R_g^2 + \sigma^2_kC^2 + \epsilon_\sigmak^2 + (\alphak-1)^2R_s^2\right).
\end{align}
Similarly, as $\{\alphak\}_{k=0}^{N-1}$ decreases to one at some iteration $K>0$, we have
\begin{equation}
\frac{1}{N}\sum_{k=0}^{N-1} (\alphak-1)^2R_s^2 = \frac{1}{N}\sum_{k=0}^{K-1} (\alphak-1)^2R_s^2 \nonumber
\end{equation}
which asymptotically goes to zero. This means that weighted annealing will not introduce extra error influencing the convergence accuracy.
We can derive Theorem~4 by simplifying~\eqref{Eq:Averaged}
\begin{align}
&\frac{1}{N\gamma}\int_{0}^{N\gamma} \FI(\nu_t \,\|\, \pi) \,dt \nonumber\\
&\leq \frac{4\KL(\nu_{0} \,\|\, \pi)+\gamma\zeta}{N\gamma} \,+\, 
(24n\Lmax^2 + 48 R_g^2 \frac{1}{N}\sum_{k=0}^{N-1}\gamma L_k L_{\sigma_k} )\gamma \,+\, 
48C^2\sigmabar^2 \,+\, 
48\epsilonbar^2 \nonumber\\
&\leq \frac{4\KL(\nu_{0} \,\|\, \pi)+\gamma\zeta}{N\gamma} \,+\, 
\underbrace{(24n\Lmax^2 + 9L_{\sigma_\mathsf{max}} R_g^2) \vphantom{\Big(}}_{\Dsf_1}\gamma \,+\, 
\underbrace{48C^2 \vphantom{\Big(}}_{\Dsf_2}\sigmabar^2 \,+\, 
\underbrace{48 \vphantom{\Big(}}_{\Dsf_3}\epsilonbar^2
\end{align}
where $\sigmabar^2 = \frac{1}{N}\sum_{k=0}^{N-1}\sigma_k^2$ and $\epsilonbar = \frac{1}{N}\sum_{k=0}^{N-1}\epsilon_\sigmak$, and the constants are defined as $\zeta=48\sum_{k=0}^{K-1}(\alphak-1)^2R_s^2$ and $L_{\sigma_\mathsf{max}} = \sup\,\{L_\sigmak\}_{k=0}^{N-1}$. In the derivation, we use the facts $L_\sigmak \leq L_k$ and $\gamma L_k \leq \frac{1}{\sqrt{32}}$.

\section{Additional Technical Details}

This section presents the additional technical details omitted in Section~5 and Section~6 of the main paper. 
Detailed configurations of the hyperparameters and code are available via this link\footnote{\url{https://github.com/sunyumark/PnP-MonteCarlo}}.

\subsection{Technical details in numerical validations of theory}

\paragraph{Numerical validation of convergence} We use the same set of hyperparameters for PMC-PnP/RED to collect the convergence results.
We consider an exponential annealing schedule for $\{\sigmak\}_{k=0}^{N-1}$ and $\{\alphak\}_{k=0}^{N-1}$ defined as
\begin{equation}
\label{Eq:Annealing}
\sigmak = \min\{\sigma_0 \xi^k, \sigma_\mathsf{min}\} \quad\text{and}\quad \alphak = \max\{ \alpha_0\sigma_k^2, 1\},
\end{equation}
where $\xi$ denotes the decaying rate.
Note that we alway make $\alpha_0\leq 1/\sigma_\mathsf{min}^2$ to ensure $\{\alphak\}$ will converge to one.
We set $\xi=0.975$ and $\sigma_0=\alpha_0=10$ in this validation. 
Except the test for $\sigma_\mathsf{min}$, we set $\sigma_\mathsf{min}=0$ in all other tests. Note that the schedule in~\eqref{Eq:Annealing} is also used in the remaining experiments presented in the paper with different parameter realizations. We initialize the APMC algorithm with random points uniformly distributed in the subspace $[-50, 50]^2$. 
Fig.~\ref{Fig:2d_GMM} visualizes one example of the testing posterior distributions.

\begin{figure}[t!]
\centering
\includegraphics[width=0.5\linewidth]{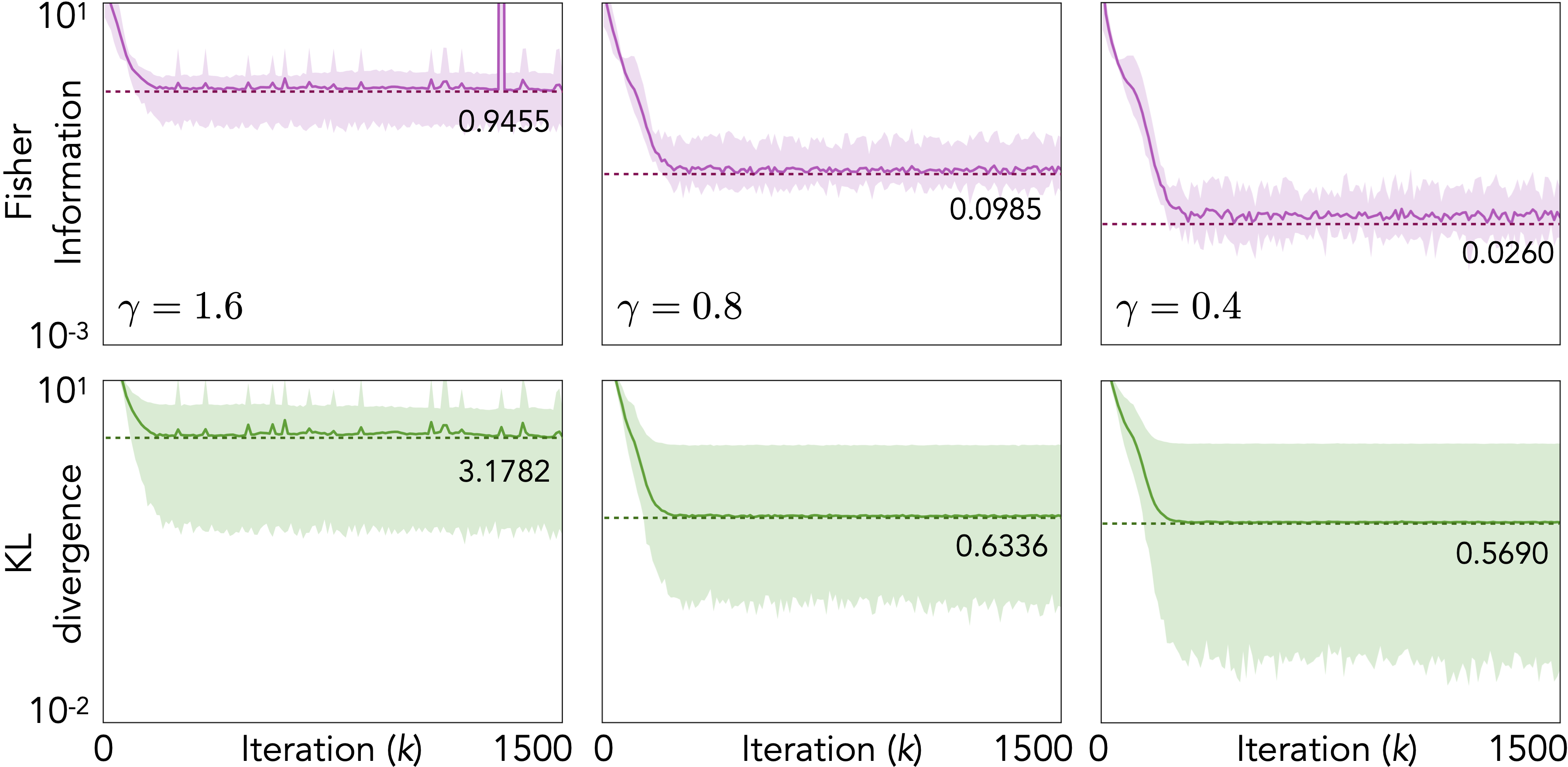}
\caption{
Illustration of the influence of the step-size $\gamma$ on the convergence of APMC-PnP. 
The $\FI(\nu_k \,\|\, \pi)$ are plotted against the iteration number for $\gamma \in \{1.6, 0.8, 0.4\}$.
We also include the plots of $\KL(\nu_k \,\|\, \pi)$ as reference.
The shaded areas in the plots represent the range of values attained across all test distributions.
The dotted line at the bottom shows the minimal value attained by the algorithm.
This plot illustrates that the empirical behavior of APMC-PnP is consistent with Theorem~4, where the convergence accuracy improves with smaller $\gamma$.
}
\vspace{-10pt}
\label{Fig:2d_PMCPnP_step-size}
\end{figure}

\begin{figure}[t!]
\centering
\includegraphics[width=0.5\linewidth]{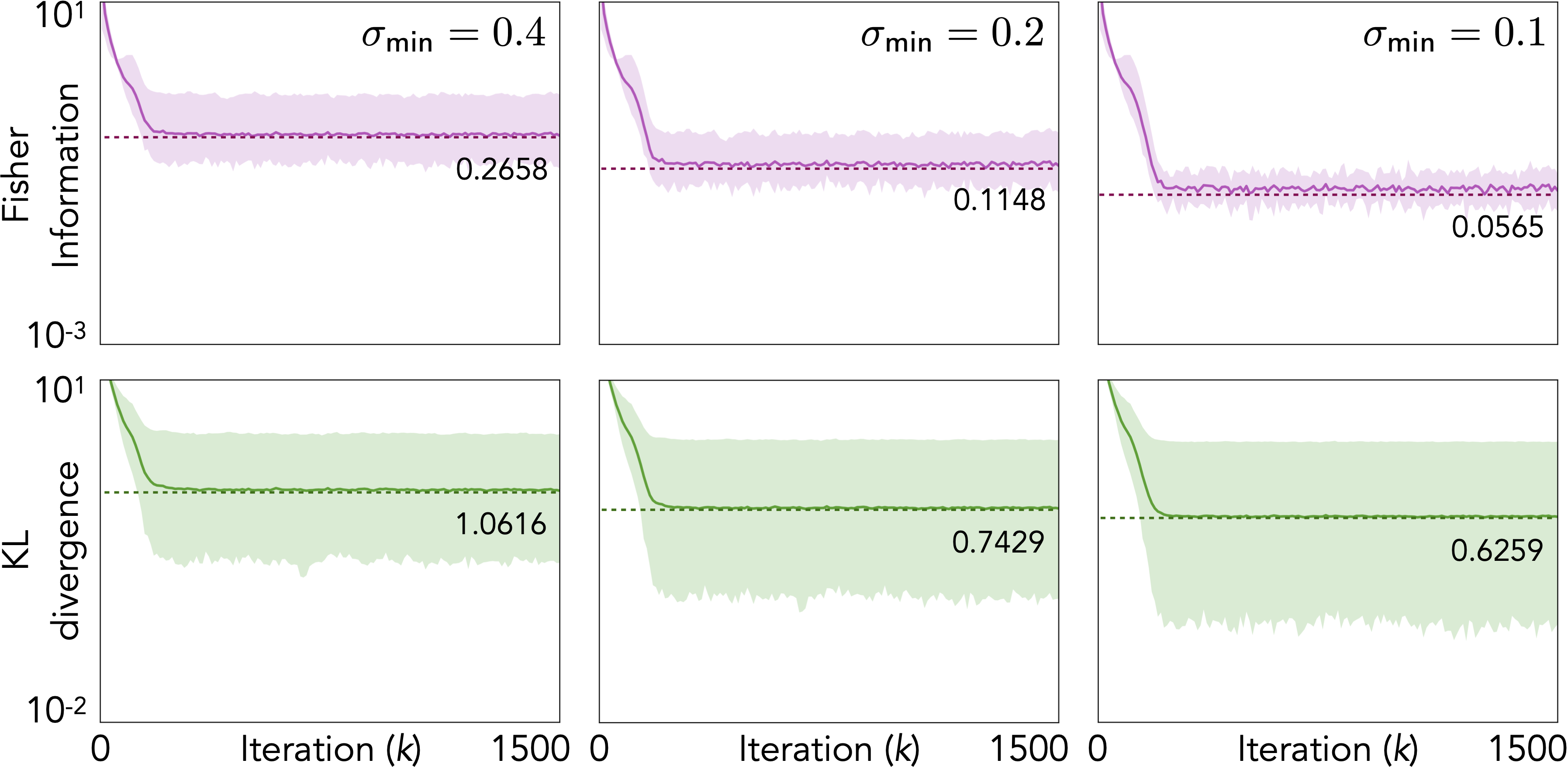}
\caption{
Illustration of the influence of the averaged smoothing strength $\sigma_\mathsf{min}$ on the convergence of APMC-PnP.
The $\FI(\nu_k \,\|\, \pi)$ are plotted against the iteration number for $\sigma_\mathsf{min} \in \{0.4, 0.2, 0.1\}$.
We also include the plots of $\KL(\nu_k \,\|\, \pi)$ as reference.
The shaded areas in the plots represent the range of values attained across all test distributions.
The dotted line at the bottom shows the minimal value attained by the algorithm.
This plot illustrates that the empirical behavior of APMC-PnP is consistent with Theorem~4, where the convergence accuracy improves with smaller $\sigmabar^2$.
}
\label{Fig:2d_PMCPnP_Mismatch}
\end{figure}

\begin{figure}[t!]
  \centering
  \includegraphics[width=0.5\linewidth]{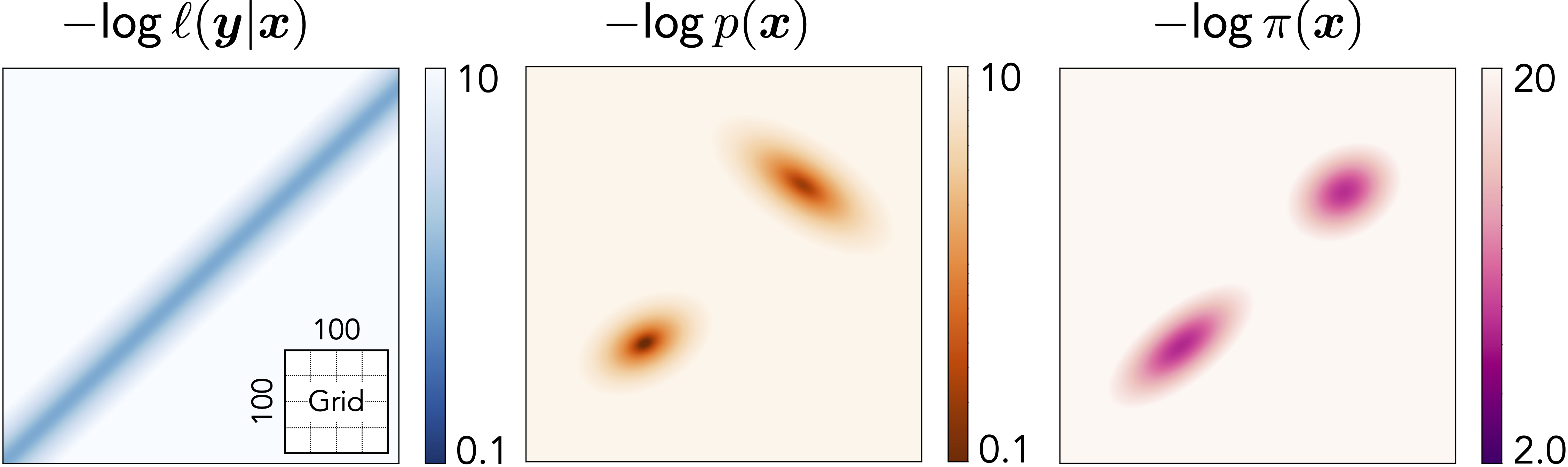}
  \caption{One example of the twenty test 2D posterior distributions. All distributions are defined on a subspace of $[-50,50]^2$.}
  \label{Fig:2d_GMM}
  \vspace{-10pt}
\end{figure}

There are two difficulties in numerically computing the Fisher information (FI) and Kullback–Leibler (KL) divergence: \emph{1)} how to specify the distributions defined by the intermediate samples and \emph{2)} how to evaluate the FI/KL formula which is defined in a continuous space.
To address the first problem, we use Gaussian mixture model (GMM) to fit a distribution to the intermediate samples, enabling the computation of the probability. 
The number of components in the GMM is set to that of the ground-truth posterior distribution, which is two in our case. 
It's worth noting that other numerical density estimation methods, such as kernel density estimation (KDE), can also be utilized. 
However, we select GMM as it is also used to construct the ground-truth distribution and offers computational efficiency.
To address the second problem, we define a fine discrete grid of $1000\times1000$ unit areas uniformly distributed in $[-50,50]^2$. 
By numerically evaluating the FI/KL formula in each unit area and then summing all values, we are able to obtain an approximation of the FI/KL values.

\vspace{-0.5em}
\paragraph{Statistical validation of image posterior sampling}

The network used in this validation is adapted from our customized score network by reducing the number of levels in the U-Net to accommodate $32\times32$ images.
In the experiment, we set the annealing parameters as $\xi=0.975$, $\sigma_\mathsf{min}=\sqrt{1/4000}$, $\alpha_0=4e3$, and $\sigma_0=192$.
When equipped with the analytic score, we notice that APMC algorithms require a smaller $\alpha_0$ to avoid running into numerical errors in early iterations. 
Hence, we set $\alpha_0=2.5$ in that experiment.
The algorithms are initialized with random images uniformly distributed in $[-3, 3]^n$, where $n>0$ denotes the dimensionality of the image. We set $\gamma=1e\shortneg3$ for all algorithms in all cases.

\begin{figure}[t!]
\centering
\includegraphics[width=0.5\linewidth]{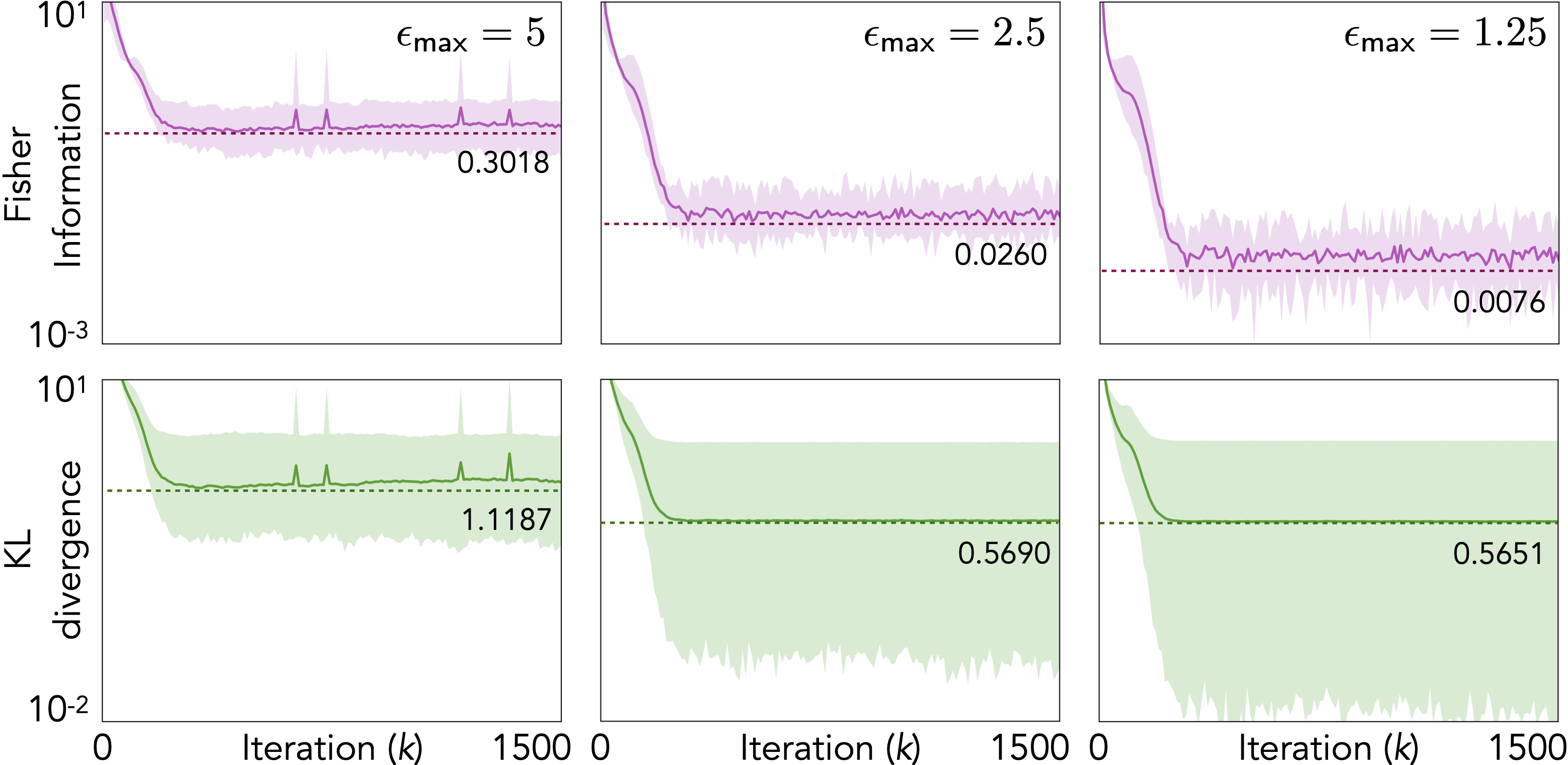}
\caption{
Illustration of the influence of the maximal approximation error $\epsilon_\mathsf{max}$ on the convergence of APMC-PnP.
The $\FI(\nu_k \,\|\, \pi)$ are plotted against the iteration number for $\epsilon_\mathsf{max} \in \{5, 2.5, 1.25\}$.
We also include the plots of $\KL(\nu_k \,\|\, \pi)$ as reference.
The shaded areas in the plots represent the range of values attained across all test distributions.
The dotted line at the bottom shows the minimal value attained by the algorithm.
This plot illustrates that the empirical behavior of APMC-PnP is consistent with Theorem~4, where the convergence accuracy improves with smaller $\epsilonbar^2$.
}
\label{Fig:2d_PMCPnP_Error}
\end{figure}

\begin{figure*}[t!]
\centering
\includegraphics[width=0.5\linewidth]{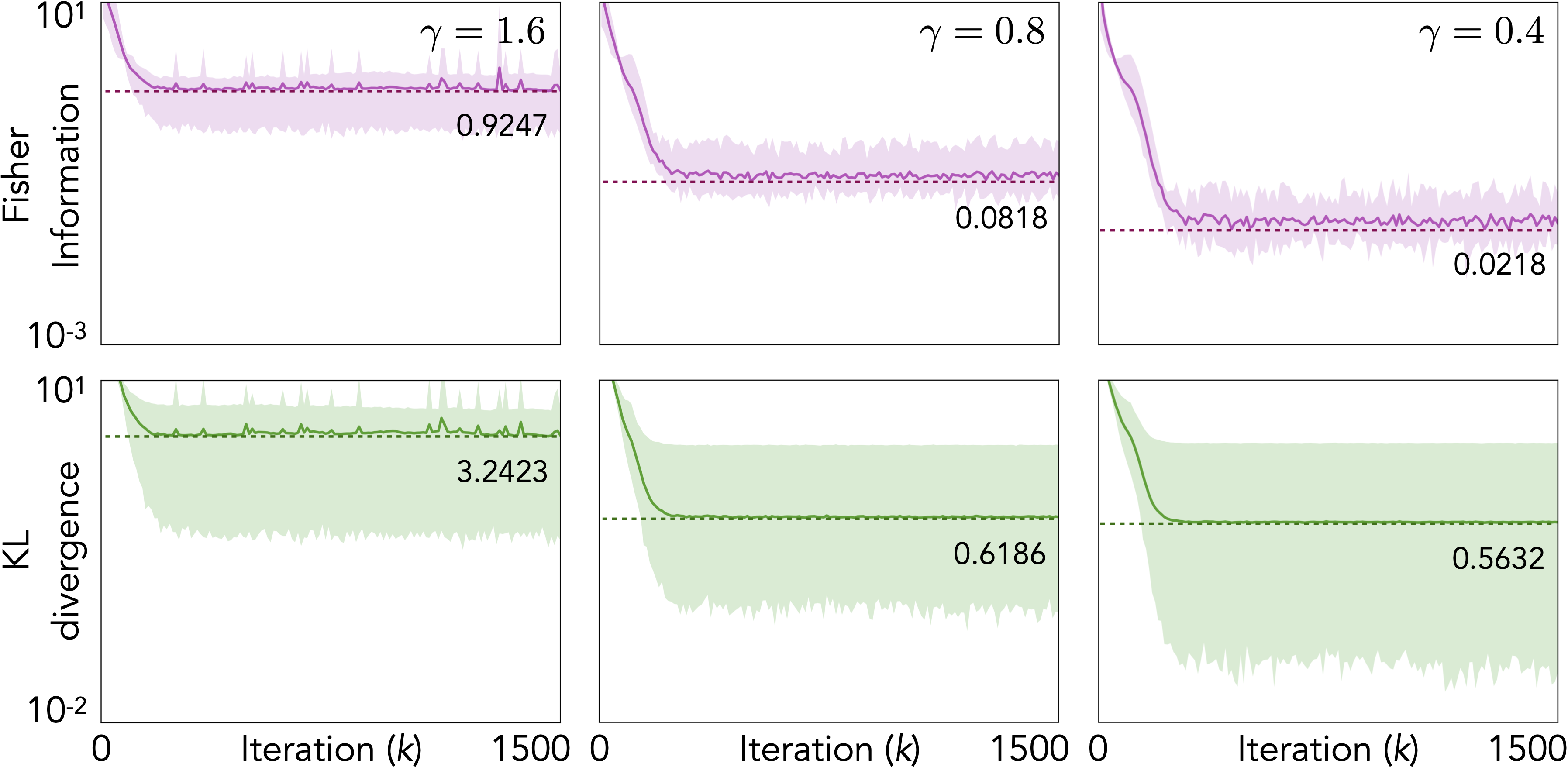}
\caption{
Illustration of the influence of the step-size $\gamma$ on the convergence of APMC-RED.
The $\FI(\nu_k \,\|\, \pi)$ are plotted against the iteration number for $\gamma \in \{1.6, 0.8, 0.4\}$.
We also include the plots of $\KL(\nu_k \,\|\, \pi)$ as reference.
The shaded areas in the plots represent the range of values attained across all test distributions.
The dotted line at the bottom shows the minimal value attained by the algorithm.
This plot illustrates that the empirical behavior of APMC-RED is consistent with Theorem~3, where the convergence accuracy improves with smaller $\gamma$.
}
\label{Fig:2d_PMCRED_step-size}
\end{figure*}

\begin{figure*}[t!]
\centering
\includegraphics[width=0.5\textwidth]{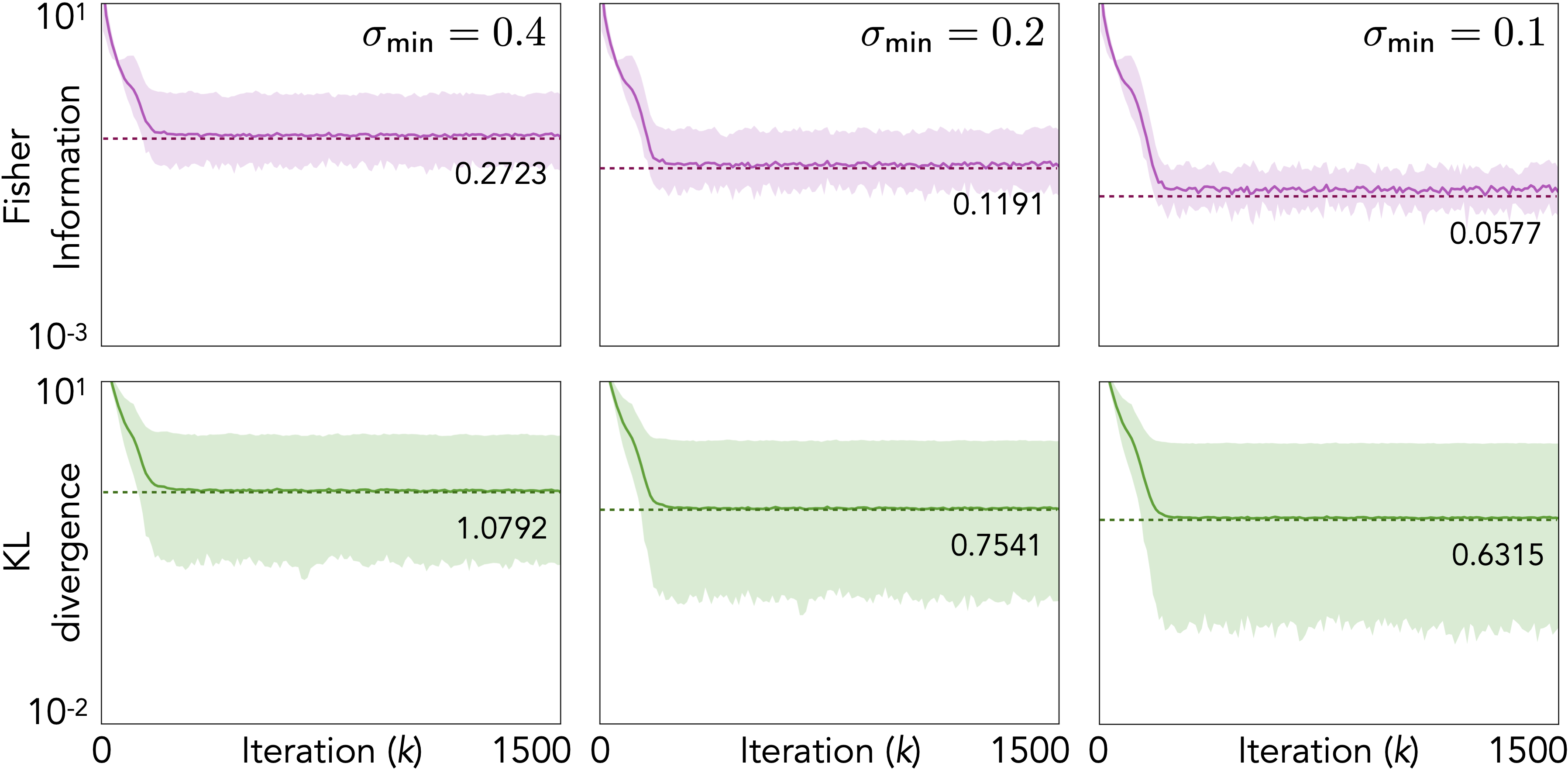}
\caption{
Illustration of the influence of the averaged smoothing strength $\sigma_\mathsf{min}$ on the convergence of APMC-RED.
The $\FI(\nu_k \,\|\, \pi)$ are plotted against the iteration number for $\sigma_\mathsf{min} \in \{0.4, 0.2, 0.1\}$.
We also include the plots of $\KL(\nu_k \,\|\, \pi)$ as reference.
The shaded areas in the plots represent the range of values attained across all test distributions.
The dotted line at the bottom shows the minimal value attained by the algorithm.
This plot illustrates that the empirical behavior of APMC-RED is consistent with Theorem~3, where the convergence accuracy improves with smaller $\sigmabar^2$.
}
\label{Fig:2d_PMCRED_Mismatch}
\end{figure*}

\begin{figure*}[t!]
\centering
\includegraphics[width=0.5\textwidth]{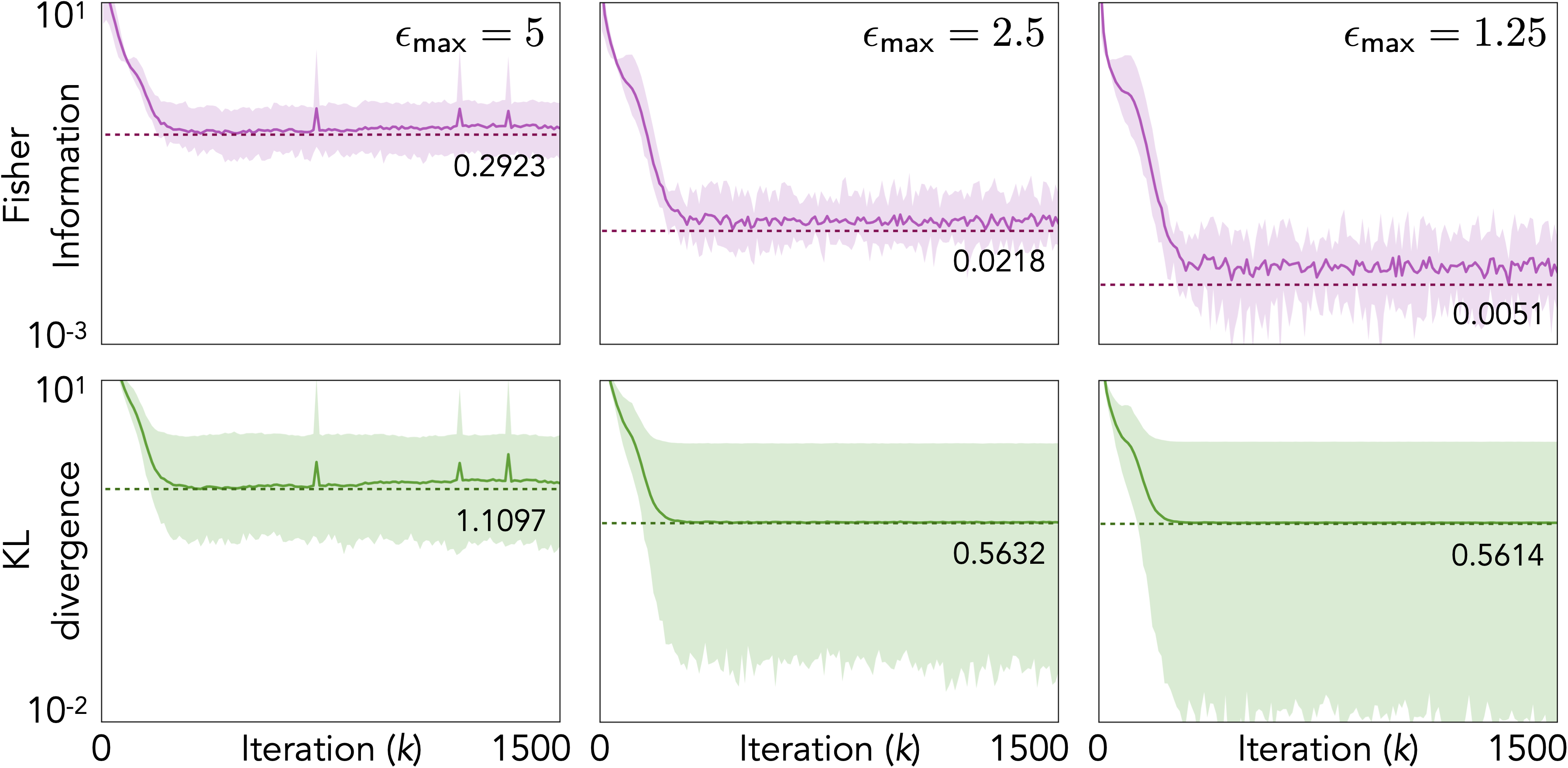}
\caption{
Illustration of the influence of the maximal approximation error $\epsilon_\mathsf{max}$ on the convergence of APMC-RED. 
The $\FI(\nu_k \,\|\, \pi)$ are plotted against the iteration number for $\epsilon_\mathsf{max} \in \{5, 2.5, 1.25\}$.
We also include the plots of $\KL(\nu_k \,\|\, \pi)$ as reference.
The shaded areas in the plots represent the range of values attained across all test distributions.
The dotted line at the bottom shows the minimal value attained by the algorithm.
This plot illustrates that the empirical behavior of APMC-RED is consistent with Theorem~3, where the convergence accuracy improves with smaller $\epsilonbar^2$.
}
\label{Fig:2d_PMCRED_Error}
\end{figure*}

\subsection{Technical details in image recovery tasks}

\paragraph{Score-based generative priors}
The original U-Net in~\cite{Dhariwal.etal2021} takes the time step $t$ as the auxiliar input, accommodating the diffusion process described by the \emph{variance preserving (VP)} SDE. We observe that a naive change from $t$ to the smoothing strength $\sigma$ leads to inferior performance. 
Instead, we replace the encoding network with the one used for the \emph{variance exploding} SDE~\cite{Song.etal2021score} which takes $\sigma$ as the input.
We train the score networks on the machine equipped with the AMD EPYC 7H12 64-Core CPU Processor and NVIDIA RTX A6000 GPU.
We apply random flipping to the training images for data augmentation.
The training time of our network is around 48 hours on a single GPU.
For the computation of per-iteration runtime, we use the machine equipped with AMD EPYC 9354 32-Core CPU Processor and NVIDIA A100 GPU.

\vspace{-0.5em}
\paragraph{Linear inverse problems}

We implement PnP and RED by following the residual formulations in~Eq. (10) and~Eq.~(11), respectively.
In the test, the algorithms are equipped with the pre-trained DnCNN denoisers~\cite{Zhang.etal2017} for CS and MRI images.
We note that the original implementation of PnP-ULA pre-trains a DnCNN to approximate the score, which suggests that our implementation with the score network may lead to better empirical performance.
We implement DPS by using the code provided in the repository associated with~\cite{Chung.etal2023diffusion}. 
We use the pre-trained score network for the face images, but re-train a separate network for the brain MRI images.
We compute the PSNR value using the following formula 
$$\mathsf{PSNR}(\xbmhat, \xbm) = 10\log_{10} \frac{\max(\xbm)^2}{\mathsf{MSE}(\xbmhat,\xbm)}\quad\text{where}\quad \mathsf{MSE}(\xbmhat,\xbm) = \frac{1}{N} \|\xbmhat-\xbm\|_2^2$$
where $\max(\xbm)$ computes the largest pixel value in the reference $\xbm$.
The hyperparameter values are summarized in Table~\ref{Tab:CSParam1}-\ref{Tab:MRIParam2}.
In the experiments, we set the annealing parameters as $\xi=0.99$ and $\sigma_0=368$. We apply a grid search to finetune the remaining parameters for the best PSNR values. 
To reduce the degree of freedom, we couple $\sigma_\mathsf{min}$ and $\alpha_0$ by making  
$ \sigma_\mathsf{min}^2 \alpha_0 =1$.
We initialize the APMC algorithm with random uniformly distributed in the hypercube $[-1, 1]^n$.

\vspace{-0.5em}
\paragraph{Black-hole interferometry} We adapt our customized score network for $64\times64$ images by reducing the number of levels in the U-Net.
We normalized the training images to $[0,1]$ and applied data augmentation by random flipping and resizing of the black hole.
The term GRMHD stands for \emph{general relativistic magnetohydrodynamic} simulation which can generate high-fidelity, high-resolution black hole images. 
The GRMHD dataset contains around $3000$ images in total.
In the experiment, we set the parameters as $\gamma=5e\shortneg6$, $\alpha_0=8e3$, $\xi=0.97$, $\sigma_\mathsf{min}=0.002$, and $\sigma_0=192$.
The weight of the total flux constraint is set to $\rho=0.5$.
We initialize the APMC algorithm with random images uniformly distributed in the hypercube $[0,1]^n$.

\section{Additional Experimental Results} 

\begin{figure*}[t!]
\centering
\includegraphics[width=0.9\textwidth]{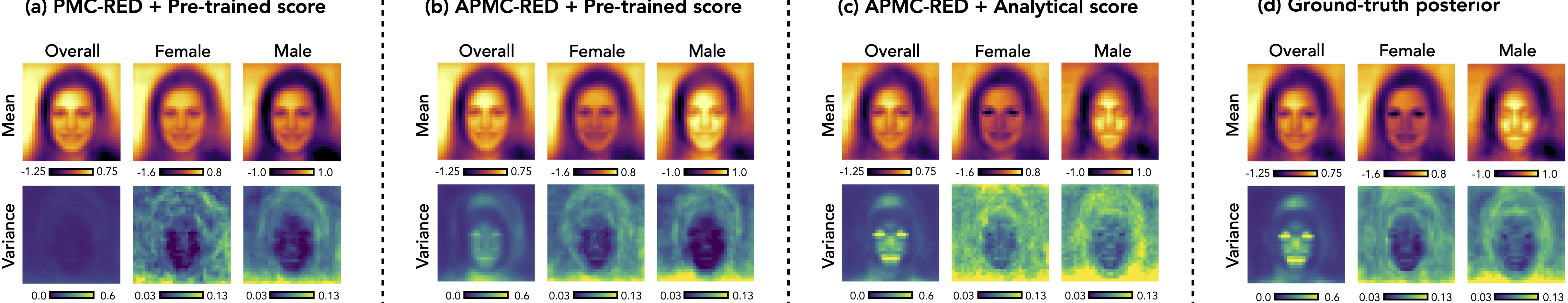}
\caption{
Comparison of sampling performance of APMC-RED against its stationary counterpart and ground-truth posterior distribution.
Each test algorithm is run to collect $1000$ samples which are classified into two modes according to their distance and angle with respect to ground-truth means.
Under the pre-trained score, APMC-RED significantly improves the sampling performance over PMC-RED by distinguishing the female and male modes.
In the ideal case of the analytical score, APMC-RED recovers a high-fidelity distribution as the ground-truth posterior.
Note that APMC-RED yields consistent results as APMC-PnP.
}
\label{Fig:Face_PMCRED}
\end{figure*}

\begin{algorithm}[t]
\small
\setstretch{1.05}
\caption{Plug-and-Play Unadjusted Langevin Algorithm (PnP-ULA)}
\begin{algorithmic}[1]
\STATE \textbf{input: } $\xbm_0 \in \R^n$, $\gamma > 0$, $\sigma>0$, and a convex and compact set $\mathsf{C\subset\R^n}$.
\FOR{$k = 0, 1, \dots, N-1$}
\STATE $\Zk \leftarrow \Ncal(0, I)$
\STATE $\Gcal(\Xk) \leftarrow \nabla g(\Xk) - \Scal_\theta \big( \Xk, \sigma \big)$
\STATE $\Xkpo \leftarrow \Pi_\mathsf{c} \Big(\Xk - \gamma \Gcal(\Xk) + \sqrt{2\gamma} \Zk \Big)$ \hfill $\triangleright$ Projection to the convex and compact set $\mathsf{C}$
\ENDFOR\label{euclidendwhile}
\end{algorithmic}
\label{Alg:PPnPULA}
\end{algorithm}

\subsection{Convergence plots of APMC algorithms}

Fig.~\ref{Fig:2d_PMCPnP_step-size}-\ref{Fig:2d_PMCPnP_Error} plot the convergence of $\FI(\nu_k \,\|\, \pi)$ and $\KL(\nu_k \,\|\, \pi)$ obtained by \anneal-PnP with $\gamma \in \{1.6, 0.8, 0.4\}$, $\sigma_\mathsf{min} \in \{0.4, 0.2, 0.1\}$, and $\epsilon_\mathsf{max} \in \{5.0, 2.5, 1.25\}$, respectively.
The shaded areas in the plots represent the range of values attained across all test distributions. 
The plots clearly illustrate the improvement in $\FI(\nu_k \,\|\, \pi)$ by reducing the value of these parameters as illustrated in Theorem~4. 
Although FI can not be interpreted as a direct proxy of KL~\cite{Balasubramanian.etal2022}, remarkably a similar trend is also observed for $\KL(\nu_k \,\|\, \pi)$.
Additionally, we note that our theoretical analysis does not predict the monotonic reduction of FI, which also seems to be consistent with the empirical behavior of \anneal-PnP.
Fig.~\ref{Fig:2d_PMCRED_step-size}-\ref{Fig:2d_PMCRED_Error} visualize the convergence evolution for \anneal-RED under the same parameter settings.
These figures highlight the same convergence trends for \anneal-RED, where $\gamma$, $\sigmabar$, and $\epsilonbar$ controls the accuracy.

\afterpage{
\begin{table}[t!]
\centering
\tiny
\caption{
Averaged PSNR and MSE values obtained by PnP-ULA under different hyperparameter setups for the CS ($m/n=0.1$ \& $m/n=0.3$) and MRI ($\text{Accel.}=8\times$ \& $\text{Accel.}=4\times$) tasks.
The results of the proposed APMC-RED is additionally included as a reference.
The best numerical values of PnP-ULA are highlighted in \textbf{bold}, and the second best values are highlighted in \underline{underline}.
}
\begin{tabular*}{480pt}{L{75pt} C{35pt}C{35pt} C{35pt}C{35pt} C{0pt} C{35pt}C{35pt} C{35pt}C{35pt}} \toprule
\multirow{2}{*}{\textbf{Method}} & \multicolumn{2}{c}{$m/n=0.1$} & \multicolumn{2}{c}{$m/n=0.3$} & & \multicolumn{2}{c}{$\text{Accel.} = 8\times$} & \multicolumn{2}{c}{$\text{Accel.} = 4\times$}\\
\cmidrule{2-5} \cmidrule{7-10}
& PSNR $\uparrow$ & MSE $\downarrow$ & PSNR $\uparrow$ & MSE $\downarrow$ & & PSNR $\uparrow$ & MSE $\downarrow$ & PSNR $\uparrow$ & MSE $\downarrow$\\
\cmidrule{1-10}
PnP-ULA-NP (PMC-RED) & \first{8.75} & \first{\expnumber{1.43}{\shortneg1}} & $18.18$ & $\expnumber{1.67}{\shortneg2}$ & & $27.51$ & $\expnumber{2.35}{\shortneg3}$ & \first{35.98} & \first{\expnumber{2.84}{\shortneg4}} \\ 
PnP-ULA-Cus & \second{$8.73$} & \second{$\expnumber{1.44}{\shortneg1}$} & \second{$31.00$} & \second{$\expnumber{1.98}{\shortneg3}$} & & \second{$30.97$} & \second{$\expnumber{9.05}{\shortneg4}$} & $34.24$ & $\expnumber{4.82}{\shortneg4}$ \\ 
PnP-ULA-Ori & $8.56$ & $\expnumber{1.51}{\shortneg1}$ & \first{33.58} & \first{\expnumber{4.99}{\shortneg4}} & & \first{32.03} & \first{\expnumber{6.87}{\shortneg4}} & \second{$34.57$} & \second{$\expnumber{3.90}{\shortneg4}$} \\
\noalign{\vskip 0.5ex}
\hdashline\noalign{\vskip 0.5ex}
APMC-RED & $28.37$ & $\expnumber{1.67}{\shortneg3}$ & $34.52$ & $\expnumber{4.06}{\shortneg4}$ & & $32.67$ & $\expnumber{6.43}{\shortneg4}$ & $36.02$ & $\expnumber{2.82}{\shortneg4}$ \\ 
\bottomrule
\end{tabular*}
\label{Tab:ULARecon}
\end{table}
\begin{figure}[t!]
\centering
\includegraphics[width=0.8\textwidth]{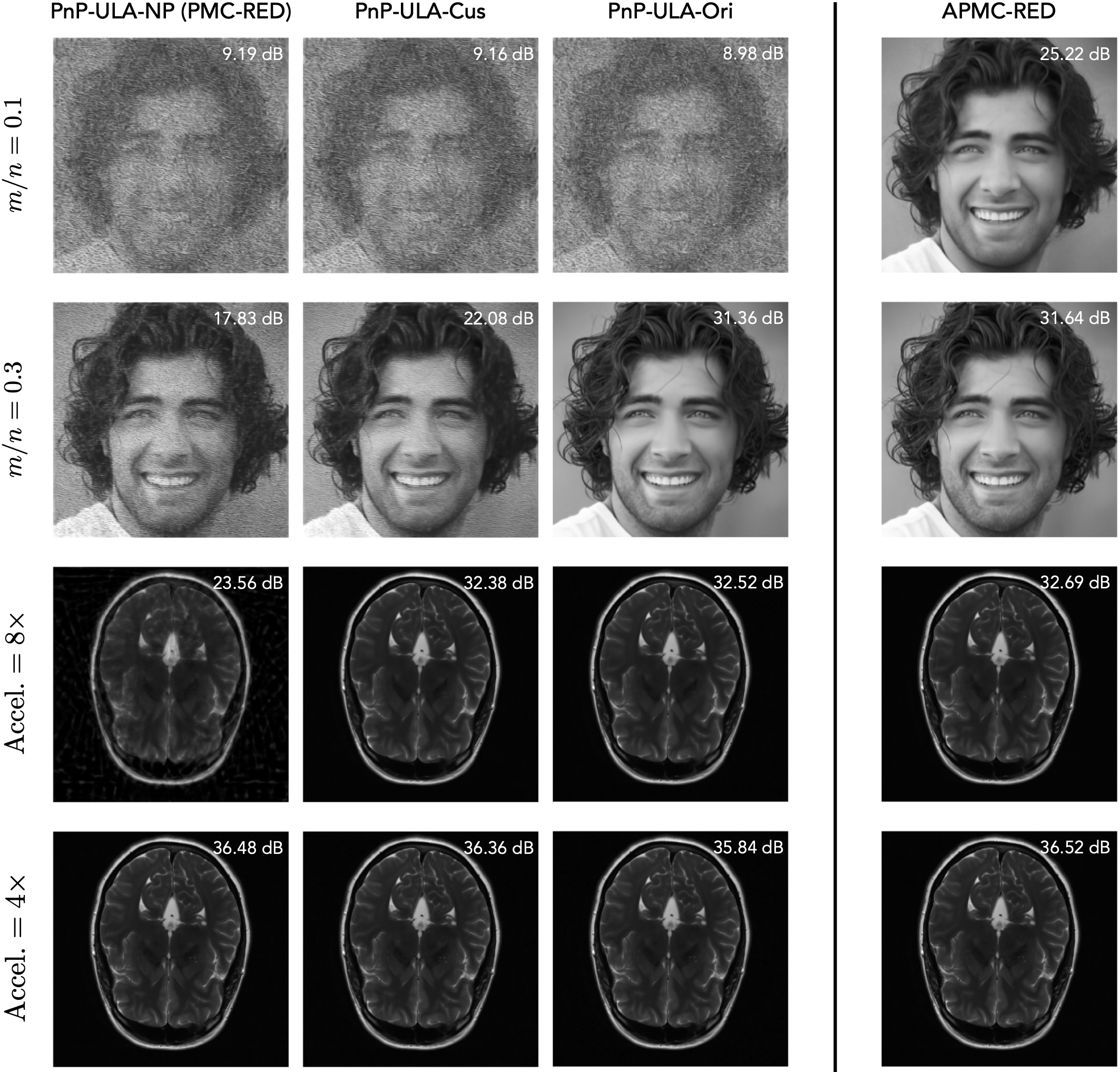}
\caption{
Visual comparison of the reconstructions obtained by PnP-ULA under different hyperparameter setups.
The first two rows correspond to the CS ($m/n=0.1$ \& $m/n=0.3$) tasks, and the last two rows to the MRI ($\text{Accel.}=8\times$ \& $\text{Accel.}=4\times$) tasks.
The final reconstructed images are obtained by averaging $50$ image samples.}
\label{Fig:pnpula_recon}
\end{figure}
}

\afterpage{
\begin{table}[t!]
\centering
\tiny
\caption{
Averaged NLL values obtained by PnP-ULA under different hyperparameter setups for the CS ($m/n=0.1$ \& $m/n=0.3$) and MRI ($\text{Accel.}=8\times$ \& $\text{Accel.}=4\times$) tasks. The values of absolute error ($|\xbmbar-\xbm|$) and standard deviation (SD) that jointly determine NLL are included.
The results of the proposed APMC-RED is additionally included as a reference.
The best numerical values of PnP-ULA are highlighted in \textbf{bold}, and the second best values are highlighted in \underline{underline}.}
\begin{tabular*}{500pt}{L{40pt} C{20pt}C{30pt}C{20pt} C{20pt}C{30pt}C{20pt} C{0pt} C{20pt}C{30pt}C{20pt} C{20pt}C{30pt}C{20pt}} \toprule
\multirow{2}{*}{\textbf{Method}} & \multicolumn{3}{c}{$m/n=0.1$} & \multicolumn{3}{c}{$m/n=0.3$} & & \multicolumn{3}{c}{$\text{Accel.} = 8\times$} & \multicolumn{3}{c}{$\text{Accel.} = 4\times$}\\
\cmidrule{2-7} \cmidrule{9-14}
& \makecell{NLL \\ $\downarrow$} & \makecell{$|\xbmbar-\xbm|$ \\ $\downarrow$} & \makecell{SD \\ $\downarrow$} & \makecell{NLL \\ $\downarrow$} & \makecell{$|\xbmbar-\xbm|$ \\ $\downarrow$} & \makecell{SD \\ $\downarrow$} & & \makecell{NLL \\ $\downarrow$} & \makecell{$|\xbmbar-\xbm|$ \\ $\downarrow$} & \makecell{SD \\ $\downarrow$} & \makecell{NLL \\ $\downarrow$} & \makecell{$|\xbmbar-\xbm|$ \\ $\downarrow$} & \makecell{SD \\ $\downarrow$}\\
\cmidrule{1-14} 
PnP-ULA-NP (PMC-RED)  & \first{1.366} & \first{0.3122} & \first{0.1873} & $\shortneg0.622$ & $0.1006$ & $0.1083$ & & $\shortneg2.056$ & $0.0290$ & $0.0407$ & \first{\shortneg2.883} & \first{0.0116} & \first{0.0166}\\
PnP-ULA-Cus & \second{$1.396$} & \second{$0.3131$} & \second{$0.1873$} & \second{$\shortneg2.176$} & \second{$0.0268$} & \second{$0.0419$} & & \second{$\shortneg2.270$} & \second{$0.0183$} & \second{$0.0545$} & $\shortneg2.529$ & $0.0138$ & $0.0396$\\
PnP-ULA-Ori & $1.414$ & $0.3205$ & $0.1886$ & \first{\shortneg2.474} & \first{0.0161} & \first{0.0249} & & \first{\shortneg2.345} & \first{0.0166} & \first{0.0405} & \second{$\shortneg2.550$} & \second{$0.0135$} & \second{$0.0352$} \\
\noalign{\vskip 0.5ex}
\hdashline\noalign{\vskip 0.5ex}
APMC-RED & $\shortneg2.005$ & $0.0278$ & $0.0348$ & $\shortneg2.667$ & $0.0140$ & $0.0179$ & & $\shortneg2.689$ & $0.0152$ & $0.0207$ & $\shortneg2.884$ & $0.0115$ & $0.0165$ \\
\bottomrule
\end{tabular*}
\label{Tab:ULAUQ}
\end{table}
\begin{figure}[t!]
\centering
\includegraphics[width=0.9\textwidth]{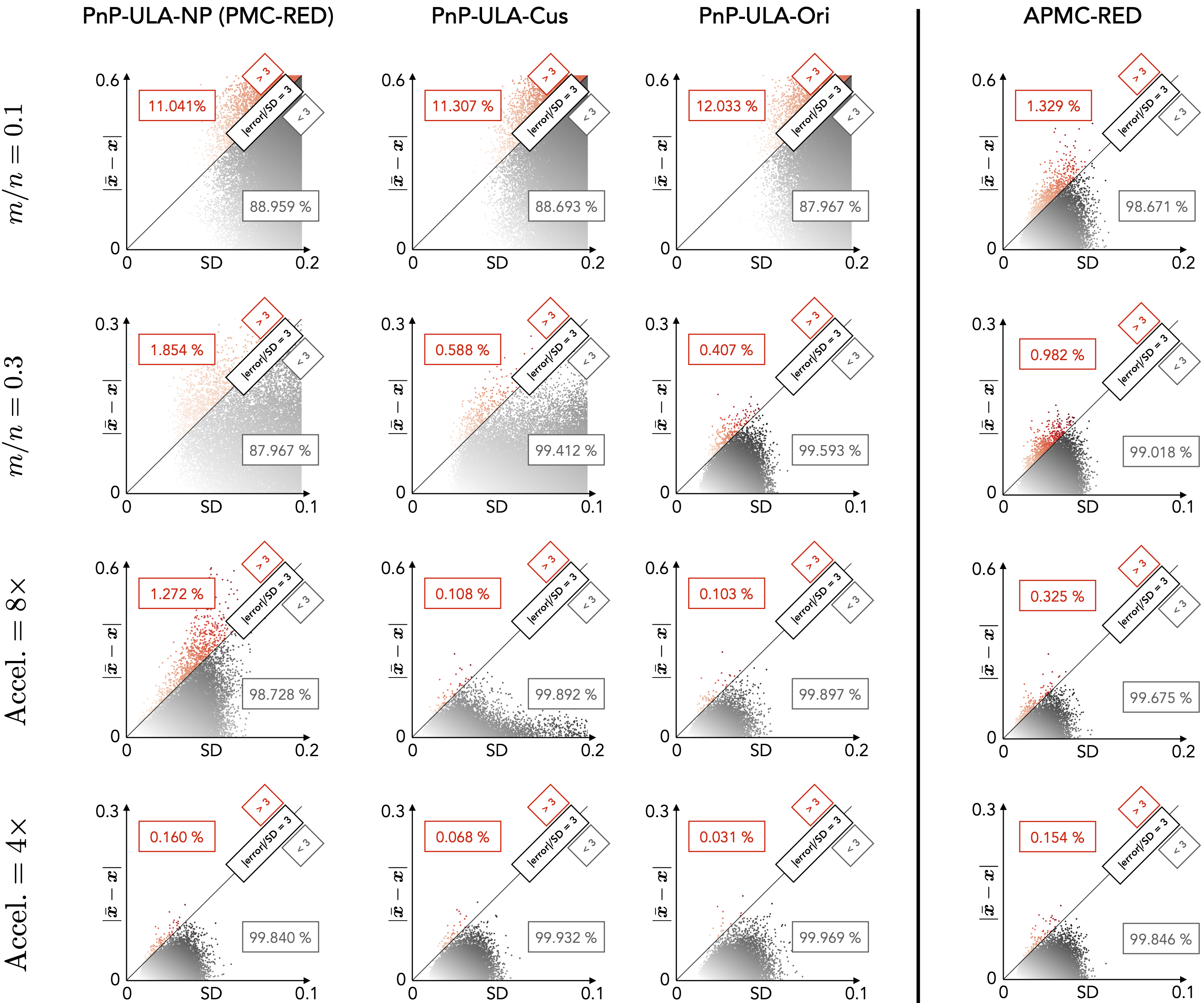}
\caption{
Visualization of the 3-SD credible intervals associated with the CS and MRI reconstructions shown in Fig.~\ref{Fig:pnpula_recon}. The absolute difference ($|\xbmbar - \xbm|$) is plotted against the standard deviation (SD) for each pixel.
The first two rows correspond to the CS ($m/n=0.1$ \& $m/n=0.3$) tasks, and the last two rows to the MRI ($\text{Accel.}=8\times$ \& $\text{Accel.}=4\times$) tasks. The outlying pixels are highlighted in red.
}
\label{Fig:pnpula_uq}
\end{figure}
}

\subsection{Image posterior sampling results of APMC-RED}
Fig.~\ref{Fig:Face_PMCRED} compares the sampling performance of \anneal-RED against its stationary counterpart and ground-truth posterior.
All algorithms are run until convergence to sample $1000$ images.
The figure demonstrates that weighted annealing also help alleviate mode collapse for PMC-RED.
Under the analytic score, APMC-RED obtains almost identical distribution as the ground-truth posterior, showing similar results as APMC-PnP.
Note that this experiment further corroborate the consistency between APMC-PnP and APMC-RED as they are governed by the same gradient-flow ODE.

\subsection{Additional evaluation of PnP-ULA}
In this section, we evaluate the performance of PnP-ULA in the context of linear inverse problems. 
We implemented PnP-ULA by closely following Alg.~\ref{Alg:PPnPULA}, which corresponds to the projected PnP-ULA algorithm in~\cite{Laumont.etal2022}. 
We chose projected PnP-ULA due to its algorithmic simplicity. 
We did not introduce an additional regularization parameter since $\Scal_\theta(\xbm,\sigma)$ directly outputs an approximation of $\nabla\log\psigma(\xbm)$.

In total, we consider three PnP-ULA variants using different hyperparameter setups.
The first two variants are based on the choice of $\mathsf{C}=[\shortneg1,2]^n$, which is aligned with the setup used in~\cite{Laumont.etal2022}. 
We further consider two different $\sigma$ setup: a) $\sigma=5/255$, which is the value originally used in~\cite{Laumont.etal2022}, and b) $\sigma=\sigma_\mathsf{pmcred}$, where $\sigma_\mathsf{pmcred}$ denotes the value used in PMC-RED.
For simplicity, we use \emph{PnP-ULA-Ori} and \emph{PnP-ULA-Cus} to denote the PnP-ULA using $\sigma=5/255$ and $\sigma=\sigma_\mathsf{pmcred}$, respectively.
Note that $\sigma_\mathsf{pmcred}$ is always smaller than $5/255$, suggesting a potentially more accurate approximation of the true score $\nabla\log p(\xbm)$. 
As PMC-RED can be viewed as PnP-ULA equipped with a non-activated projection, we include PMC-RED as the third variant of PnP-ULA, denoted as \emph{PnP-ULA-NP} in the experiment. We note that the iterates of PMC-RED never exceed $[\shortneg10^4, 10^4]^n$, equivalent to using $\mathsf{C} = [\shortneg10^4, 10^4]^n$.
For each algorithm, we selected the largest possible step-size from the sets $\{8e\shortneg5, 4e\shortneg5, 2e\shortneg5, 1e\shortneg5\}$ and $\{1e\shortneg4, 5e\shortneg5, 1e\shortneg5, 5e\shortneg6\}$ for CS and MRI recovery tasks, respectively. 
A summary of hyperparameter values is presented in Table~\ref{Tab:HyperPnPULA}.

The evaluation is performed on the same $30$ test images used in the initial submission.
We run all PnP-ULA algorithms for a maximum number of $10,000$ iterations to infer a batch of $50$ samples initialized randomly.
Table~\ref{Tab:ULARecon} and~\ref{Tab:ULAUQ} respectively summarize the averaged PSNR and NLL values obtained by each PnP-ULA variant. We include the results of AMPC-RED as a reference.
We observe that $\mathsf{C}=[-1,2]^n$ enables the use of larger step-sizes for the $30\%$ CS, $4\times$ MRI, and $8\times$ MRI recovery tasks, resulting in faster convergence speeds.
This is reflected in the better PNSR and NLL values obtained by PnP-ULA-Ori and PnP-ULA-Cus for $30\%$ CS and $8\times$ MRI recovery tasks.
On the other hand, PnP-ULA-NP converges within $10,000$ iterations for the $4\times$ MRI task, achieving better numerical values than the other two PnP-ULA variants.
This fact implies that PnP-ULA encounters a trade-off: using a small projection range enables a large step-size at the expense of increased estimation bias. 
Nevertheless, we note that $\mathsf{C}=[-1,2]^n$ does not enable a larger step-size for the $10\%$ CS task. Hence, PnP-ULA-Ori and PnP-ULA-Cus yields unsatisfactory sampling performance, similar to PnP-ULA-NP, due to slow convergence.
Fig.~\ref{Fig:pnpula_recon} and~\ref{Fig:pnpula_uq} provides a visual comparison on the reconstruction quality and uncertainty quantification performance for all PnP-ULA variants.

\subsection{Experiment on convergence speed}

We conduct a convergence comparison between APMC-RED, PnP-ULA-NP (PMC-RED), and PnP-ULA-Ori in the context of linear inverse problems.
We restrict our comparison to APMC-RED for brevity, since APMC-PnP demonstrates almost the same convergence behavior.
We note that PnP-ULA-Ori allows for larger step-sizes for the $30\%$ CS, $4\times$ MRI, and $8\times$ MRI tasks, while using the same step-size as PnP-ULA-NP for the $10\%$ CS task.
A summary of hyperparameters used in PnP-ULA-NP and PnP-ULA-Ori is provided in Table~\ref{Tab:HyperPnPULA}.

Fig.~\ref{Fig:CS_Convergence} plots the PSNR curves for the considered algorithms in the context of CS tasks.
Each curve is computed using the sample mean of $10$ image samples and averaging over two test images.
In the scenario of $30\%$ CS recovery, PnP-ULA-Ori converges within $10,000$ iteration, yielding a faster convergence speed than PnP-ULA-NP. Nevertheless, APMC-RED still outperforms PnP-ULA-Ori in terms of speed.
In the scenario of $10\%$ CS recovery, both PnP-ULA-NP and PnP-ULA-Ori demonstrate slow convergence. 
Note that we ran the two algorithms for $150,000$ iterations, but they still fail to reach a similar PSNR level as APMC-RED. 
In contrast, APMC-RED generally converges around $4,000$ iterations due to its utilization of weighted annealing.
Fig.~\ref{Fig:MRI_Convergence} plots the PSNR curves in the context of MRI tasks.
PnP-ULA-Ori demonstrates much faster convergence than PnP-ULA-NP in both acceleration setting, achieving convergence speeds on par with APMC-PnP.
Nevertheless, PnP-ULA-Ori yields lower PSNR values than APMC-RED and PnP-ULA-NP after the convergence.
Table~\ref{Tab:Runtime} summarizes the per-iteration runtime of all the sampling algorithms for generating 10 and 50 image samples.

\begin{table}[!t]
	\centering
	\tiny
	\setstretch{1.2}
	\caption{A summary of different hyperparameter setups used for evaluating the performance of PnP-ULA.}
        \label{Tab:HyperPnPULA}
	\begin{tabular*}{390pt}{L{50pt}C{5pt}C{100pt}C{54pt}C{54pt}C{54pt}} 			
		\toprule
		& \multicolumn{2}{c}{\textbf{Hyperparameters}} & \textbf{PnP-ULA-NP (PMC-RED)} & \textbf{PnP-ULA-Cus} & \textbf{PnP-ULA-Ori} \\	 
		\cmidrule(lr){2-3} \cmidrule(lr){4-6}
		\multirow{3}{*}{\makecell{CS \\ ($m/n=0.1$)}} & $\gamma$  & step-size & $1e$-$5$ & $1e$-$5$ & $1e$-$5$ \\
		& $\sigma$  & second input of $\Scal_\theta(\xbm, \sigma)$ & $0.015$ & $0.015$ & $5/255$ \\
            & $\mathsf{C}$  & projection range of $\Pi_\mathsf{C}$ & $[\shortneg10^4, 10^4]^n$ & $[\shortneg1, 2]^n$ & $[\shortneg1, 2]^n$ \\
		\noalign{\vskip 0.5ex}
		\hdashline\noalign{\vskip 0.5ex}
		\multirow{3}{*}{\makecell{CS \\ ($m/n=0.3$)}} & $\gamma$  & step-size & $1e$-$5$ & $2e$-$5$ & $2e$-$5$ \\
		& $\sigma$  & second input of $\Scal_\theta(\xbm, \sigma)$ & $0.009$ & $0.009$ & $5/255$ \\
            & $\mathsf{C}$  & projection range of $\Pi_\mathsf{C}$ & $[\shortneg10^4, 10^4]^n$ & $[\shortneg1, 2]^n$ & $[\shortneg1, 2]^n$ \\
  		\noalign{\vskip 0.5ex}		
            \hdashline\noalign{\vskip 0.5ex}
		  \multirow{3}{*}{\makecell{MRI \\ (Accel.$=8\times$)}} & $\gamma$  & step-size & $5e$-$6$ & $5e$-$5$ & $5e$-$5$ \\
		& $\sigma$  & second input of $\Scal_\theta(\xbm, \sigma)$ & $0.01$ & $0.01$ & $5/255$ \\
            & $\mathsf{C}$  & projection range of $\Pi_\mathsf{C}$ & $[\shortneg10^4, 10^4]^n$ & $[\shortneg1, 2]^n$ & $[\shortneg1, 2]^n$ \\
  		\noalign{\vskip 0.5ex}		
            \hdashline\noalign{\vskip 0.5ex}
		\multirow{3}{*}{\makecell{MRI \\ (Accel.$=4\times$)}} & $\gamma$  & step-size & $5e$-$6$ & $5e$-$5$ & $5e$-$5$ \\
		& $\sigma$  & second input of $\Scal_\theta(\xbm, \sigma)$ & $0.01$ & $0.01$ & $5/255$ \\
            & $\mathsf{C}$  & projection range of $\Pi_\mathsf{C}$ & $[\shortneg10^4, 10^4]^n$ & $[\shortneg1, 2]^n$ & $[\shortneg1, 2]^n$ \\
        \bottomrule
	\end{tabular*}
	\label{Tab:ULAParam} 
\end{table}

\begin{figure*}[t!]
\centering
\includegraphics[width=0.73\textwidth]{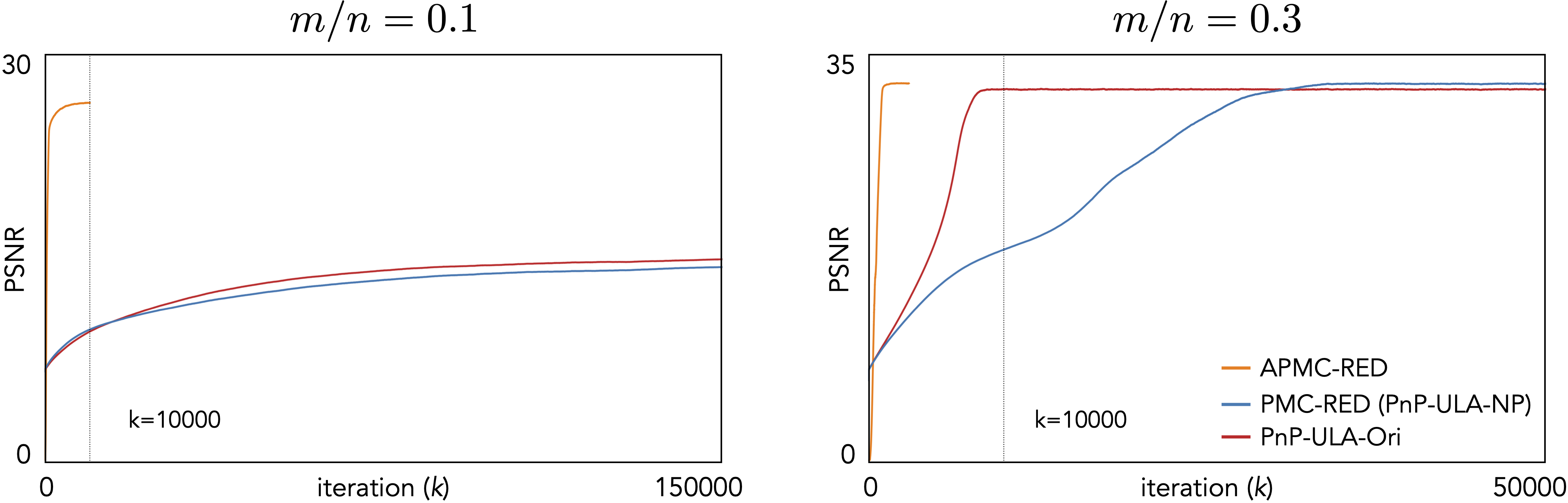}
\caption{Comparison of the convergence speed between APMC-RED, PnP-ULA-NP (PMC-RED), and PnP-ULA-Ori in the context of CS tasks.
Each sampling algorithm is run to generate $10$ samples, the mean of which is used to compute the PSNR values.
Note the significant speed acceleration enabled by weighted annealing.}
\label{Fig:CS_Convergence}
\end{figure*}

\begin{figure*}[t!]
\centering
\includegraphics[width=0.73\textwidth]{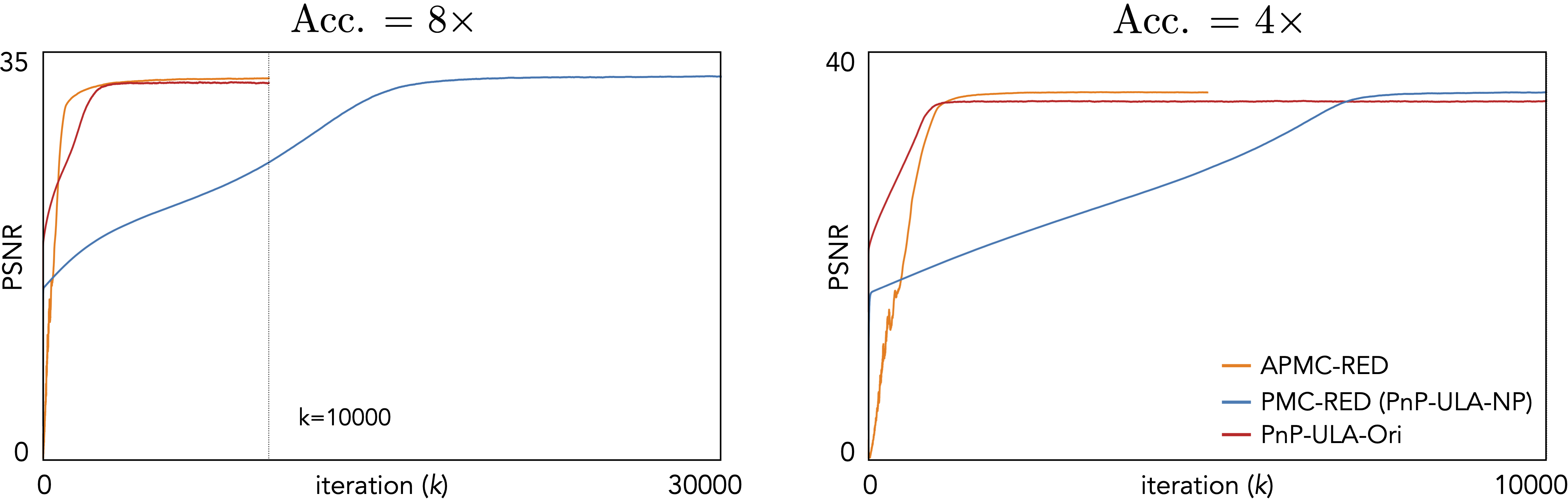}
\caption{
Comparison of the convergence speed between APMC-RED, PnP-ULA-NP (PMC-RED), and PnP-ULA-Ori in the context of accelerated MRI. 
Each sampling algorithm is run to generate $10$ samples, the mean of which is used to compute the PSNR values.
The final PSNR curves are averaged over two test images.
Note that PnP-ULA-Ori converges to slightly worse PSNR levels than APMC-RED.
}
\label{Fig:MRI_Convergence}
\end{figure*}

\begin{table*}[t]
\centering
\scriptsize
\caption{
Per-iteration runtime (\emph{in second}) obtained by the annealed PMC (APMC) and baseline sampling algorithms for the CS and MRI recovery tasks. The table summarizes the time taken per iteration by each algorithm to generate 10 and 50 samples. The per-iteration runtime of the deterministic PnP algorithm equipped with the DnCNN denoiser is included as a reference, where only a single image sample is computed during inference.
}
\begin{tabular*}{470pt}{L{70pt} C{35pt}C{35pt} C{35pt}C{35pt} C{0pt} C{35pt}C{35pt} C{35pt}C{35pt}} \toprule
\multirow{2}{*}{\textbf{Method}} & \multicolumn{2}{c}{$m/n=0.1$} & \multicolumn{2}{c}{$m/n=0.3$} & & \multicolumn{2}{c}{$\text{Accel.} = 8\times$} & \multicolumn{2}{c}{$\text{Accel.} = 4\times$}\\
\cmidrule{2-5} \cmidrule{7-10}
& 10 samples & 50 samples & 10 samples & 50 samples & & 10 samples & 50 samples & 10 samples & 50 samples\\
\cmidrule{1-10}
PnP-ULA & $0.1159$ & $0.5674$ & $0.1556$ & $0.7664$ & & $0.0969$ & $0.4677$ & $0.0968$ & $0.4691$ \\ 
DPS & $0.1178$ & $0.5657$ & $0.1561$ & $0.7615$ & & $0.1006$ & $0.4743$ & $0.1008$ & $0.4720$ \\ 
\noalign{\vskip 0.5ex}
\hdashline\noalign{\vskip 0.5ex}
APMC-RED (ours) & $0.1163$ & $0.5675$ & $0.1543$ & $0.7668$ & & $0.0968$ & $0.4689$ & $0.0967$ & $0.4686$\\ 
APMC-PnP (ours) & $0.1164$ & $0.5677$ & $0.1556$ & $0.7665$ & & $0.0967$ & $0.4693$ & $0.0967$ & $0.4686$\\ 
\noalign{\vskip 0.5ex}
\hline\hline
\noalign{\vskip 1ex}
PnP & \multicolumn{2}{c}{$0.00414$} & \multicolumn{2}{c}{$0.00874$} & & \multicolumn{2}{c}{$0.00257$} & \multicolumn{2}{c}{$0.00264$} \\
\bottomrule
\end{tabular*}
\label{Tab:Runtime}
\end{table*}

\begin{table}[t!]
\center
\scriptsize
\caption{Averaged PSNR values obtained by \anneal~algorithms and the state-of-the-art VarNet for MRI recovery. The best PSNR values are highlighted in bold.}
\begin{tabular*}{190pt}{L{60pt} C{40pt} C{0pt} C{40pt}} \toprule
\multirow{2}{*}{\textbf{Method}}& $8\times$ & & $4\times$\\
\cmidrule{2-2}\cmidrule{4-4}
& PSNR (dB) $\uparrow$ & & PSNR (dB) $\uparrow$ \\
\cmidrule{1-4}
VarNet-$4\times$ & $28.51$ & & $36.00$ \\
VarNet-$8\times$ & $\mathbf{32.82}$ & & $35.01$ \\
\noalign{\vskip 0.5ex}
\hdashline\noalign{\vskip 0.5ex}
APMC-RED (ours) & $32.67$ & & $36.02$ \\ 
APMC-PnP (ours) & $32.70$ & & $\mathbf{36.03}$ \\
\bottomrule
\end{tabular*}
\label{Tab:MRI_VarNet}
\end{table}

\subsection{Comparison with VarNet}
We compare the APMC algorithms with the state-of-the-art end-to-end VarNet~\cite{Sriram.etal2020} in terms of the reconstruction quality. We train two VarNets respectively corresponding to $4\times$ and $8\times$ acceleration by using the same FastMRI brain dataset. 
The performance of VarNet is optimized when the test setup matches the training data.
Table~\ref{Tab:MRI_VarNet} summarizes the numerical results.
As iterative methods using model-agnostic priors, both APMC algorithms demonstrate comparable performance with the optimized VarNet, which relies on a large number of matched measurement-image pairs for training; recall that such a dedicated dataset is not required for training the score network used in our proposed approach.
Furthermore, APMC algorithms are robust to different imaging setups while VarNet shows a big drop in reconstruction quality when data mismatch occurs due to additional subsampling.
A visual comparison is presented in Fig.~\ref{Fig:MRI_VarNet}.

\begin{figure*}[t!]
\centering
\includegraphics[width=0.93\textwidth]{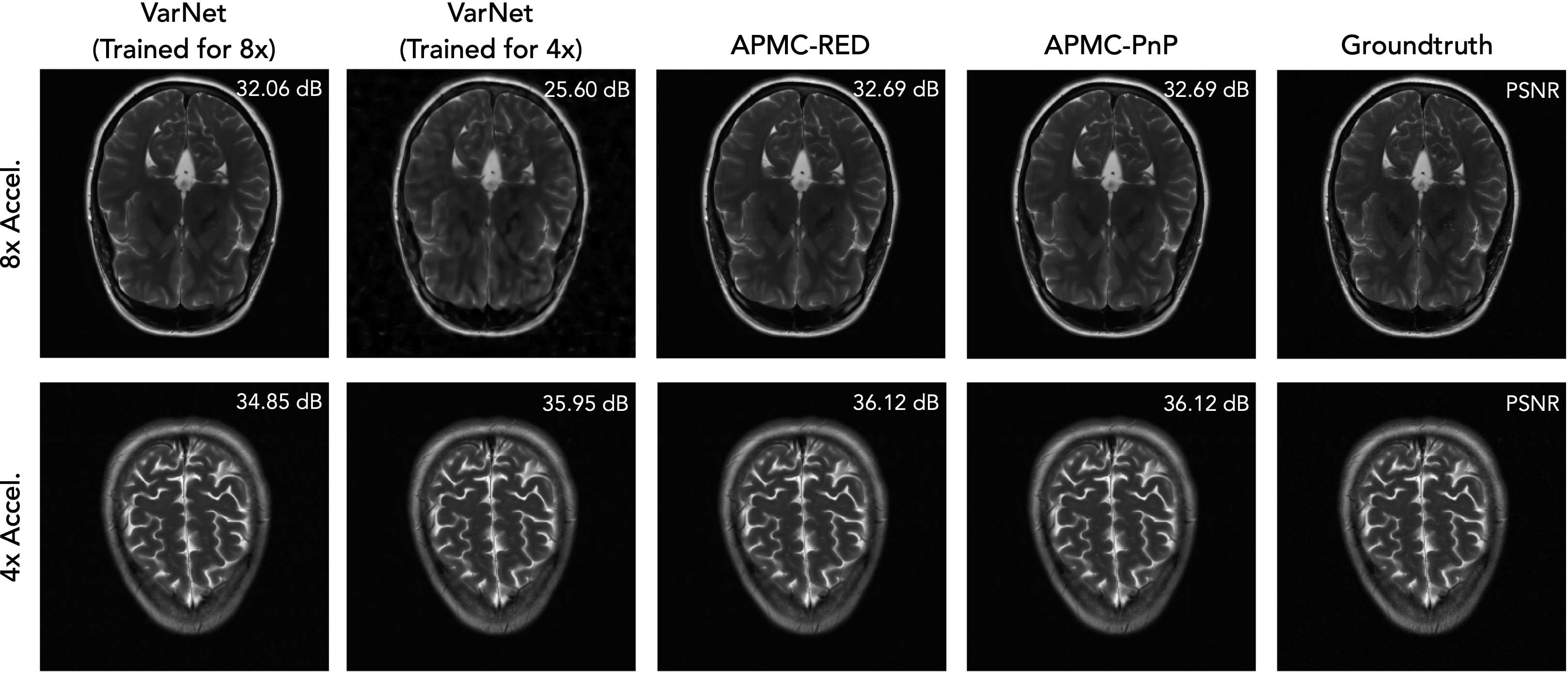}
\caption{
Visual comparison of the reconstructions obtained by the APMC and VarNet. The first and second row visualizes the images reconstructed by each algorithm for $8\times$ and $4\times$ MRI recovery, respectively.
Here, \emph{VarNet (Trained for $8\times$)} denotes the VarNet model trained for $8\times$ MRI recovery, and \emph{VarNet (Trained for $4\times$)} denotes the model trained for $4\times$ MRI recovery.
}
\label{Fig:MRI_VarNet}
\end{figure*}

\subsection{Other supplementary results}

Fig.~\ref{Fig:CS_MMSE} and~\ref{Fig:CS_UQ} visualize the reconstruction and uncertainty quantification (UQ) results for $30\%$ CS recovery.
Visually, APMC significantly outperforms PnP, RED, and DPS. Note how APMC algorithm restore the fine details of hair and hat.
While PnP-ULA-Ori yields visual quality on par with APMC algorithms, but the latter still numerically outperform the former by about $0.4$ dB.
The more accurate estimate of $\xbmbar$ also leads to a better UQ performance by avoiding a large SD.

Fig.~\ref{Fig:MRI_MMSE} and~\ref{Fig:MRI_UQ} visualize the same results for $4\times$ MRI recovery.
First, APMC algorithms significantly outperform the PnP and RED algorithms, achieving a PSNR margin of $\sim2.1$ dB.
This highlights the potential of APMC in solving general MRI reconstruction.
Comparing with other sampling algorithms, APMC algorithms obtain better visual performance in restoring fine details (see the zoom-in regions). This is also corroborated by their higher PSNR values.
In the comparison of UQ, APMC algorithms generate samples with smaller SD while still maintaining a 99\% 3-SD coverage.

We conclude this section by demonstrating the capability of APMC algorithms to perform unconditional image generation in Fig.~\ref{Fig:Uncond_Face}-Fig.~\ref{Fig:Uncond_BH}.
To do this, we simply remove the likelihood term and make APMC algorithms sample from the prior distribution.
Note that APMC-PnP and APMC-RED will share the same formulation (without $g$) and become a standard of SGM based on Langevin diffusion.
We use the hyperparameter setup in~\cite{Song.etal2020b} to configure the algorithm.
Note the wide diversity and high fidelity of the images sampled by the APMC algorithm.

\newpage

\begin{table}[!t]
	\centering
	\tiny
	\setstretch{1}
	\caption{List of algorithmic hyperparameters used in the CS recovery task ($m/n = 0.1$)}
	\begin{tabular*}{470pt}{L{45pt}C{5pt}C{120pt}C{54pt}C{54pt}C{54pt}C{54pt}} 			
		\toprule
		& \multicolumn{2}{c}{\textbf{Hyperparameters}} & \textbf{PMC-PnP/RED} & \textbf{APMC-PnP/RED} & \textbf{PnP-ULA-Ori} & \textbf{DPS} \\	 
		\cmidrule(lr){2-3} \cmidrule(lr){4-7}
		\multirow{4}{*}{\makecell{Common \\ parameters}} & $\xbm^{0}$ & initial point of reconstructions & \multicolumn{4}{c}{i.i.d uniformly drawn from $[\shortneg1,1]^n$} \\ 
		& $\gamma$  & step-size & $1e$-$5$ & $1e$-$5$ & $1e$-$5$ & - \\
		& $\alpha$  & coefficient of the prior score & $1$ & - & $1$ & - \\
		& $\sigma$  & strength of $\Scal_\theta(\xbm, \sigma)$ & $0.015$ & - & $5/255$ & - \\
		\noalign{\vskip 0.5ex}
		\hdashline\noalign{\vskip 0.5ex}
		\multirow{4}{*}{\makecell{Annealing \\ parameters}} & $\xi$ & exponential decay rate & - & $0.99$ & - & - \\
		& $\alpha_0$ & parameter for generating $\{\alphak\}$  & - & $4444$ & - & - \\
		& $\sigma_0$ & initial value of $\{\sigmak\}$ & - & $348$ & - & - \\
		& $\sigma_\mathsf{min}$ & minimal value of $\{\sigmak\}$ & - & $0.015$ & - & - \\
		\noalign{\vskip 0.5ex}
		\hdashline\noalign{\vskip 0.5ex}
		\multirow{2}{*}{DPS parameter} & \multirow{2}{*}{$q$} & the original $\gamma$ parameter in~\cite{Chung.etal2023diffusion} for adjusting the contribution of the likelihood & \multirow{2}{*}{-} & \multirow{2}{*}{-} & \multirow{2}{*}{-} & \multirow{2}{*}{$0.7$} \\
		\noalign{\vskip 0.5ex}
		\hdashline\noalign{\vskip 0.5ex}
		PnP-ULA parameter & \multirow{2}{*}{$\mathsf{C}$} &  \multirow{2}{*}{the range of the projection $\Pi_\mathsf{C}$} & \multirow{2}{*}{-} & \multirow{2}{*}{-} & \multirow{2}{*}{$[\shortneg1,2]^n$} & \multirow{2}{*}{-} \\		
		\bottomrule
	\end{tabular*}
	\label{Tab:CSParam1} 
\end{table}

\begin{table}[!t]
	\centering
	\tiny
	\setstretch{1}
	\caption{List of algorithmic hyperparameters used in the CS recovery task ($m/n = 0.3$)}
	\begin{tabular*}{470pt}{L{45pt}C{5pt}C{120pt}C{54pt}C{54pt}C{54pt}C{54pt}}			
		\toprule
		& \multicolumn{2}{c}{\textbf{Hyperparameters}} & \textbf{PMC-PnP/RED} & \textbf{APMC-PnP/RED} & \textbf{PnP-ULA-Ori} & \textbf{DPS} \\	 
		\cmidrule(lr){2-3} \cmidrule(lr){4-7}
		\multirow{4}{*}{\makecell{Common \\ parameters}} & $\xbm^{0}$ & initial point of reconstructions & \multicolumn{4}{c}{i.i.d uniformly drawn from $[\shortneg1,1]^n$} \\ 
		& $\gamma$  & step-size & $1e$-$5$ & $1e$-$5$ & $2e$-$5$ & - \\
		& $\alpha$  & coefficient of the prior score & $1$ & - & $1$ & - \\
		& $\sigma$  & strength of $\Scal_\theta(\xbm, \sigma)$ & $0.009$ & - & $5/255$ & - \\
		\noalign{\vskip 0.5ex}
		\hdashline\noalign{\vskip 0.5ex}
		\multirow{4}{*}{\makecell{Annealing \\ parameters}} & $\xi$ & exponential decay rate & - & $0.99$ & - & - \\
		& $\alpha_0$ & parameter for generating $\{\alphak\}$  & - & $12345$ & - & - \\
		& $\sigma_0$ & initial value of $\{\sigmak\}$ & - & $348$ & - & - \\
		& $\sigma_\mathsf{min}$ & minimal value of $\{\sigmak\}$ & - & $0.009$ & - & - \\
		\noalign{\vskip 0.5ex}
		\hdashline\noalign{\vskip 0.5ex}
		\multirow{2}{*}{DPS parameter} & \multirow{2}{*}{$q$} & the original $\gamma$ parameter in~\cite{Chung.etal2023diffusion} for adjusting the contribution of the likelihood & \multirow{2}{*}{-} & \multirow{2}{*}{-} & \multirow{2}{*}{-} & \multirow{2}{*}{$2.0$} \\
		\noalign{\vskip 0.5ex}
		\hdashline\noalign{\vskip 0.5ex}
		PnP-ULA parameter & \multirow{2}{*}{$\mathsf{C}$} &  \multirow{2}{*}{the range of the projection $\Pi_\mathsf{C}$} & \multirow{2}{*}{-} & \multirow{2}{*}{-} & \multirow{2}{*}{$[\shortneg1,2]^n$} & \multirow{2}{*}{-} \\	
		\bottomrule
	\end{tabular*}
	\label{Tab:CSParam2} 
\end{table}

\begin{table}[!t]
	\centering
	\tiny
	\setstretch{1}
	\caption{List of algorithmic hyperparameters used in the MRI recovery task ($\text{Accel.} = 8\times$)}
	\begin{tabular*}{470pt}{L{45pt}C{5pt}C{120pt}C{54pt}C{54pt}C{54pt}C{54pt}} 	 			
		\toprule
		& \multicolumn{2}{c}{\textbf{Hyperparameters}} & \textbf{PMC-PnP/RED} & \textbf{APMC-PnP/RED} & \textbf{PnP-ULA-Ori} & \textbf{DPS} \\	 
		\cmidrule(lr){2-3} \cmidrule(lr){4-7}
		\multirow{4}{*}{\makecell{Common \\ parameters}} & $\xbm^{0}$ & initial point of reconstructions & \multicolumn{4}{c}{i.i.d uniformly drawn from $[\shortneg1,1]^n$} \\ 
		& $\gamma$  & step-size & $5e$-$6$ & $5e$-$6$ & $5e$-$5$ & - \\
		& $\alpha$  & coefficient of the prior score & $1$ & - & $1$ & - \\
		& $\sigma$  & strength of $\Scal_\theta(\xbm, \sigma)$ & $0.01$ & - & $5/255$ & - \\
		\noalign{\vskip 0.5ex}
		\hdashline\noalign{\vskip 0.5ex}
		\multirow{4}{*}{\makecell{Annealing \\ parameters}} & $\xi$ & exponential decay rate & - & $0.99$ & - & - \\
		& $\alpha_0$ & parameter for generating $\{\alphak\}$  & - & $10000$ & - & - \\
		& $\sigma_0$ & initial value of $\{\sigmak\}$ & - & $348$ & - & - \\
		& $\sigma_\mathsf{min}$ & minimal value of $\{\sigmak\}$ & - & $0.01$ & - & - \\
		\noalign{\vskip 0.5ex}
		\hdashline\noalign{\vskip 0.5ex}
		\multirow{2}{*}{DPS parameter} & \multirow{2}{*}{$q$} & the original $\gamma$ parameter in~\cite{Chung.etal2023diffusion} for adjusting the contribution of the likelihood & \multirow{2}{*}{-} & \multirow{2}{*}{-} & \multirow{2}{*}{-} & \multirow{2}{*}{$4.0$} \\
		\noalign{\vskip 0.5ex}
		\hdashline\noalign{\vskip 0.5ex}
		PnP-ULA parameter & \multirow{2}{*}{$\mathsf{C}$} &  \multirow{2}{*}{the range of the projection $\Pi_\mathsf{C}$} & \multirow{2}{*}{-} & \multirow{2}{*}{-} & \multirow{2}{*}{$[\shortneg1,2]^n$} & \multirow{2}{*}{-} \\	
		\bottomrule
	\end{tabular*}
	\label{Tab:MRIParam1}
\end{table}

\begin{table}[!t]
	\centering
	\tiny
	\setstretch{1}
	\caption{List of algorithmic hyperparameters used in the MRI recovery task ($\text{Accel.} = 4\times$)}
	\begin{tabular*}{470pt}{L{45pt}C{5pt}C{120pt}C{54pt}C{54pt}C{54pt}C{54pt}} 			
		\toprule
		& \multicolumn{2}{c}{\textbf{Hyperparameters}} & \textbf{PMC-PnP/RED} & \textbf{APMC-PnP/RED} & \textbf{PnP-ULA-Ori} & \textbf{DPS} \\	 
		\cmidrule(lr){2-3} \cmidrule(lr){4-7}
		\multirow{4}{*}{\makecell{Common \\ parameters}} & $\xbm^{0}$ & initial point of reconstructions & \multicolumn{4}{c}{i.i.d uniformly drawn from $[\shortneg1,1]^n$} \\ 
		& $\gamma$  & step-size & $5e$-$6$ & $5e$-$6$ & $5e$-$5$ & - \\
		& $\alpha$  & coefficient of the prior score & $1$ & - & $1$ & - \\
		& $\sigma$  & strength of $\Scal_\theta(\xbm, \sigma)$ & $0.01$ & - & $5/255$ & - \\
		\noalign{\vskip 0.5ex}
		\hdashline\noalign{\vskip 0.5ex}
		\multirow{4}{*}{\makecell{Annealing \\ parameters}} & $\xi$ & exponential decay rate & - & $0.99$ & - & - \\
		& $\alpha_0$ & parameter for generating $\{\alphak\}$  & - & $10000$ & - & - \\
		& $\sigma_0$ & initial value of $\{\sigmak\}$ & - & $348$ & - & - \\
		& $\sigma_\mathsf{min}$ & minimal value of $\{\sigmak\}$ & - & $0.01$ & - & - \\
		\noalign{\vskip 0.5ex}
		\hdashline\noalign{\vskip 0.5ex}
		\multirow{2}{*}{DPS parameter} & \multirow{2}{*}{$q$} & the original $\gamma$ parameter in~\cite{Chung.etal2023diffusion} for adjusting the contribution of the likelihood & \multirow{2}{*}{-} & \multirow{2}{*}{-} & \multirow{2}{*}{-} & \multirow{2}{*}{$6.0$} \\
		\noalign{\vskip 0.5ex}
		\hdashline\noalign{\vskip 0.5ex}
		PnP-ULA parameter & \multirow{2}{*}{$\mathsf{C}$} &  \multirow{2}{*}{the range of the projection $\Pi_\mathsf{C}$} & \multirow{2}{*}{-} & \multirow{2}{*}{-} & \multirow{2}{*}{$[\shortneg1,2]^n$} & \multirow{2}{*}{-} \\	
		\bottomrule
	\end{tabular*}
	\label{Tab:MRIParam2}
\end{table}

\begin{figure*}[!t]
\centering
\includegraphics[width=0.9\textwidth]{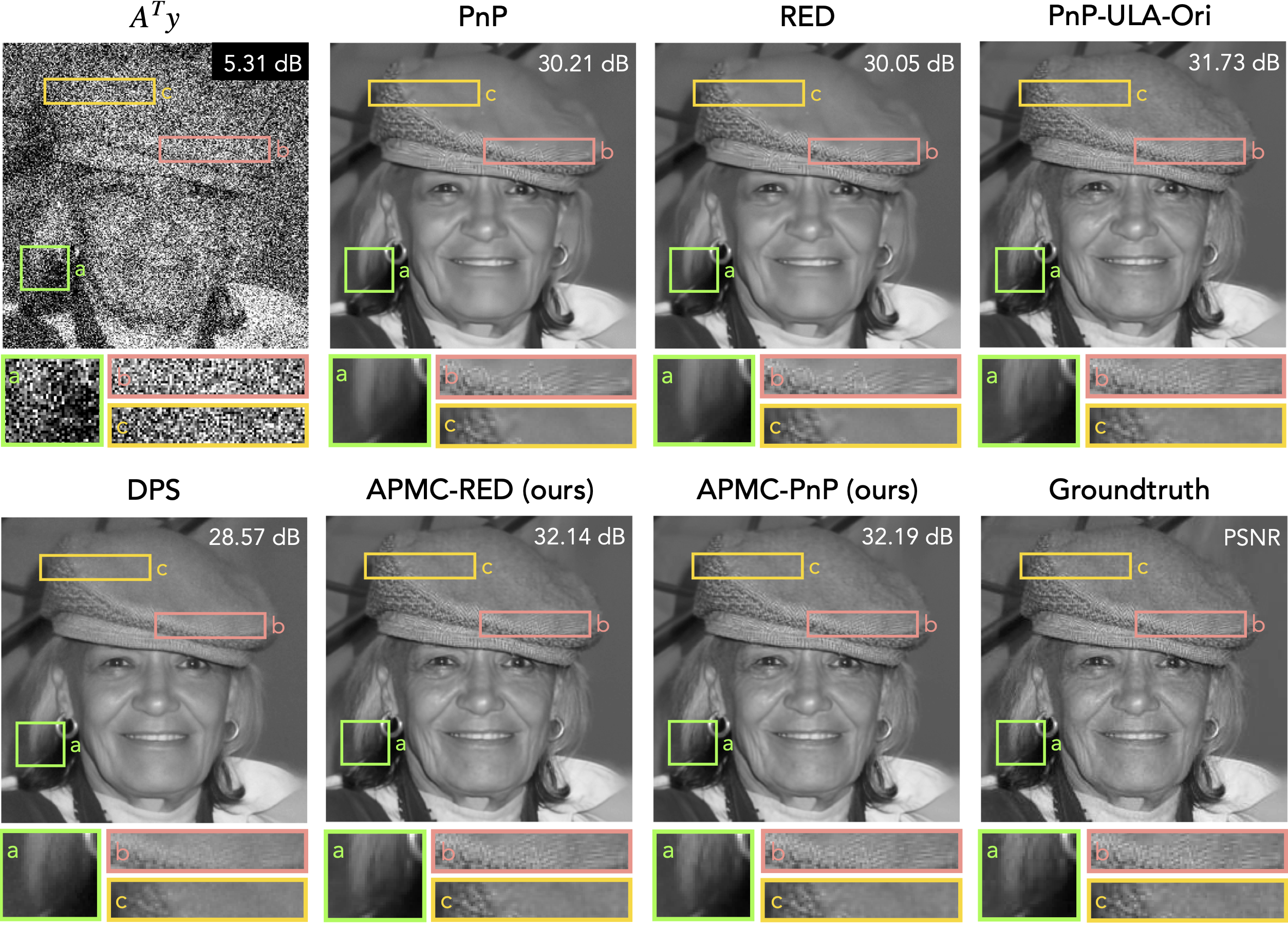}
\caption{
Visual comparison of the reconstructions obtained by the APMC and baseline algorithms for $30\%$ CS recovery.
We include the simple adjoint reconstruction ($\Abm^T\ybm$) to illustrate the ill-posedness of this task.
The final images of the sampling algorithms (PnP-ULA-Ori, DPS, APMC-RED, and APMC-PnP) are obtained by averaging $50$ image samples.
The visual difference is highlighted in the zoom-in images.
Note how APMC algorithms restore the fine textures of the hair and hat.
While PnP-ULA-Ori yields visual quality on par with APMC algorithms, but the latter still numerically outperform the former by about $0.4$ dB.
}
\label{Fig:CS_MMSE}
\end{figure*}

\begin{figure*}[t!]
\centering
\includegraphics[width=0.9\textwidth]{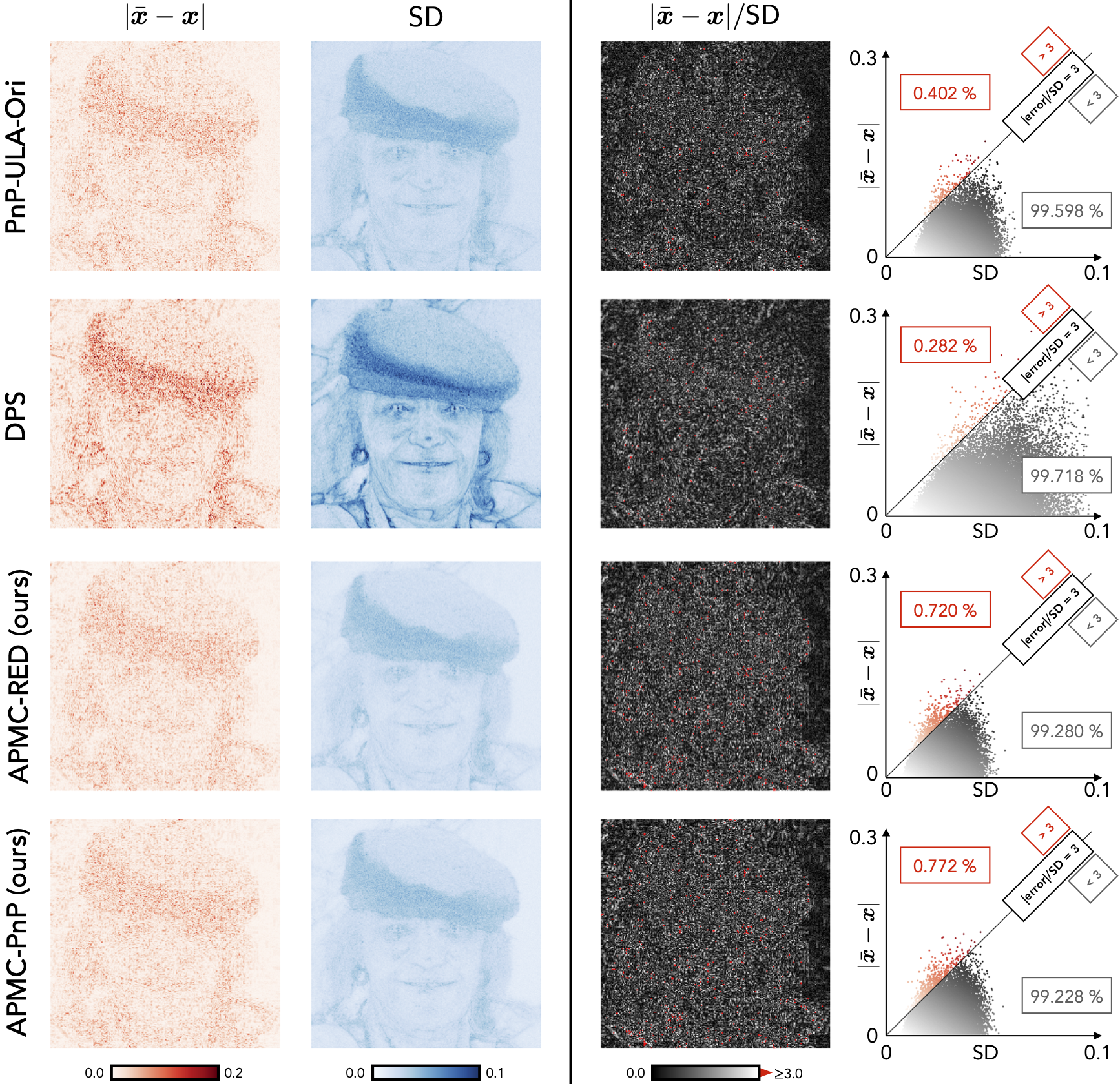}
\caption{
Visualization of the pixel-wise statistics associated with the reconstructions shown in Fig.~\ref{Fig:CS_MMSE}.
The left columns plot the absolute error ($|\xbmbar-\xbm|$) and standard deviation (SD), and the right columns plot the 3-SD credible interval with the outlying pixels highlighted in red.
Note that APMC algorithms lead to a better UQ performance than the baselines by recovering an accurate mean and avoiding large SD.
}
\label{Fig:CS_UQ}
\end{figure*}

\begin{figure*}[t!]
\centering
\includegraphics[width=0.9\textwidth]{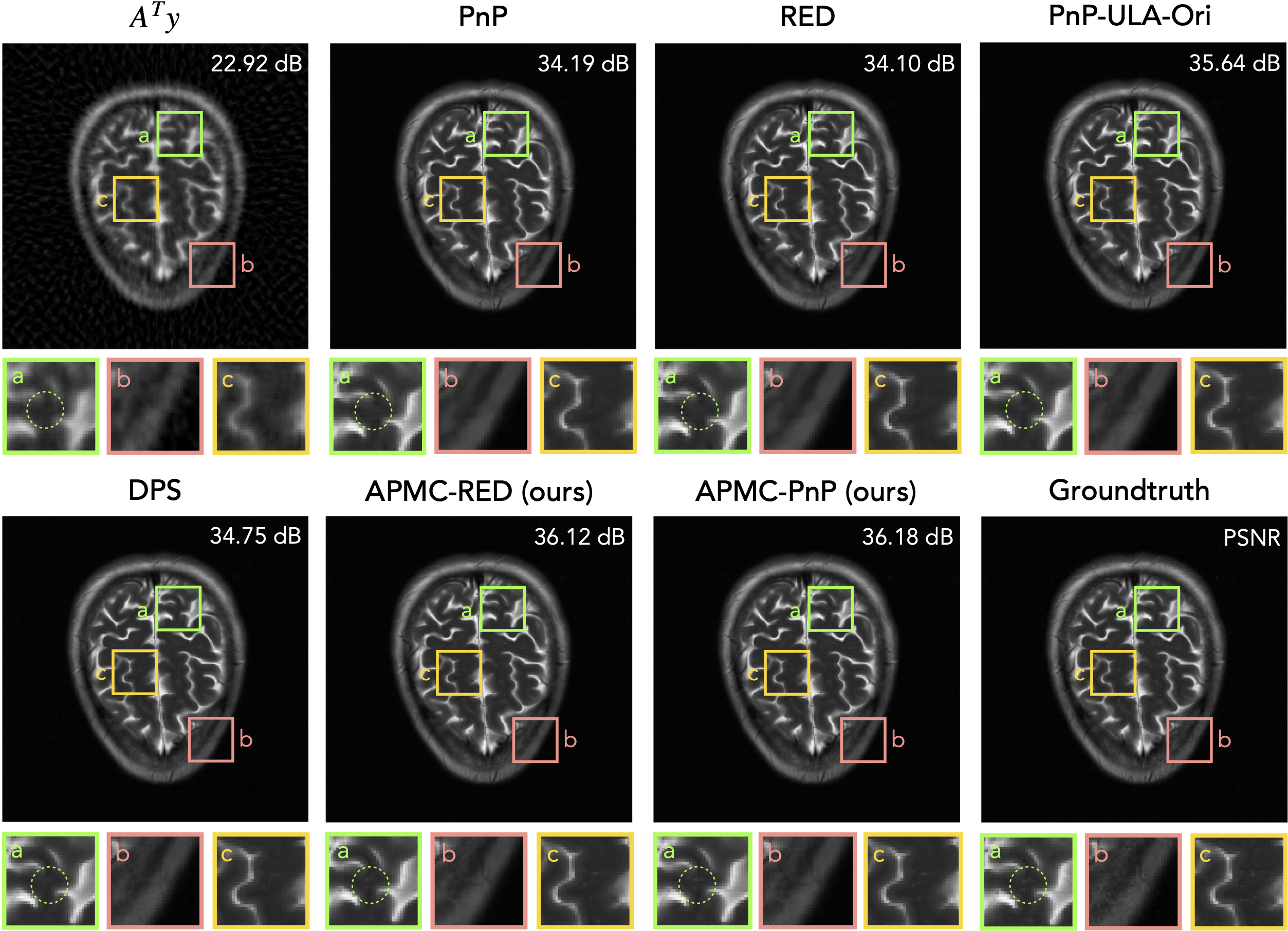}
\caption{
Visual comparison of the reconstructions obtained by the APMC and baseline algorithms for $4\times$ MRI recovery.
We include the simple adjoint reconstruction ($\Abm^T\ybm$) to illustrate the ill-posedness of this task.
The final images of the sampling algorithms (PnP-ULA, DPS, APMC-RED, and APMC-PnP) are obtained by averaging $50$ image samples.
The visual difference is shown in the zoom-in images.
We note that all algorithms yield a relatively high reconstruction quality due to the moderate ill-posedness of this problem.
Note that APMC algorithms still numerically outperform PnP-ULA-Ori.
}
\label{Fig:MRI_MMSE}
\end{figure*}

\begin{figure*}[t!]
\centering
\includegraphics[width=0.9\textwidth]{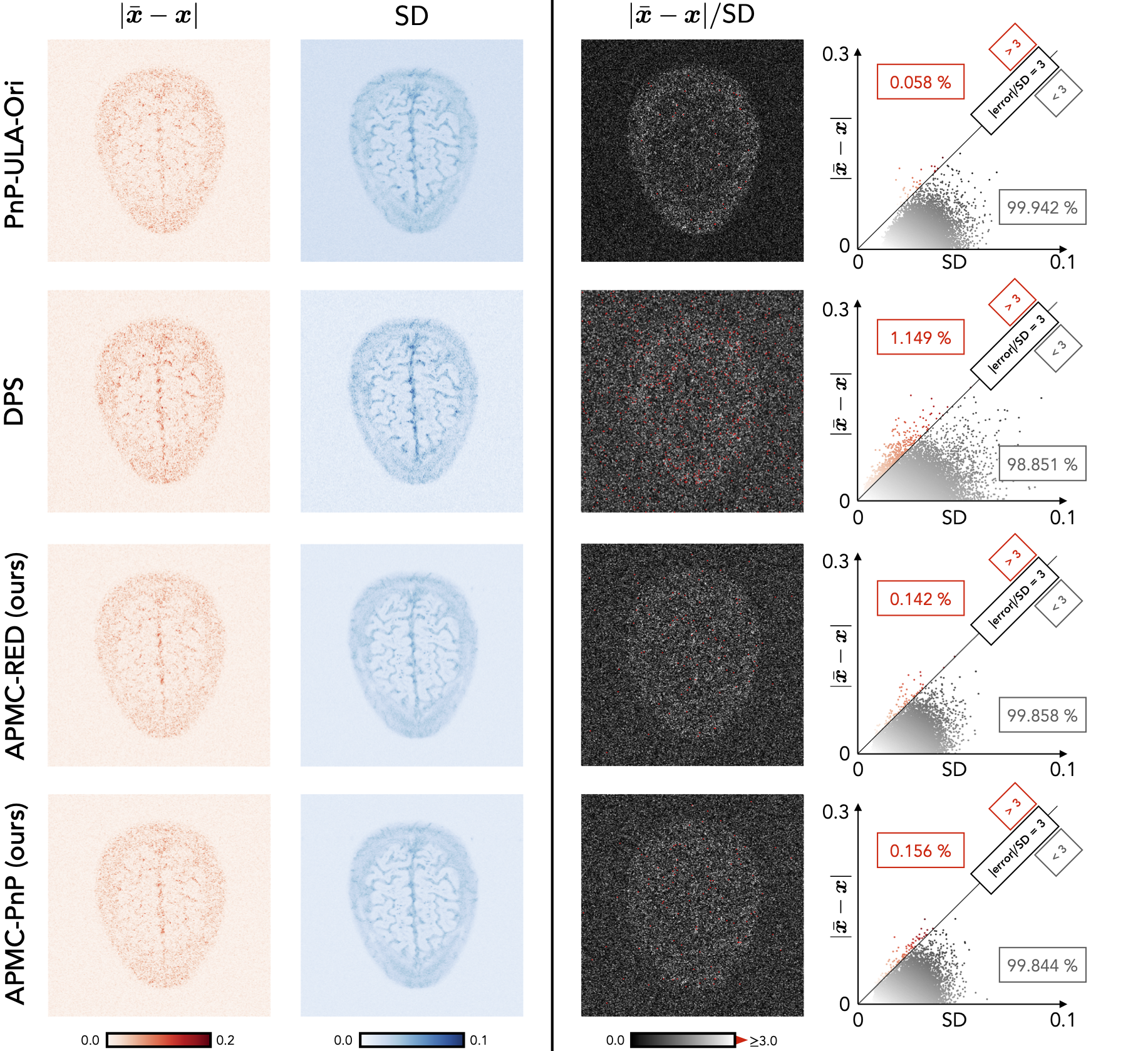}
\caption{
Visualization of the pixel-wise statistics associated with the reconstructions shown in Fig.~\ref{Fig:MRI_MMSE}.
The left columns plot the absolute error ($|\xbmbar-\xbm|$) and standard deviation (SD), and the right columns plot the 3-SD credible interval with the outside pixels highlighted in red.
While all algorithms achieve similar UQ performance due to the mild ill-posedness, we observe APMC algorithms obtain slightly smaller SD while still maintaining 99\% 3-SD coverage. 
}
\label{Fig:MRI_UQ}
\end{figure*}

\begin{figure*}[t!]
\centering
\includegraphics[width=0.9\textwidth]{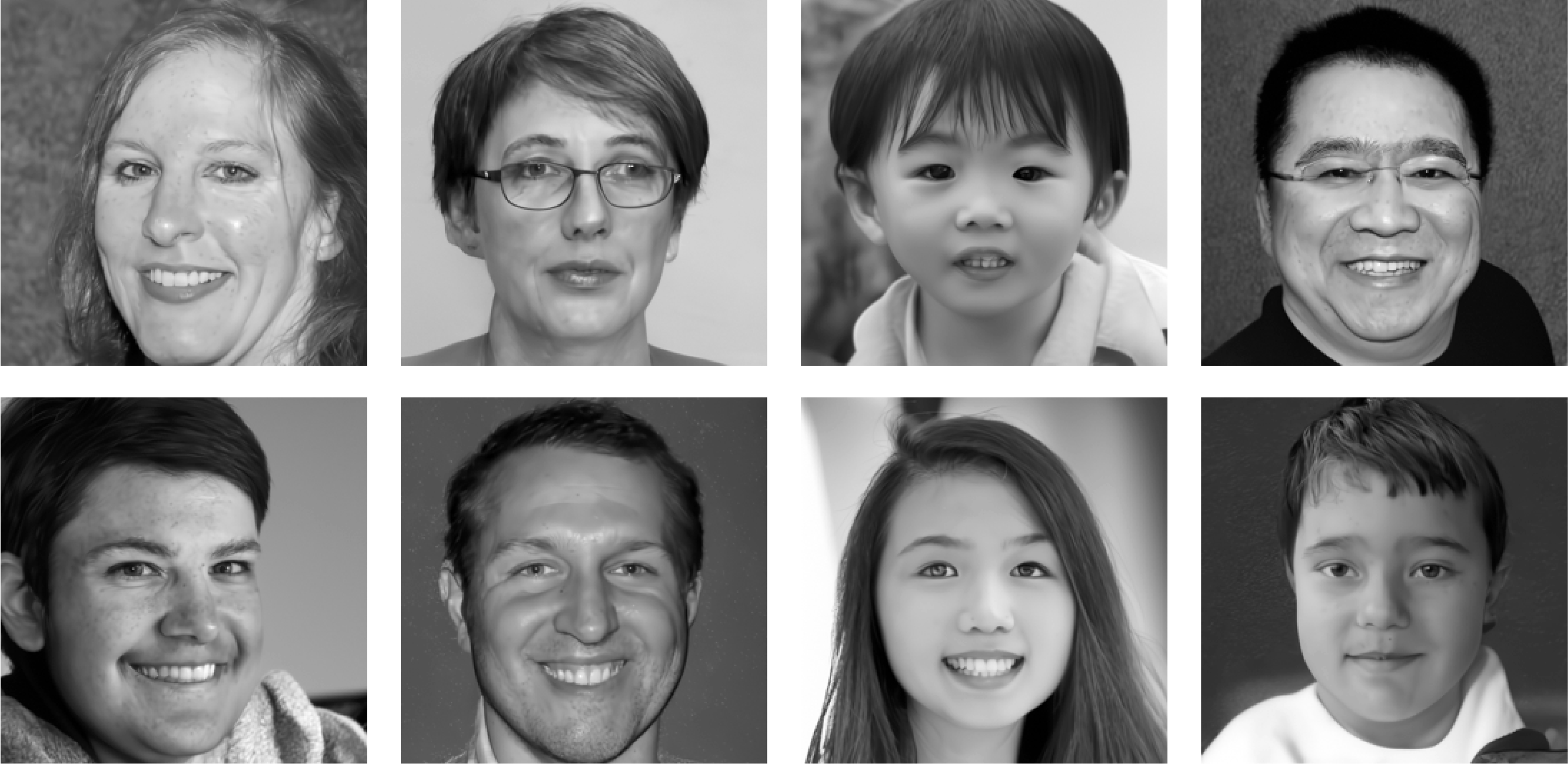}
\caption{
Visualization of the capability of APMC to perform unconditional generation of human face images.
The images are generated using the score network trained on the FFHQ dataset. Note the APMC can generate images with different gender, age, ethnicity, and background.
}
\label{Fig:Uncond_Face}
\end{figure*}

\begin{figure*}[t!]
\centering
\includegraphics[width=0.9\textwidth]{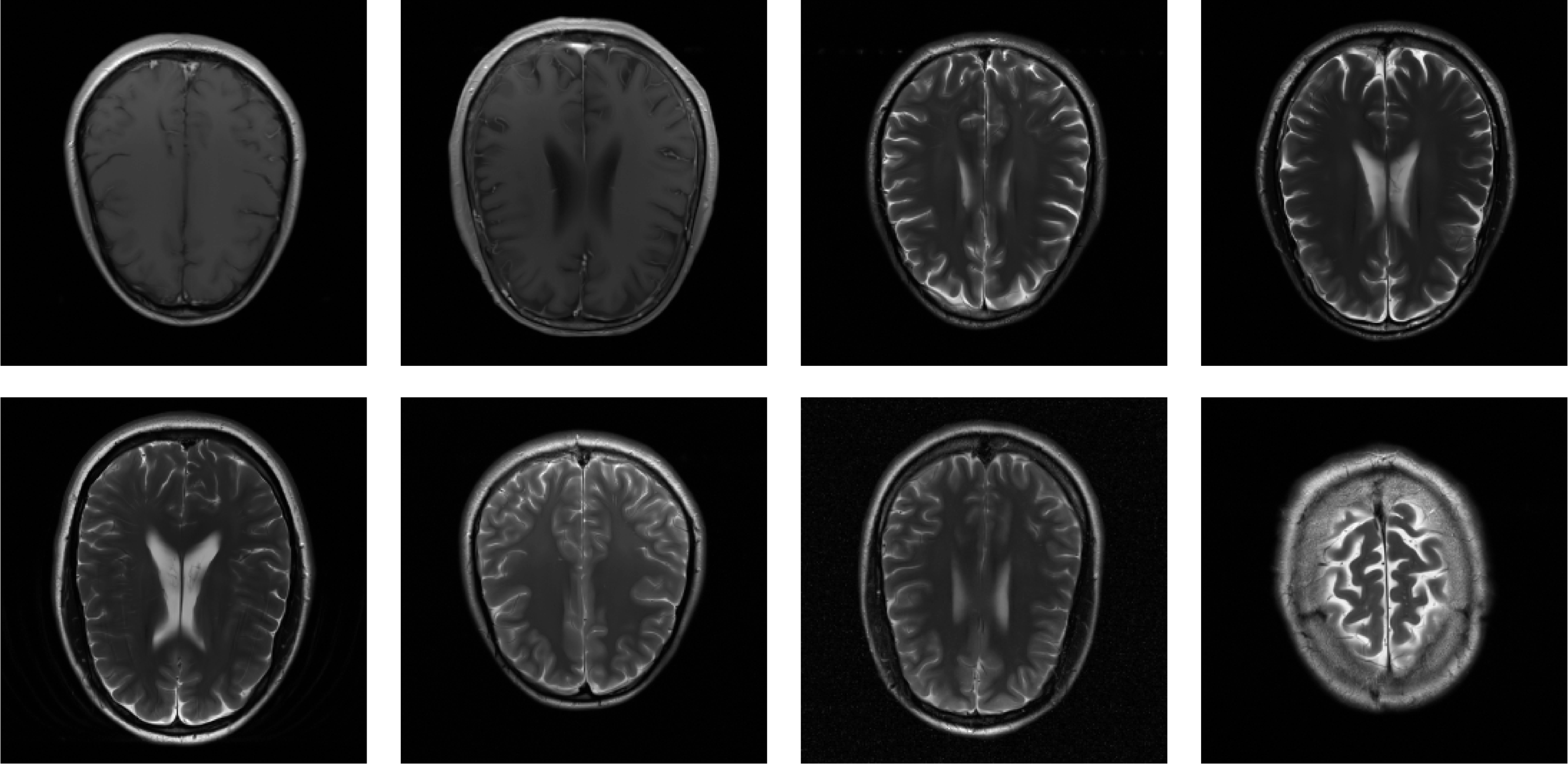}
\caption{
Visualization of the capability of APMC to perform unconditional generation of MRI images.
The images are generated using the score network trained on the FastMRI brain dataset, which images of different contrasts (T1, T2, and FLAIR).
Note that APMC can successfully restore such technical variety with high fidelity.
}
\label{Fig:Uncond_MRI}
\end{figure*}

\begin{figure*}[t!]
\centering
\includegraphics[width=0.9\textwidth]{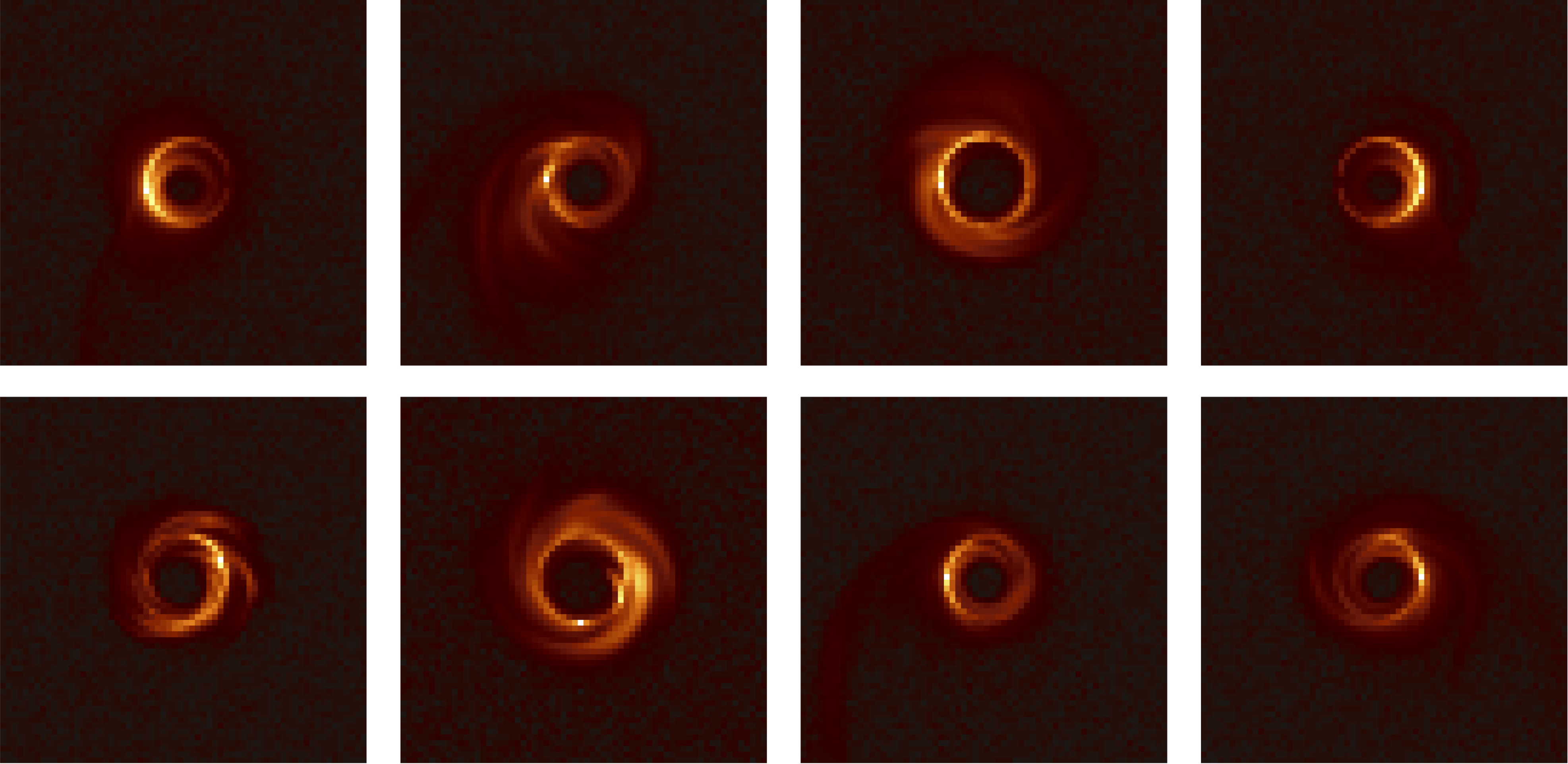}
\caption{
Visualization of the capability of APMC to perform unconditional generation of black hole images.
The images are generated using the score network trained on the GRMHD dataset with data augmentation of flipping and resizing of the black hole.
Note that APMC can generate images with different flux rotations and diameters of the event horizon.
}
\label{Fig:Uncond_BH}
\end{figure*}

\end{document}